\documentclass[11pt]{amsart}

\usepackage[utf8]{inputenc}
\usepackage{amssymb,amsmath,amscd,slashed,mathrsfs}
\usepackage{graphicx,hyperref,textcomp,eucal}
\newtheorem{Theorem}{Theorem}[section]          
\newtheorem{Proposition}[Theorem]{Proposition}   \newtheorem{Corollary}[Theorem]{Corollary}      
\newtheorem{Lemma}[Theorem]{Lemma}               \newtheorem{Definition}[Theorem]{Definition}

\newcommand{\A}{\mathscr{A}}        \newcommand{\Bo}{\mathcal{B}}
\newcommand{\C}{\mathbb{C}}      \newcommand{\cont}{\mbox{cont}}
\newcommand{\der}{\mathcal{D}}   \newcommand{\De}{\mathscr{D}}

\newcommand{\G}{\mathbb{G}}      \newcommand{\Ge}{\mathscr{G}}    \newcommand{\g}{\mathfrak{g}}
\newcommand{\Ho}{\mathcal{H}}    \newcommand{\Int}{\mbox{Int}}    \newcommand{\J}{\mathfrak{I}}

\newcommand{\La}{\mathcal{L}}    \newcommand{\LB}{\mathcal{S}}        
\newcommand{\M}{\mathcal{M}}     \newcommand{\magn}{\mbox{mag}}    
\newcommand{\Mo}{\mathfrak{M}}   \newcommand{\mo}{\mathfrak{m}}       \newcommand{\N}{\mathbb{N}}	 
\newcommand{\Nei}{\mathcal{N}}   \newcommand{\og}{\overline{\gamma}}  \newcommand{\Pe}{\mathfrak{P}}

\newcommand{\Qa}{\mathcal{Q}}            \newcommand{\R}{\mathbb{R}}         \newcommand{\Rn}{\mathcal{R}}     
\newcommand{\U}{\mathcal{U}}             \newcommand{\supp}{\mbox{supp}}     \newcommand{\Sum}{\mbox{$\sum$}}
\newcommand{\suppi}{\mbox{\emph{supp}}}  \newcommand{\Tu}{\mathcal{T}}       \newcommand{\Ts}{\mathbb{T}}  
\newcommand{\Tm}[1]{\mbox{$#1$\textlbrackdbl$\Mo$\textrbrackdbl}}
\newcommand{\Tlm}[2]{\mbox{$#1$\textlbrackdbl$#2$\textrbrackdbl}}

\newcommand{\W}{\mathcal{W}}   \newcommand{\ti}[1]{\tilde{#1}}   \newcommand{\wt}[1]{\widetilde{#1}} \newcommand{\wg}{\widetilde{\gamma}}  
\newcommand{\X}{\mathcal{X}}   \newcommand{\Z}{\mathbb{Z}}
\newcommand{\hs}[1]{\hspace{#1cm} }
\newcommand{\graph}[5]{\hspace{#1cm} \raisebox{#2cm}{\includegraphics[scale=#3]{#4.pdf}} \hspace{#5cm} }

\numberwithin{equation}{section}

\begin{document}

\title{Resurgent transseries \& Dyson-Schwinger equations}

\author{Lutz Klaczynski}  

\address{Department of Physics, Humboldt University Berlin, 12489 Berlin, Germany}
\email{klacz@mathematik.hu-berlin.de}
\date{\today}

\begin{abstract}
We employ resurgent transseries as algebraic tools to investigate two self-consistent Dyson-Schwinger equations, one in Yukawa theory and one in quantum 
electrodynamics. After a brief but pedagogical review, we derive fixed point equations for the associated anomalous dimensions and insert a 
moderately generic log-free transseries ansatz to study the possible strictures imposed. While proceeding in various stages, we develop 
an algebraic method to keep track of the transseries' coefficients.  
We explore what conditions must be violated in order to stay clear of fixed point theorems to eschew a unique solution, if so desired, as we explain. 
An interesting finding is that the flow of data between the different sectors of the transseries shows a pattern typical of resurgence, ie 
the phenomenon that the perturbative sector of the transseries talks to the nonperturbative ones in a one-way fashion. 
However, our ansatz is not exotic enough as it leads to trivial solutions with vanishing nonperturbative sectors, even when logarithmic monomials are included. 
We see our result as a harbinger of what future work might reveal about the transseries representations of observables 
in fully renormalised four-dimensional quantum field theories and adduce a tentative yet to our mind weighty argument as to why one should not expect 
otherwise.      

This paper is considerably self-contained. Readers with little prior knowledge are let in on the basic reasons why perturbative series in quantum field 
theory eventually require an upgrade to transseries. Furthermore, in order to acquaint the reader with the language utilised extensively in this work, 
we also provide a concise mathematical introduction to grid-based transseries. 

\end{abstract}

\maketitle

\tableofcontents

\newpage

\section{Introduction}\label{sec:Intro}                                                   
\subsection{Divergent series}\label{divergs}
It is surely a strength of a method if it indicates its own weaknesses and limitations. Luckily, such is the case for perturbation theory in
quantum field theory (QFT): the growth in the number of Feynman diagrams with loop order makes it seem next to impossible for the perturbative expansions 
to have other than zero radius of convergence. For example, let 
\begin{equation}\label{pertex}
 \ti{F}(\alpha, Q ) = \sum_{n \geq 0} a_{n}(Q) \alpha^{n}
\end{equation}
be a renormalised perturbative expansion of an observable quantity $F(\alpha,Q)$ in a QFT with a single coupling parameter $\alpha$ and external 
kinematical data signified by $Q$ (momenta and scales). Then typically, for fixed $Q$, the coefficients behave asymptotically as 
\begin{equation}\label{divergi}
 a_{n} \sim A^{n} n^{b} n! \ , 
\end{equation}
where $A,b \in \R$ are constants \cite{GuZi90,Sti02}. And there is yet another known source of such growth behaviour leading to the divergence 
of (\ref{pertex}), namely so-called \emph{renormalons}, caused by factorially growing coefficients of subseries due to the 
integration over high or low momenta in certain renormalised Feynman integrals (aptly called UV or IR renormalons, respectively). 
In other words, it is in this case the value of the Feynman integrals themselves rather than their abundance which brings about this 
phenomenon \cite{Ben99}. Because our models exhibit renormalons as well, we shall say some more on them in the main text.  

However, the issue with such series is that 
one can never be sure how good an approximation to the actual observable their truncation at a given positive coupling $\alpha$ really is, especially
in situations where one is interested in certain kinematical regimes like low energies in quantum chromodynamics or very high energies in quantum 
electrodynamics, because there, in these more or less extreme regimes, the coupling needs to assume large values. 

This is the problem of \emph{optimal truncation} and its severity can easily be understood: factorials will always eventually win out over powers of 
the coupling so that, for a fixed coupling, one can drive the value of the truncated series astray as much as one desires by including more and more 
terms until the result has absolutely nothing to do with the actual function (see \cite{Mar14} for a nice illustration of optimal truncation). 
This basic truth is already sufficient to understand that perturbative computations \emph{alone} can never capture the physics of (renormalisable) QFTs 
completely, no matter how far they reach in terms of loop order.  

The fact that perturbation theory generally leads to divergent power series was actually noticed and studied quite early in the 1950s by 
Hurst \cite{Hu52a,Hu52b} and Petermann \cite{Pe53a, Pe53b}, while Dyson came up with a physical interpretation\footnote{The argument was vehemently opposed by 
Simon in \cite{Si69}, p.201: ''It cannot be emphasized too greatly that this argument must be 
considered unacceptable for the problem at hand ... the uncritical use of physical arguments will lead one astray ...''} in the case of quantum 
electrodynamics in \cite{Dys51}. Further systematic studies for scalar theories were then conducted by Thirring \cite{Th53} 
and Lipatov \cite{Li77}. Beyond growth estimates, Jaffe proved that in two spacetime dimensions, perturbation theory must produce divergent series 
in superrenormalisable scalar theories with nonderivative boson self-interactions, including $(\varphi^{4})_{2}$ theory \cite{Ja65} (see also the review 
\cite{Hu06} by Hurst). Furthermore, we mention the work of Lautrup \cite{Lau77}, Itzykson, Parisi and Zuber \cite{IPZu77} who obtained estimates of the 
type (\ref{divergi}) for (scalar) quantum electrodynamics which extended previous work by Bender and Wu on the anharmonic oscillator in 
quantum mechanics \cite{BeWu71}. 

Note that this so-called 'divergence of perturbation theory' is not directly related to (UV) renormalisation: even theories without any need to 
be renormalised like the genus expansion in string theory or examples from quantum mechanics exhibit this phenomenon \cite{CESVo16,GraMaZ15,Zi81}. 
What these examples teach us at the very least is that renormalisation is not a necessary condition for the divergence of perturbation theory. 

But renormalisation does indeed play an important role for the divergence of perturbative series in quantum field theory. 
Simon studied in \cite{Si69} the renormalised perturbative expansions of the fermion's disconnected Green's functions in 2-dimensional Yukawa theory 
($Y_{2}$) and found that they had a finite radius of convergence at 
least as long as they were regularised by cutoffs in both momentum and position space. While at the time he could not tell what would happen upon 
removal of those cutoffs, another author, Parisi, reported later that according to his estimates, these series had a finite radius of convergence due to ''strong 
cancellations ... among Feynman diagrams with different topologies'', which he put down to Pauli's exclusion principle \cite{Pa77}. 
For the fermions in $Y_{2}$ theory, this implies that the abundance of Feynman diagrams is effectively switched off by these cancellations. 
Not so for the bosonic sector, where he found 
\begin{equation}
 a_{n} \sim A^{n}  n^{b} \Gamma(n(d-2)/d) \cos(2\pi n/d) 
\end{equation}
for spacetime dimension $d>2$ and $a_{n} \sim (\log n)^{n}$ in case $d=2$ (in our notation). 

The impact of renormalisation can best be seen in the \emph{Gross-Neveu model}, a renormalisable Fermi theory in two spacetime 
dimensions \cite{GroN74}. Similar to $Y_{2}$ theory, it shows perturbative 
series with a nonzero radius of convergence as long as a UV cutoff is in place. But here, the model acquires renormalons in the UV limit leading
to an asymptotic behaviour of the coefficients of type $\sim n!$ \cite{FeMaRiS85,FeMaRiS86}. 

Hence, as the UV cutoff is increased in this model, an infinite number of coefficients of the renormalised series grow in magnitude in
such a way that in the limit, the series are forced to pass the threshold from Gevrey 0 to Gevrey 1 (for a definition, see §\ref{subsec:Bor}). 

Renormalisation has in this case indeed a severe impact and, as we believe, even more so in four-dimensional renormalisable theories, where it for 
one thing \emph{introduces a nontrivial coupling dependence} into the theory's Lagrangian and hence to its correlation functions. 
For another, it is very likely to \emph{affect the exponential size of the Borel transforms}, as explained in §\ref{sec:RenorTra}.

However, these ideas unavoidably come with some measure of speculation, and the author is very well aware of
it, especially when considering the nonperturbative status of quantum electrodynamics (QED), Yukawa theory and renormalisable QFTs in general. The 
issue is that in order to incorporate renormalisation into the narrative of fully quantised Lagrangian 
field theories and their equations of motion, one needs to include the so-called Z factors whose nonperturbative existence can to this day at best only be
assumed and is far from being well-defined \cite{Ost86}. Moreover, given the fact that QED can only be an effective field theory and the 
possible existence of a Landau pole, it takes a fair amount of optimism to believe that renormalised 
QED exists as a mathematically consistent theory. Of course, being a QED-like toy model, the same goes for Yukawa theory as well. 

\subsection{Analysable functions \& resurgent transseries}\label{AnalyF}
The puzzle posed by divergent series and the question as to what to make of it had been tackled by our forebears more than a century ago. One answer 
of interest to us here is a procedure named \emph{Borel summation} which roughly works as follows \cite{Ha49, BeO91}: 
first, the divergent series is turned into a convergent one. This convergent series is then, in a second step, 
subjected to an integration procedure that essentially reverts the change brought about by the first step.
In benign cases, this resummation scheme, say implemented by an operator $\LB$, produces a function which has an asymptotic expansion that coincides with the original series. For example, if we assume of the series $\ti{F}(\alpha)$ in (\ref{pertex}) to fall under this rubric of well-behavedness, then, 
suppressing the kinematical data in the notation, the Borel summation operator $\LB$ produces a function 
$\LB [\ti{F}](\alpha)$ to which the series $\ti{F}$ is asymptotic \cite{Co09}, ie in our case
\begin{equation}
 \left|\LB[\ti{F}](\alpha) - \Sum_{n=0}^{N-1}a_{n}\alpha^{n} \right| \sim A^{N} N^{b} N! \alpha^{N} 
 \hs{1} \text{as \ $\alpha \rightarrow 0$ \ for all $N \geq 1$},
\end{equation}
where in even more benign circumstances, the operator $\LB$ may produce the actual function sought after and we have $\LB[\ti{F}](\alpha) = F(\alpha)$.

In order to handle less friendly cases, this summation method has since been further developed by \'Ecalle to the powerful machinery of 
\emph{accelero-summation}, designed to be applied to what he called \emph{analysable functions} \cite{Eca81,Eca93}. 

To get a flavour of what these functions are, we consider the three elementary functions $z, \exp z$ and $\log z$. If we take these together with 
constants in $\C$ and allow for all field operations (addition, multiplication, algebraic inversion) and composition including functional inversion
when possible, then it is clear that already some highly singular players will have entered the game. If we furthermore demand our set of functions be stable under 
differentiation and integration, then, for instance,
\begin{equation}\label{anal}
\int z^{-1} e^{z}  , \hs{1} \int \frac{1}{\log z} 
\end{equation}
must be added as primitives because they can neither be obtained from consecutive application of the field operations nor by functional composition. The 
resulting 'field with no escape' \cite{vH06} is that of analysable functions \cite{Co09}. Both functions in (\ref{anal}) can be expanded into series
\begin{equation}\label{transen}
\int z^{-1} e^{z} = e^{z} \sum_{k \geq 0}k! z^{-k-1} , 
\hs{1} \int \frac{1}{\log z} = z \sum_{k \geq 0}\frac{k!}{(\log z)^{k+1}},
\end{equation}
formally obtained by an infinite number of partial integrations. Both are perfect examples of \emph{resurgent transseries} whose basic building blocks 
are referred to as \emph{(resurgent) transmonomials}. Although often failing to converge, such series nonetheless carry asymptotic information \cite{Co09}. 

As already mentioned, accelero-summation is a tool designed to establish a connection between the world of divergent transseries and analysable functions.
The most commonly used name for the corresponding theory is \emph{resurgence theory}, referring to the idea that these functions 'resurge' or 'resurrect' from 
their perturbative series \cite{Sa14}.

These developments in mathematics have been parallelled by considerable progress on the physics side, the big question being what subclass of analysable
functions the observables of (quantum) physics fall into. On the transseries side, this boils down to seeking the necessary transmonomials and the construction 
of the corresponding transseries needed to capture the whole nonperturbative physical picture. Of particular interest is, moreover, their physical meaning.

Lately, pertinent results have been obtained in quantum mechanics in connection with so-called instantons \cite{Zi81, Zi02, ZiJ04} and, most recently, 
from promising investigations concerning a nonlinear sigma model, albeit of toy character \cite{DunU12, DunU13}. 
Adding to it are results for the free energy in minimal (super)string theory \cite{ASVo12,SchiVa14} and topological string 
theory \cite{CESVo15,CESVo16}, for which transseries ans\"atze proved viable. 
Furthermore, we mention the results for the transseries representations of the cusp anomalous dimension in $\mathcal{N}=4$ SUSY Yang-Mills theory which have 
been explored in \cite{An15,DoH15}. 

We glean from these latest developments that the current approaches towards a complete (nonperturbative) characterisation 
aim at upgrading the formal series $\ti{F}(\alpha)$ in (\ref{pertex}) to a resurgent transseries of the form 
\begin{equation}\label{pertexT}
\wt{F}(\alpha) = \sum_{\sigma \in \N_{0}^{r}} \alpha^{c \cdot \sigma} e^{-(b \cdot \sigma )/\alpha} P_{\sigma}(\log \alpha) \sum_{n \geq 0} a_{(\sigma,n)} \alpha^{n}
\end{equation}
where $b,c \in \C^{r}$ with Re$(b) \neq 0$ are fixed parameters and the symbol $P_{\sigma}(\log \alpha) \in \C[\log \alpha]$ is a 
polynomial in logs, one for each 'sector' $\sigma \in \N^{r}_{0}$. 
The series (\ref{pertexT}) is a so-called $r$-parameter transseries \cite{CESVo15}. 
For convenience, we will work with $r=1$, that is, a 'one-parameter' transseries\footnote{Our transseries ansatz will have more than one parameter, this denomination
is therefore unfortunate.} which will not affect the results of our work. 

Note that the original series $\ti{F}(\alpha)$ is still part of the expression (\ref{pertexT}) as its zeroth sector ($\sigma=0$) with coefficients $a_{(0,n)}=a_{n}$ and 
a trivial log polynomial $P_{0}=1$. The replacement 
\begin{equation}
 \ti{F}(\alpha) \in \C[[\alpha]] \hs{0.5} \longrightarrow \hs{0.5}
 \wt{F}(\alpha) \in \bigoplus_{\sigma \in \N_{0}^{r}} \alpha^{c \cdot \sigma} e^{-(b \cdot \sigma )/\alpha} \C[[\alpha]][\log \alpha]
\end{equation}
is often referred to as \emph{nonperturbative completion} of the formal series\footnote{Note the size of the hovering tilde!}. 
The next step is then to accelero-sum the transseries $\wt{F}(\alpha)$ sector-wise, which leads to a convergent series of the form 
\begin{equation}\label{pertexTa}
\LB[\wt{F}](\alpha) = \sum_{\sigma \in \N_{0}^{r}} \alpha^{c \cdot \sigma} e^{-(b \cdot \sigma )/\alpha} 
\LB \left[ P_{\sigma}(\log \alpha) \Sum_{n \geq 0} a_{(\sigma,n)} \alpha^{n} \right],
\end{equation}
where the summation operator $\LB$ has now turned each sector's perturbation series together with its log polynomials into a function. 
The idea here is that depending on the choice of the operator $\LB$, this resulting expression converges somewhere in the complex plane and thereby 
defines an analysable function which is asymptotic to the sought-after function $F(\alpha)$. However, we will not 
employ accelero-summation in our work. In fact, transseries of the type (\ref{pertexT}) are perfectly suited for a special instance of accelero-summation 
called \emph{\'Ecalle-Borel (BE) summation}. Although we continue this discussion to some extent in §\ref{sec:NPC} to give some more background on 
resurgence and this type of accelero-summation, we have to refer the interested reader at this point to the literature. 
BE summation is explained in many places. We recommend \cite{Sa07,Sa14} and \cite{Co09} for a mathematical introduction and \cite{Do14,Mar14}
for the physics-oriented reader. In \cite{CoSVa15}, the authors perform BE summation on the gauge-theoretic large $N$ expansion in the quartic 
matrix model and give an account of their results regarding analytic continuation, Stokes phenomena and monodromies.  
However, for the general scheme of things, we refer to \'Ecalle's lecture notes \cite{Eca93} and \cite{vH07} for a concise review on more general 
accelero-summation.   
 
\subsection{Nonperturbative equations \& resurgence}
Part of the story about resurgence is that the coefficients in the various sectors of a transseries depend 
on each other. For the most extreme form of resurgence this means in particular that the perturbative 
part has all the information necessary to construct the nonperturbative part, albeit in encoded form.    
'Encoded' means that in order to construct the full transseries, one needs to know \emph{how} to extract the nonperturbative information from the 
perturbative sector, that is, the coefficients of the perturbative series alone do not suffice and rules to compute the nonperturbative coefficients are 
needed. 

It is in canonical cases an ordinary differential equation that prescribes how the sectors communicate, as for example stated in Theorem 31.5 of \cite{Sa14}. 
In quantum mechanics, this has been demonstrated by Dunne and \"Unsal for the energy eigenvalues of the double-well 
potential and likewise for the sine-Gordon potential \cite{DunU14}, where the authors utilised a boundary condition for the corresponding wave functions 
to derive the necessary differential equation relating the sectors. This so-called \emph{Dunne-\"Unsal relation} has been checked in \cite{GaT15} for 
cubic, quartic and higher-degree potentials. While cubic and quartic potentials stood the test, the authors found that for quintic ones and beyond, it 
does not hold and is yet to be generalised.    

For QFTs, one might expect Dyson-Schwinger equations (DSEs) to harbour the corresponding relations. 
These equations, however, canonically derived from path integrals, make for a nice narrative but are plagued 
with more than just blemishes \cite{RuVeX09}: firstly, they are an infinite tower of coupled integral equations.
This fact already calls into question their aptness to define a QFT nonperturbatively. Secondly, when the 
need for renormalisation arises, one is required to smuggle in the aforementioned Z factors.
As already alluded to, these factors are to this day only well-defined perturbatively (with a finite cutoff) and therefore bring in a perturbative 
feature. To still achieve a nonperturbative interpretation of DSEs, one simply has to assume that these factors  
exist and do their job properly when taken as nonperturbative objects.

Assuming all this, one expects DSEs of a QFT to capture everything the theory can possibly 
describe \cite{CurP90}. But because the infinite tower must for practical purposes be truncated, one 
certainly loses part of the phenomena potentially captured by the tower. 

The very least we have to concede about QED and Yukawa theory in four spacetime dimensions is that 
in the light of these concerns, and given their tentative nonperturbative status as mentioned at the end
of §\ref{divergs}, the two theories are not well-defined nonperturbatively.

We nevertheless deemed it worthwhile to study two approximative DSEs in these theories and investigate 
their nonperturbative features by means of transseries, not least because we wanted to present an 
interesting new method.

\subsection{Scope of this work} 
In this paper, we present our results pertaining to approximations for the anomalous dimension of 
\begin{itemize}
 \item the fermion field in massless Yukawa theory and 
 \item the photon field in massless quantum electrodynamics (QED),
\end{itemize}
both in four dimensions of spacetime and in momentum subtraction scheme. 
We show that the associated self-consistent \emph{Dyson-Schwinger equations} (DSEs) should provide at least in principle the nonperturbative conditions needed 
to relate the various sectors of the anomalous dimension's transseries. 
This is realised in the form of a \emph{fixed point equation} derived from a DSE which 
\begin{enumerate}
 \item [(i)]  prescribes how the perturbative sector relates to the nonperturbative ones and
 \item [(ii)] shows that if the anomalous dimension can ever be characterised by a resurgent transseries, then (\ref{pertexT}) regrettably is \emph{not} 
              the answer, as its entire nonperturbative part must vanish to satisfy the DSE; our ansatz is simply not elaborate enough a transseries.   
\end{enumerate}
We achieve this by, in a nutshell, plugging a slightly generalised version of the transseries ansatz (\ref{pertexT}) into these equations and prove the two
statements (i) and (ii), the latter by induction. For simplicity, we let all log polynomials be trivial, ie $P_{\sigma}=1$ for all $\sigma \geq 0$, a choice which does not alter the results, 
as will be explained.

However, the reader be warned that we use the terminology of transseries theory in the spirit of \cite{Ed09}. The reason we go about this task 
equipped with this seemingly abstract lingo is that the fixed point equations are highly nontrivial; to get a foretaste of how horrendous these equations
are, consider this: let $\M=(s\alpha \partial_{\alpha}-1)$ be a linear differential operator ($s=1,2$), let $\gamma$ the anomalous dimension and 
$\gamma_{n+1} = (\gamma \M)^{n}\gamma$ define a family of functions obtained from applying the operator $\Rn= \gamma \M$ multiple times to $\gamma$, 
ie $\gamma_{1}=\gamma$. We call these functions 'RG functions' and the relation $\gamma_{n+1}=\Rn\gamma_{n}=\gamma \M \gamma_{n}$ 'RG recursion', 
so called because its origin lies in the renormalisation group (RG) equation. Then the fixed point equation for $\gamma$ is of the form 
\begin{equation}\label{flavdse}
 \gamma_{1}(\alpha) = \sum_{i=1}^{N} \alpha^{i} \sum_{j=1}^{\infty}  X_{ij}(\gamma_{1}(\alpha), \gamma_{2}(\alpha), \dotsc ) \hs{1} 
 (\mbox{'DSE for the anomalous dimension'})  
\end{equation}
where the $X_{ij}$'s are polynomials of various degrees and numbers of variables into which the RG functions are plugged. Because the rhs of 
(\ref{flavdse}) has an infinite number of terms, it is surely not a differential equation. $N \geq 1$ enumerates the skeletons 
inherited from the skeleton diagrams of the DSE for the self-energy. While on the one hand the limit $N \rightarrow \infty$ is fishy, it may on the other
be understood as a sequence of DSEs. The so-defined sequence of nonperturbative solutions may have a limit which then gives us the anomalous dimension.     

What we then do in this work is to replace all RG functions by their transseries representations $\wg_{1}, \wg_{2}=\Rn \wg_{1}, \wg_{3}=\Rn^{2} \wg_{1}, \dotsc $ in this equation and subsequently investigate sector-wise 
how the coefficients of the anomalous dimension's transseries $\wg_{1}$ are related. Since this must be conducted with prudence, we have taken care in 
decomposing the task into smaller feasible units. 

And here is an uncomfortable obstruction we are facing: in the algebra of transseries, the fixed point equation (\ref{flavdse}) is known to have a 
solution if the rhs represents a contractive operator in a sense to be explained \cite{Ed09}. There are two salient aspects to this. 
Firstly, it is currently not known whether a solution still can be related to an analysable function even if the rhs of (\ref{flavdse}) is contractive 
as a nonlinear operator in some Banach space: the set of analysable functions may not be closed with respect to limits of this kind \cite{Co09}. 
Secondly, we know that the perturbative series alone satisfies this equation. Because we seek a transseries solution serving as a nonperturbative 
completion, here is something that adds to the intricacy of the situation: a fixed point theorem that offers a unique fixed point must either be 
crossed off our wish list or a subset of transseries be found that does not contain the perturbative 
solution and on which the DSEs admit a unique solution.  

This is a typical situation in which physics musters the blitheness to carry on assuming that there is such a 
function to work out the transseries, preferably armed with physical arguments. 

On the downside, our analysis is by its very nature algebraic, technical and almost completely void of physical considerations.
It offers no adhoc physical explanation as to why (\ref{pertexT}) is not the correct transseries, an aspect being especially unsatisfactory because this 
type of transseries would at least nicely account for and take care of the Stokes effect which we expect to play a role (explained in §\ref{sec:NPC}).

On the upside, our investigation does not exclude the possibility of there being a solution of (\ref{flavdse}) and presents a method which may also be 
useful in other contexts, or at least complementary to the conventional approaches. 
However, we  
\begin{itemize}
 \item [(iii)] present a weighty argument as to why a transseries ansatz of the form (\ref{pertexT}) cannot be expected to capture the physics 
              of a fully renormalised theory,
 \item [(iv)] discuss an oddity incurred by the infinite skeleton expansion DSE in the case of the photon 
        corresponding to the limit $N \rightarrow \infty$ in (\ref{flavdse}).
\end{itemize}
Regarding (iii), we acknowledge the speculative and tentative character of our ideas and that the case of nonabelian gauge 
theories may be entirely different.

\subsection{Outline} §\ref{sec:NPC} introduces the reader to the basics of Borel summation and muses some more on why it is that the observables of QFT 
need (at the very least) a transseries representation. Already here will the reader be confronted with an argument as to why the author believes that 
the resurgence question will be much harder to tackle (in future projects) as one enters the realm of renormalised quantum field theories. 

The subsequent two sections (§§\ref{sec:DySch},\ref{sec:ReTra}) prepare the ground for the main body of our work which commences in later sections. 
In particular, §\ref{sec:DySch} has a review on the Dyson-Schwinger equations (DSEs) studied in this work and shows that the renormalisation group 
(RG) equation enables us to formulate a DSE solely for the anomalous dimension, albeit in terms of a formal expansion of the form (\ref{flavdse}). 
In order to derive this equation, we use a method based on meromorphic functions known as \emph{Mellin transforms} which we will explain alongside 
the above-mentioned RG recursion \cite{Y11}. 

For a gentle pedagogical start, we first go through the derivation in the case of the well-known \emph{rainbow approximation} for the fermion self-energy 
in massless Yukawa theory, the most trivial DSE available in four-dimensional QFT. Next, we gear up and turn to the DSE of 
the \emph{Kilroy approximation} in massless Yukawa theory and play the analogous game there. The same procedure is then gone through for the photon's 
self-energy and its anomalous dimension in QED, considerably less trivial than the two Yukawa cases. In contrast to the situation we face in QED, the 
Yukawa model implies a nonlinear ODE in a straightforward manner. Because it is particularly amenable to a transseries investigation, we have included
its derivation.

Up to this point, the material is not entirely new and goes back to the work of Broadhurst, Kreimer and Yeats \cite{BroK01,Krei06,KrY06,Y11}. The work of 
Bellon and Clavier concerning the Wess-Zumino model in \cite{Bel10, BeC15} is to some extent related as the authors also make use of the RG 
recursion and Mellin transforms to obtain ODEs for the anomalous dimension in the spirit of \cite{KrY06}. But in contrast to their investigations, we neither look for 
singularities in the Borel plane nor aim at finding the asymptotics of the perturbative coefficients. 
We present the material here partly for the convenience of the reader but also because the DSEs for the anomalous dimension (\ref{flavdse}) 
cannot be found explicitly anywhere which is why these equations are in some sense novel. This is true in a strict sense for the QED case which 
confronts us with an equation that becomes most interesting in the seeming limit to the full theory. 
We contend that the limit of an infinite-skeleton DSE cannot per se be considered as a nonperturbative equation for the full theory although the 
combinatorics of Feynman diagrams suggests so. However, as alluded to above, we propose to understand it as a sequence of DSEs which in 
turn defines a sequence of nonperturbative solutions whose limit is what we are after.  

Since we do not expect the reader to be familiar with resurgent transseries, we have devoted §\ref{sec:ReTra} to a concise introduction to this 
topic. The main sources we have drawn on and whose (to our mind) apt lingo we use in our work are \cite{Ed09,vH06}. 

In §\ref{sec:RenorTra} we explain why we believe renormalisation to be a game changer when it comes to the class of transseries that might have to 
be employed for renormalised quantum field theories. We then describe our transseries ansatz in §\ref{sec:Transatz}.    

Building on the preliminary material covered in the preceding sections, we will in §\ref{sec:DynSys} treat the RG recursion as a discrete dynamical 
system in the algebra of our transseries and study whether it may converge in a sense to be expounded. To this end, we analyse how the support of the 
RG functions' transseries changes along the orbit (explained there in detail).

§\ref{sec:DSE} explores whether the fixed point equation for the anomalous dimension (\ref{flavdse}) may allow for more than one solution. 
We review a pertinent fixed point theorem and seek for conditions imposed on the transseries ansatz. 
Because it is very hard to find a subset on which the associated 'Dyson-Schwinger (DS) operator' is contractive, the most convenient stance is to
consider a rock-solid fixed point theorem as the last thing we can possibly want: it would spoil the game by decreeing that there be only one transseries 
solution, namely the perturbative series.  

In §\ref{sec:RGflow} we introduce an algebraic method suitable to analyse the flow of data along the orbit of the RG recursion. 
To keep track of the flow of perturbative and nonperturbative information, we employ transseries with coefficients in a graded free algebra. 
The flow of the RG recursion turns out to preserve one key feature of the transseries which we call \emph{sector homogeneity}. Because 
sector-homogeneous transseries form a subalgebra stable under the RG operator, a certain degree of orderliness in which information is 
being passed on along the orbit is warranted. 

We come in §\ref{sec:PinNP} to the first main result pertaining to the resurgence of the anomalous dimension.
Although the DSE prescribes the perturbative sector to communicate with all nonperturbative ones in a manner clearly one-way and characteristic of 
resurgence, we prove in §\ref{sec:ansatz} that the nonperturbative sector vanishes, our second main result. The ODE for the Yukawa model enables us to 
extend our investigation to a wider class of transseries, albeit also with a negative outcome.

However, we believe that there is no contradiction to the result on intersectorial communication, as the general pattern of sector crosstalk should still be 
valid as long as exponentials are involved in defining nonperturbative sectors. It only means we have conducted our investigation in the wrong subclass 
of transseries.
Finally, §\ref{sec:Conclu} briefly summarises and discusses the obtained results.

\section{The need for a nonperturbative completion}\label{sec:NPC}			  
We continue the discussion started in the introductory section on transseries in quantum mechanics and QFT to provide a little more background on 
Borel summation and how transseries arise. Before we properly justify the idea that one should not expect a transseries 
like (\ref{pertexT}) to be appropriate for a renormalised QFT in §\ref{subsec:Recha}, we will in this section already present the main argument, suited 
to the technical level at which we have so far expounded transseries.   

\subsection{Borel summation}\label{subsec:Bor}
We mentioned in §\ref{sec:Intro} that to obtain the sought-after function $F(\alpha)$, its asymptotic (that is, divergent) series $\ti{F}(\alpha)$ may in 
'benign' cases be put through the \emph{Borel summation} procedure, a mathematical machinery devised to construct a function to which the series is asymptotic.  
To get the general idea, suppose we are given an asymptotic series $\tilde{\varphi}=\sum_{k \geq 0} a_{k}\alpha^{k} \in \R[[\alpha]]$, then the 
formal computation 
\begin{equation}\label{Borelmach}
 \sum_{k \geq 0} a_{k}\alpha^{k} = \sum_{k \geq 0} \frac{a_{k}}{k!} \alpha^{k} \int_{0}^{\infty} \zeta^{k} e^{-\zeta} d\zeta 
 = \sum_{k \geq 0} \frac{a_{k}}{k!} \int_{0}^{\infty} (\alpha \zeta)^{k} e^{-\zeta} d\zeta 
 = \frac{1}{\alpha} \int_{0}^{\infty} e^{-\zeta/\alpha} (\sum_{k\geq 0} \frac{a_{k}}{k!} \zeta^{k}) d\zeta    
\end{equation}
nicely captures the essence of Borel summation. A benign case, for instance, is given if the following conditions are met: 
\begin{itemize}
 \item [(i)] the formal series $\sum_{k \geq 0} a_{k}\alpha^{k}$ is \emph{Gevrey 1}, which implies in particular that the series inside the integral, 
        the so-called \emph{formal Borel transform} 
        \[
          \hat{\varphi}(\zeta):=(\Bo \tilde{\varphi})(\zeta) := \sum_{k\geq 0} \frac{a_{k}}{k!} \zeta^{k}
        \]
       has nonzero radius of convergence in what is called its \emph{Borel plane}, just another name for the complex plane $\C$. The '1' in 
       'Gevrey 1' stands for the fact that the original series' coefficients need to be divided by at least one power of $k!$ to yield a convergent 
       series\footnote{If no power of $k!$ is necessary because the series is convergent from the start one calls it \emph{Gevrey 0}.};   
 \item [(ii)] the formal Borel transform $\hat{\varphi}$ possesses an analytic continuation $\cont_{\Nei} \hat{\varphi}$ to a neighbourhood $\Nei$ of the
       positive real axis $\R^{+} \subset \C$ and  
 \item [(iii)] $\cont_{\Nei} \hat{\varphi}$, known as \emph{Borel transform}, is of exponentially bounded 
       type, that is, there exist constants $c,A >0$ such that 
\begin{equation}
 |\cont_{\Nei} \hat{\varphi}(\zeta)| \leq A e^{c |\zeta|} 
\end{equation}
      for all $\zeta \in \C$ in that neighbourhood. 
\end{itemize}
In this admittedly very restricted case, the \emph{Borel-Laplace transform} $\La \hat{\varphi}$ defined by
\begin{equation}\label{LaBo}
 (\La \hat{\varphi})(\alpha) :=  \frac{1}{\alpha} \int_{0}^{\infty} e^{-\zeta/\alpha} \cont_{\Nei} \hat{\varphi}(\zeta) d\zeta
\end{equation}
is a function called the 'Borel sum' to which the series $\tilde{\varphi}$ we started out with is asymptotic. In summary, the Borel machine is an 
operator $\LB:=\La \circ \Bo$ which takes a 'reasonable enough' asymptotic series $\tilde{\varphi}$ and processes it into a function 
$\LB \tilde{\varphi} = (\La \circ \Bo )\tilde{\varphi} = \La \hat{\varphi}$ 
to which the divergent series $\tilde{\varphi}$ is asymptotic. The series $\tilde{\varphi}$ is then called \emph{Borel summable}, or, more precisely in 
the modern terminology of resurgence, \emph{fine-summable in the direction} $\R^{+}$ \cite{Sa14}. 
This is in short what is referred to as \emph{classical Borel summation}. 

It is known that Yukawa theory in dimensions $d=2,3$ and $(\varphi^{4})_{d}$ theory for $d=1,2,3$ are Borel 
summable in this sense \cite{Ri91, GliJa81}.  
But, as shall be elaborated in a moment, we have to expect that condition (ii) is in general not satisfied 
because the Borel transform \emph{cannot} be analytically continued to a neighbourhood of $\R^{+}$ due to 
singularities sitting there. In these cases the described 'Borel machine' (\ref{Borelmach}) is clearly not 
an apt tool and must be modified.

Moreover, the three conditions (i) to (iii), which are not entirely unrelated, may be violated altogether. In these cases, it will be necessary to 
enter the realm of \emph{multisummability} \cite{Ba09}, a topic which might have to be put on the agenda in future projects. The line in (\ref{Borelmach})
takes for $m$-summability a different form and reads 
\begin{equation}\label{Borelkmach}
\sum_{k \geq 0} a_{k}\alpha^{k} = \sum_{k \geq 0} \frac{a_{k}}{\Gamma(1+k/m)} \int_{0}^{\infty} (\alpha^{m}\zeta)^{k/m} e^{-\zeta} d\zeta  
 = \alpha^{-m} \int_{0}^{\infty} e^{-\zeta/\alpha^{m}} (\Bo_{m}\ti{\varphi})(\zeta^{1/m}) d\zeta,    
\end{equation}
where the assignment
\begin{equation}
\ti{\varphi}(\alpha) = \sum_{k \geq 0} a_{k}\alpha^{k} \mapsto (\Bo_{m}\ti{\varphi})(\zeta) := \sum_{k \geq 0} \frac{a_{k}}{\Gamma(1+k/m)}\zeta^{k}
\end{equation}
is the so-called \emph{formal Borel transform} with index $m \geq 1$. The Borel-Laplace transform, 
given by the integral in the rhs of (\ref{Borelkmach}) must then carry an 
index, ie $\La_{m}$, and the exponential growth bound (iii) be replaced by
\begin{equation}
 |\cont_{\Nei} (\Bo_{m}\ti{\varphi})(\zeta)| \leq A e^{c |\zeta|^{m}} .
\end{equation}
In fact, our transseries ansatz, to be introduced in §\ref{sec:Transatz}, accounts for this more general form of 
Borel summability ($m \geq 1$). We will briefly come back to multisummability in §\ref{subsec:Recha}.   

However, because we think speculations on multisummability are premature, we note only that currently, the best educated guess (or  belief) about the observables 
of a QFT floated among experts is the situation depicted in \textsc{Figure} \ref{Bplane}: the Borel plane is punctured by an infinite number of 
singularities of the Borel transform in such a way that there are a finite number of singular rays called 'Stokes rays' (or 'Stokes lines') which 
emanate from the origin and carry a countable number of singularities \cite{BaDU13}. 

\begin{figure}[ht]
\begin{center} \includegraphics[height=4cm]{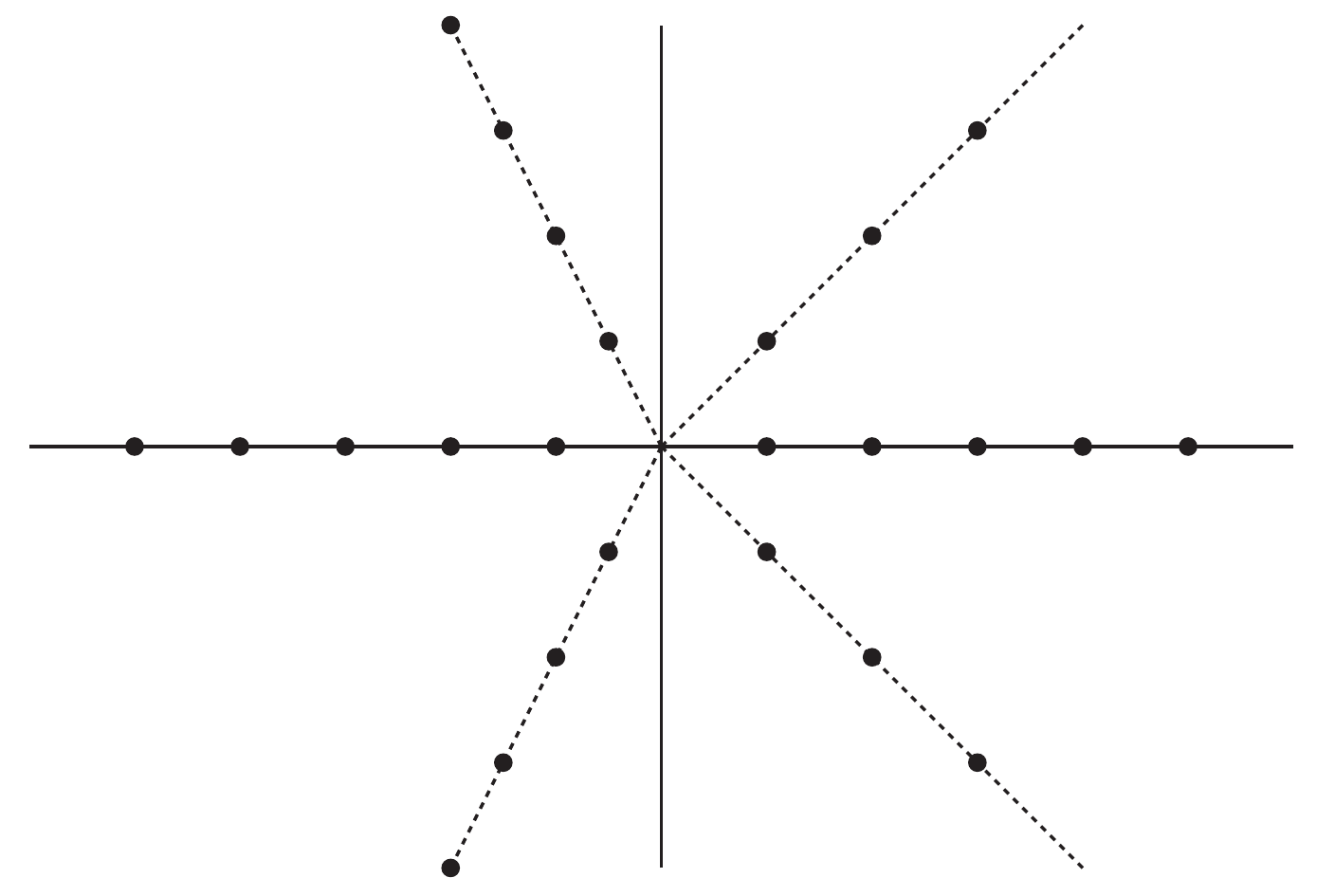} \end{center} 
\caption{\small Illustration of the conjectured Borel plane structure for QFTs with (in this case) six Stokes rays. 
Note that complex singularities always have complex conjugate partners (from \cite{BaDU13}, see discussion there). }
\label{Bplane}
\end{figure}

From the arguments laid out by Beneke in \cite{Ben99}, we know that amplitudes in quantum electrodynamics (QED) and quantum chromodynamics (QCD) are 
expected to exhibit \emph{UV} and \emph{IR renormalons}, ie singularities in the Borel plane on the real line. Since the arguments presented in \cite{Ben99} 
about QED easily carry over to Yukawa theory one finds renormalons there as well. In particular, we shall see in §\ref{subsec:Kil} that our Yukawa model has
UV renormalons on the negative real half-line $\R^{-}$ incurred by fermion line corrections. 

Unfortunately, we are unable to ascribe a nonperturbative meaning to these renormalons. Since they are of the UV type, one is 
reluctant to relate them to bound states, ie low-energy phenomena, and is more inclined to think of 
higher-energy effects like the so-called \emph{Schwinger effect}, usually modelled 
by means of a strong external electric field \cite{Dun08,RuVeX09}.   

\subsection{Resurgence \& resurgent transseries} 
As alluded to, Borel summation cannot be applied to these cases and must be altered. To make sense of the Borel-Laplace transform in (\ref{LaBo}),
the integration contour must be diverted around the singularities on $\R^{+}$ which brings about what is known as the \emph{Stokes effect}. 
In view of the Laplace transform (\ref{LaBo}) and Cauchy's theorem, it is not hard to imagine that this effect leads in particular to terms proportional to 
\begin{equation}
 \pm 2 \pi i e^{-\omega/\alpha}  \hs{2} 
 ( \ \pm 2 \pi i e^{-\omega/\alpha^{m}} \text{ for $m$-summability}\footnote{One speaks in this case of a 'higher-level Stokes effect' \cite{Ba00}.}  \ )  
\end{equation}
where $\omega \in \R^{+}$ is a singularity of $\cont_{\Nei} \hat{\varphi}$ in the Borel plane and the sign depends on whether the singularity has been circumvented to the right or to the left of the 
real axis. This introduces an ambiguity because both choices of circumventing the singularity are permissible and yield Borel sums perturbatively represented 
by the same asymptotic expansion. 

Formal power series are therefore no longer sufficient objects to contain all information needed. This is the starting point of the above-mentioned 
theory of resurgence and accelero-summation \cite{Eca81,Eca93}. It informs us about how to deal with situations like the one of \textsc{Figure} \ref{Bplane}. 
Depending on the problem at hand (eg a differential equation), the perturbation series must be replaced by a transseries. 
The ambiguity issue incurred by the Stokes effect may be resolved by an ambiguity popping up in the same way from Borel summing the 
formal power series associated with a higher nonperturbative sector. 
Given an infinite number of singularities on a Stokes ray as in \textsc{Figure} \ref{Bplane}, one clearly needs an infinite number of such higher-sector
power series. However, all this must rest on an extensive degree of interconnectedness and hence communication between the various sectors 
of the transseries. Without it, there would be no cancellations of ambiguous imaginary parts. This intersectorial communication, in canonical cases
mediated by a differential equation or simply imposed by sheer will to eradicate Stokes factors, is what one often refers to as 
\emph{resurgence} \cite{Sa14}. In \cite{AS14}, Aniceto and Schiappa have found constraints which follow if one demands that the transseries be real. 
Since Stokes factors bring in imaginary parts, they must cancel if these conditions are to be met. 

However, we will not directly make use of resurgence theory in this paper and have therefore no intention of expounding it 
here\footnote{We refer the interested reader to the literature described at the end of §\ref{AnalyF}.}. 
The purpose of mentioning resurgence theory and devoting some space to it here is to motivate our use of \emph{resurgent transseries}. 
Such series are, in fact, interesting mathematical objects in their own right. The associated theory is heavily algebraic in flavour and has developed 
a life of its own, as in particular a quick look into the monograph \cite{vH06} reveals.  
And because our analysis makes extensive use of transseries theory formulated in the abstract lingo of \cite{Ed09, vH06}, we expound some of it 
in §\ref{sec:ReTra} to the extent we deemed absolutely necessary.   

For this informal introduction, we content ourselves with a concrete example of a transseries from quantum mechanics. 
Consider a (non-relativistic) quantum particle in a one-dimensional double-well potential
\begin{equation}\label{dw}
 V(x) = \frac{1}{2}x^{2}(1-\sqrt{g}x)^{2} ,
\end{equation}
viewed as a 'perturbed' single harmonic well with anharmonicity parameter $g$, ie the double well becomes a single well with a standard 
harmonic-oscillator solution in the limit $g \rightarrow 0$.

The ground state energy can be obtained by means of the WKB or a path integral approach \cite{JenZ04, ZiJ04} in the form of a resurgent transseries 
given by 
\begin{equation}\label{grou}
E_{0}(g) = \underbrace{\sum_{m \geq 0} c_{(0,m,0)} g^{m}}_{\mbox{\tiny perturbative sector}}
+ \underbrace{\sum_{l_{1} \geq 1} \sum_{l_{3} = 1}^{l_{1}-1} \sum_{l_{2} \geq 0} 
c_{(l_{1},l_{2},l_{3})} \left(\frac{e^{-S/g}}{\sqrt{g}}\right)^{l_{1}}  g^{l_{2}}  [\log(-2/g)]^{l_{3}} }_{\mbox{\tiny nonperturbative sectors}} ,
\end{equation}
and likewise for the excited energies $E_{1}(g), E_{2}(g), \dotsc$
The first piece, the \emph{perturbative sector}, is the usual perturbative expansion composed solely of the usual monomials $g^{m}$, while the 
\emph{nonperturbative} ones sport powers of the 3 transmonomials 
\begin{equation}\label{tramon}
 \mathfrak{t}_{1}= g^{-1/2} e^{-S/g}, \hs{1} \mathfrak{t}_{2}=g, \hs{1} \mathfrak{t}_{3}= \log(- 2/g),
\end{equation}
each having a physical meaning as part of the transseries. For example, the exponentially 'flat' function $\mathfrak{t}_{1}$, so called because its Taylor series
around zero vanishes, represents an \emph{instanton (event)}. An instanton event occurs when the quantum particle tunnels from one well into the other,
while higher powers of $\mathfrak{t}_{1}$ describe $n$-instantons, several back-and-forth tunnelling events \cite{Zi02}. Notice that these effects must 
necessarily go completely unnoticed in perturbation theory.

In the path integral, such events appear as critical points and the real number $S>0$ is the action of this event which then contributes to the ground state 
energy, albeit exponentially suppressed as described by (\ref{grou}). 
The associated power series in $g$ at the $n$-instanton level of the nonperturbative part corresponds to perturbative corrections around this 
$n$-instanton's critical point. For more, see \cite{Zi02}.  

\subsection{Transseries in quantum field theory}\label{subsec:transQFT}
The situation seems to be no different in QFT, as the toy model QFTs studied in \cite{DunU12, DunU13} suggest. 
The argument as to how resurgent transseries arise in QFT runs, according to \cite{BaDU13}, roughly as follows.
Let in Euclidean formulation 
\begin{equation}\label{saddd}
 \mathcal{Z} = \int \De \phi \  e^{-S(\phi)}
\end{equation}
be the partition function in terms of a path integral for a scalar field $\phi$, where the action is given by the integral
$S(\phi)=\alpha^{-1} \int \La(\phi,\partial \phi)$ whose Lagrangian is independent of the coupling $\alpha$ (the field is assumed to be 
normalised accordingly). 
If we perform perturbation theory around each critical point of the partition function, then the so-called \emph{semi-classical expansion}
\begin{equation}\label{sadd}
 \mathcal{Z} = \int \De \phi \  e^{-S(\phi)} \approx \sum_{k \in \{\mbox{\tiny saddles, ...}\}} e^{-S_{k}/\alpha} F_{k}(\alpha)
\end{equation}
is a transseries. The expression $F_{k}(\alpha)$ is the perturbative expansion around the $k$-th critical point $\phi_{k}$ with action $S(\phi_{k})=S_{k}/\alpha$, 
where $S_{0}=0$ is the vacuum action for $\phi_{0}=0$. 

Of course, (\ref{sadd}) has to be understood schematically and is in general not feasible. 
Nevertheless, the prospect is not so bleak, as there has been considerable progress with toy models: the various critical points of a two-dimensional 
nonlinear sigma model were shown to be associated with renormalons and also instantons, where instantons in QFT are field configurations that arise in 
nonabelian gauge theories for topological reasons \cite{DunU12}. Both configurations lead to factors of exponential flatness like the 
transmonomial $\mathfrak{t}_{1}$ in (\ref{tramon}). For an introduction to instantons in QFT, we refer the reader to \cite{Zi02} or the (classical) 
review \cite{VaZNoSh82}, as well as to the nice textbooks \cite{Shi12, Mar15}.

Not all QFTs exhibit instantons. Although QED and Yukawa theory lack such configurations, one might expect
renormalons to be associated to critical points in the spirit of (\ref{sadd}). But it is not at all
clear what field configurations and nonperturbative phenomena they correspond to. Both theories are 
asymptotically infrared-free and are therefore weakly coupled for bound states like positronium.  
Because exponentials are in this regime (exponentially) suppressed, such states can be considered less 
likely to manifest themselves in the form of transmonomials with exponentials.  

Old results extracted from (nonrelativistic) Bethe-Salpeter equations concerning the decay rate and 
hyperfine splitting of positronium which revealed nonperturbative contributions free of exponentials 
and proportional to $\alpha^{2} \ln \alpha$ support this view \cite{CasLe79,APSaW15} and are to some 
extent in agreement with experimental data \cite{ARG94}. 

The renormalons are therefore more likely to be related to higher-energy states for which we know QED
does not describe what actually happens. 

However, within the scope of our two toy models, this work suggests that the anomalous dimension cannot 
be represented by a resurgent transseries of the form (\ref{pertexT}) so that the above ideas surrounding 
the semi-classical expansion (\ref{sadd}) seem no longer viable for fully-fledged four-dimensional QFTs. 

And here is what we personally believe to be the reason why things may be much more complicated in this context than currently conceivable: (\ref{sadd}) is 
\emph{not true} for a renormalised action. First, the coupling dependence cannot be scaled away and brought to the
front of the action integral to emerge in this manner. Second, if we assume a finite cut-off, the renormalisation Z factors are strictly speaking 
only given as an asymptotic expansion. 

So far, there is no nonperturbative way to determine the Z factors any more accurately than perturbatively \cite{Ost86}. Here and there, assumptions are made and 
connections to other parts of the formalism (LSZ formula, K\"allen-Lehmann spectral representation) proposed (see any textbook on QFT). 
But that does not define them nonperturbatively. 
Even if they were given in this way, the very recipe along which canonical perturbation theory is performed makes one thing very clear: the Euclidean 
damping factor in (\ref{saddd}) is treated as a formal power series plugged into an exponential. The discussion of this point will be picked up again 
in §\ref{subsec:Recha} once we have introduced transseries in their full generality. 

Of course, when renormalisation is less severe, these arguments do not apply, as eg in superrenormalisable or supersymmetric (SUSY) QFTs where the 
Z factors are more or less trivial \cite{ARuS15, BaD15}.
Basar and Dunne describe in \cite{BaD15} an interesting SU(2) SUSY gauge theory whose correspondence to quantum-mechanical systems guarantees that it 
possesses transseries representations of the 'usual' form (\ref{pertexT}), an aspect first worked out by Krefl in \cite{Kre14a,Kre14b}.

\section{Dyson-Schwinger equations \& RG recursion}\label{sec:DySch}  	                  
We review the two Dyson-Schwinger equations (DSEs) investigated in this paper and introduce the reader to an approach which makes use of 
meromorphic functions called \emph{Mellin transforms}. Each Dyson-Schwinger skeleton is associated with one such function, as we will explain in 
due course\footnote{The reader will then see the reason for the denomination 'Mellin transform'.}. 
This method has been introduced and first employed by Kreimer and Yeats to bring DSEs into a convenient form \cite{Krei06,KrY06,Y11}.  
We will make use of it to derive the equations for the associated anomalous dimensions. We will explicate it in detail as we derive the pertinent formulae. 
The level of rigour throughout this section is deliberately kept low to cut an otherwise unreasonably long 
story short\footnote{Mathematically minded readers must fill in the gaps for themselves, ie making an assumption here an there. To our mind, this is 
more than futile: no one knows whether these conditions are satisfied.}. 

DSEs first appeared in the seminal publications by Dyson \cite{Dys49} and Schwinger \cite{Schwi51} and since then have found many 
applications in high-energy physics \cite{AlS01,RoWi94}. The approximative DSEs we focus on in this work can only be understood 
combinatorially and not deduced from path integrals. We will thus motivate their formulation based on the self-similarity of Feynman diagram series. 

\subsection{Rainbow approximation} In Yukawa theory, let us consider the diagrammatic expansion 
\begin{equation}\label{rba}
\graph{0}{-0.35}{0.45}{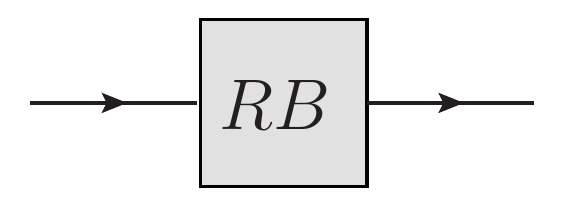}{0.1} :=  \graph{0.02}{0.05}{0.3}{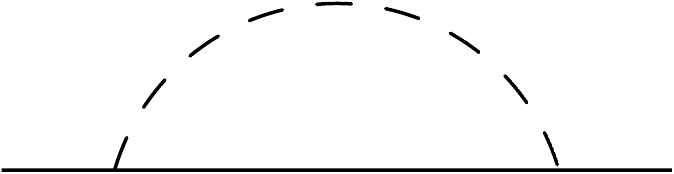}{0.1} 
+ \graph{0.02}{0.05}{0.3}{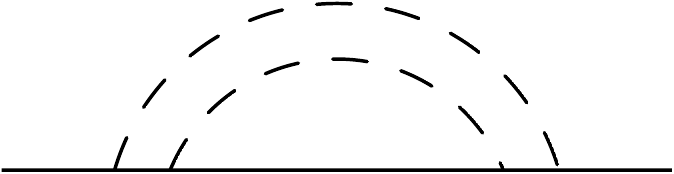}{0.1} + \graph{0.02}{0.05}{0.3}{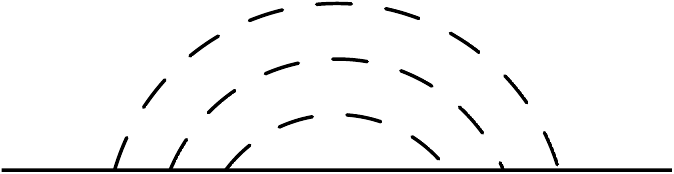}{0.1}  + \dotsc 
\end{equation}
of the fermion's self-energy in which only one-particle irreducible (1PI) \emph{rainbow diagrams} are taken into account. A shaded box will in this paper
generally stand for a 1PI series. 

This admittedly rather crude approximation will concern us here for two reasons. Firstly, it can be solved exactly in the massless case so that we have 
the luxury of being able to directly eye the transseries of its anomalous dimension \cite{Krei06, DeKaTh97}. It secondly makes for a nice preliminary 
exercise to understand the Mellin transform method when employed to derive the DSEs for anomalous dimensions in the more intricate cases. 
 
Let us translate (\ref{rba}) into Feynman integrals by the schematic prescription 
\begin{equation}
\int K =  \graph{0.02}{0.05}{0.3}{1yup1}{0.1} , \hs{2} -i\Sigma_{RB} =  \graph{0}{-0.36}{0.45}{RBFerSE}{0.1} 
\end{equation}
with the appropriate integral kernel $K$ and the self-energy $-i\Sigma_{RB}$ of the rainbow to obtain
\begin{equation}\label{1aps}
 -i\Sigma_{RB} = \int K + \int K \int K + \int K \int K \int K + \dotsc = \int K +  \int K ( \int K + \int K \int K + \dotsc ) .
\end{equation}
Since the term in brackets is again the perturbation series for the self-energy, we can rewrite everything as an integral equation 
\begin{equation}\label{rb}
 -i\Sigma_{RB} = \int K -i \int K \  \Sigma_{RB} \ .
\end{equation}  
In terms of blob diagrams, this takes the form
\begin{equation}\label{RBY}
 \graph{0}{-0.35}{0.45}{RBFerSE}{0.1} =  \graph{0}{-0.1}{0.4}{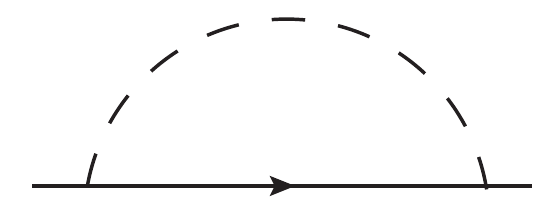}{0.2} + \graph{0.1}{-0.45}{0.45}{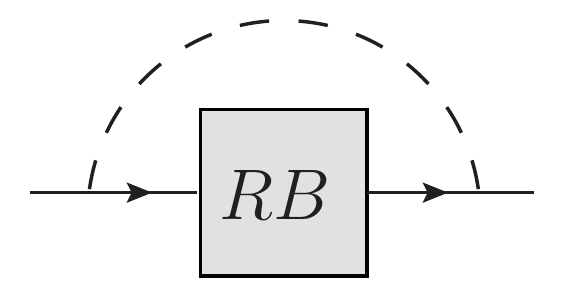}{0.2} .
\end{equation}
Explicitly, with external Minkowski momentum $q \in \R^{4}$, the Feynman integral of the skeleton graph $\graph{0}{0}{0.1}{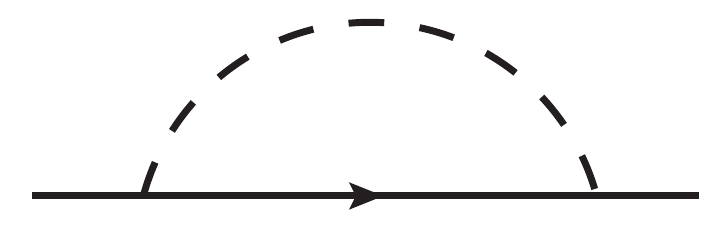}{0}$ prescribed by 
the Feynman rules of Yukawa theory reads\footnote{See \cite{PeSch95} for the Feynman rules of Yukawa theory.}
\begin{equation}
\graph{0}{-0.1}{0.4}{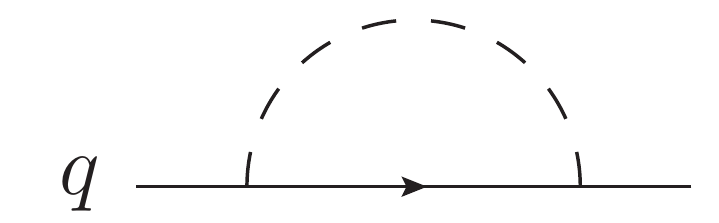}{0.1} = (-ig)^{2} \int \frac{d^{4}k}{(2\pi)^{4}} \frac{i}{\slashed{k}+i\epsilon} 
\frac{i}{(q-k)^{2}+i\epsilon} = i \slashed{q} A_{1}(q^{2}) 
\end{equation}
and we are led to the concrete analytic form of the rainbow DSE (\ref{RBY}), given by 
\begin{equation}\label{RBY1}
-i\Sigma_{RB}(q) = (-ig)^{2} \int \frac{d^{4}k}{(2\pi)^{4}} \frac{i}{\slashed{k} + i \epsilon} \left[ 1 
- i \Sigma_{RB}(k)\frac{i}{\slashed{k} + i \epsilon } \right] \frac{i}{(q-k)^{2} + i \epsilon}  .
\end{equation}
We write $-i\Sigma_{RB}(q)= i \slashed{q}A(q^{2})$ and pass over into Euclidean space to get 
\begin{equation}\label{RBY2}
i A(-q^{2}_{E}) = i a \int \frac{d^{4}k_{E}}{2 \pi^{2}} \frac{1}{k_{E}^{2}(q_{E}-k_{E})^{2}} \left[ 1 - A(-q^{2}_{E}) \right],
\end{equation}
where the new coupling $a=g^{2}/(4\pi)^{2}$ is a convenient choice \cite{BroK01}. When renormalised in momentum scheme, this equation morphs into
\begin{equation}\label{RBYR}
A_{R}(q^{2}_{E},\mu^{2}) = a \int \frac{d^{4}k_{E}}{2 \pi^{2}} \left\{ \frac{1}{k_{E}^{2}(q_{E}-k_{E})^{2}} 
-  \frac{1}{k_{E}^{2}(\ti{q}_{E}-k_{E})^{2}} \right\}
\left[ 1 - A_{R}(k^{2}_{E},\mu^{2}) \right]
\end{equation}
in which $A_{R}(q^{2}_{E},\mu^{2})$ is the renormalised self-energy's form factor with $A_{R}(\mu^{2},\mu^{2})=0$ and $\ti{q}_{E}$ is 
the Euclidean reference momentum with reference (renormalisation) scale $\mu > 0$, ie $\ti{q}_{E}^{2}=\mu^{2}$. One can solve this now by means of a
\emph{scaling ansatz} 
\begin{equation}\label{ansatz}
1 - A_{R}(q^{2}_{E},\mu^{2}) = \left( \frac{q_{E}^{2}}{\mu^{2}} \right)^{-\gamma(a)} 
\end{equation}
giving
\begin{equation}\label{RBYR1}
\begin{split}
\left( \frac{q_{E}^{2}}{\mu^{2}} \right)^{-\gamma(a)} &= 1 - a \int \frac{d^{4}k_{E}}{2 \pi^{2}} \left\{ \frac{1}{k_{E}^{2}(q_{E}-k_{E})^{2}} 
-  \frac{1}{k_{E}^{2}(\ti{q}_{E}-k_{E})^{2}} \right\} \left( \frac{k_{E}^{2}}{\mu^{2}} \right)^{-\gamma(a)} \\
&= 1 - a \left\{ \left( \frac{q_{E}^{2}}{\mu^{2}} \right)^{-\gamma(a)} - 1 \right\} F(\gamma(a)),
\end{split}
\end{equation}
where the meromorphic function $F(\rho)$ is the \emph{Mellin transform} of the skeleton $\graph{0}{0}{0.1}{Yu}{0}$ defined by
\begin{equation}\label{MellinT}
 F(\rho):= \left. \int \frac{d^{4}k_{E}}{2\pi^{2}} \frac{(k_{E}^{2})^{-\rho}}{k_{E}^{2}(k_{E}-q_{E})^{2}} \right|_{q_{E}^{2}=1} 
 = \dotsc (\mbox{Gamma functions}) \dotsc = \frac{1}{2\rho(1-\rho)}.
\end{equation}
Note that (\ref{RBYR1}) implies
\begin{equation}\label{Mellineq}
 1=-aF(\gamma(a)) = \frac{a}{2 \gamma(a)(\gamma(a)-1)}.
\end{equation}
We require $\gamma(0)=0$ as a physical condition imposed on $\gamma(a)$ because we interprete this function as the anomalous dimension. (\ref{Mellineq})
is an algebraic equation whose solution is an algebraic function,
\begin{equation}\label{ad}
 \gamma(a) = \frac{1 - \sqrt{1+2a}}{2} = \sum_{n \geq 1} (-1)^{n} \frac{(2n-3)!!}{2 \cdot n!} a^{n},
\end{equation}
from which we can see that its transseries is 'trivial' in the sense that it is a convergent power series, void of any nonperturbative pieces. 
Its Borel transform yields an entire function with no poles in the Borel plane. The reason is obvious: for one thing, in (\ref{rba}) there is only one graph at each loop order which 
entails that the growth in the number of Feynman graphs with loop order is precisely zero since the number of contributing diagrams does not increase.
For another, this model has no renormalons (explained in §\ref{subsec:Kil}).

Imagine we were not able to guess the ansatz (\ref{ansatz}). There is another way of deriving the result (\ref{ad}). It is well-known that renormalised
perturbative contributions evaluate in the single-scale case to polynomials in the kinematical variable $L_{q}:=\log(q_{E}^{2}/\mu^{2})$. One might 
therefore formally expand 
\begin{equation}\label{ansatz1}
1 - A_{R}(q^{2}_{E},\mu^{2}) = 1 - \sum_{n \geq 1} \frac{\gamma_{n}(a)}{n!}L_{q}^{n} =: G(a,L_{q}) \hs{2} \mbox{('log expansion')},   
\end{equation}
in momentum logarithms, where $\gamma_{n}(a)$ are what we call the \emph{RG functions}, ie the derivatives of the self-energy with respect to $L_{q}$ 
evaluated at the reference point $L_{q}=0$. This so-called \emph{log expansion} is central to our approach and will be employed in both the Yukawa and 
the QED model. We identify the function $\gamma_{1}(a)=\gamma(a)$ with the anomalous dimension. The log expansion is then inserted into (\ref{RBYR}) and, 
applying the 'trick' 
$\lim_{\rho \rightarrow 0}(-\partial_{\rho})^{n}(k_{E}^{2}/\mu^{2})^{-\rho} = L_{k}^{n}$, we get 
\begin{equation}\label{RBYR2}
G(a,L_{q}) = 1 - a \lim_{\rho \rightarrow 0} G(a,-\partial_{\rho}) \int \frac{d^{4}k_{E}}{2 \pi^{2}} \left\{ \frac{1}{k_{E}^{2}(q_{E}-k_{E})^{2}} 
-  \frac{1}{k_{E}^{2}(\ti{q}_{E}-k_{E})^{2}} \right\} \left( \frac{k_{E}^{2}}{\mu^{2}} \right)^{-\rho} ,
\end{equation}
where $G(a,-\partial_{\rho})=1 - \sum_{n \geq 1} \frac{\gamma_{n}(a)}{n!}(-\partial_{\rho})^{n}$ is a formal differential operator. 

Using the Mellin transform (\ref{MellinT}), this DSE is now rewritten to yield \cite{KrY06} 
\begin{equation}\label{RBMell}
G(a,L_{q}) = 1 -  a \lim_{\rho \rightarrow 0} G(a,-\partial_{\rho})  ( e^{-\rho L_{q}} - 1 ) F(\rho) \hs{0.5}  (\mbox{rainbow DSE in Mellin guise}).
\end{equation}
By differentiating this equation with respect to $L_{q}$ and then setting $L_{q}=0$ we find  
\begin{equation}\label{RBMellad}
\gamma_{1}(a) = - a \lim_{\rho \rightarrow 0} G(a,-\partial_{\rho}) \rho F(\rho) 
= -\frac{a}{2} \left[ 1 - \sum_{n=1}^{\infty} (-1)^{n} \ \gamma_{n}(a) \right] \hs{1} (\mbox{rainbow}).
\end{equation}
This is almost the fixed point equation for the anomalous dimension because, as we will see now, the renormalisation group (RG) equation 
(also known as Callan-Symanzik equation) relates the \emph{RG functions} $\gamma_{n}(a)$ to the anomalous dimension. For the Yukawa fermion 
the RG equation reads\footnote{Notice that this is the RG equation for the form factor of the inverse propagator and that our convention for 
the anomalous dimension differs by a factor of $1/2$ for convenience, ie the conventional one is $1/2\gamma(a)$.} 
\begin{equation}\label{RGeq}
 \left[-\frac{\partial}{\partial L_{q}} + \Rn \right]G(a,L_{q})=0  \hs{2} (\mbox{RG equation})
\end{equation}
with RG operator $\Rn:=\beta(a)\partial_{a}-\gamma(a)$, where $\beta(a)$ is the beta function. The two operators $\partial_{L_{q}}$ 
and $\Rn$ commute so that the RG equation (\ref{RGeq}) implies the recursion 
\begin{equation}\label{RGre}
\left. \partial^{n+1}_{L_{q}} G(a,L_{q}) \right|_{L_{q}=0} = \left. \partial^{n}_{L_{q}} \Rn G(a,L_{q}) \right|_{L_{q}=0}
 = \left. \Rn \partial^{n}_{L_{q}} G(a,L_{q}) \right|_{L_{q}=0} = \left. \Rn^{n} \partial_{L_{q}}G(a,L_{q}) \right|_{L_{q}=0} 
\end{equation}
for the derivatives of the function $G(a,L_{q})$, where $\Rn^{n} = \Rn \circ \dotsc \circ \Rn$ is the $n$-fold application of the RG operator \cite{KrY06,Y11}. 
In terms of the RG functions $\gamma_{n}(a)=-\partial^{n}_{L_{q}}G(a,L_{q})|_{L_{q}=0}$, this recursion reads 
\begin{equation}\label{RGrec}
 \gamma_{n+1}(a) = [\beta(a)\partial_{a}-\gamma(a)]^{n}\gamma(a) = \Rn^{n}\gamma(a) \hs{2} \mbox{('RG recursion')}
\end{equation}
and will be referred to as \emph{renormalisation group (RG) recursion} throughout this work. It takes the same form in QED for the photon's 
anomalous dimension. 

Because the rainbow approximation's beta function vanishes, $\beta(a)=0$ \cite{Krei06}, the RG operator is a simple multiplication 
operator $\Rn=-\gamma(a)$ and the RG recursion yields
\begin{equation}\label{rec}
\gamma_{n}(a)=(-1)^{n-1}\gamma(a)^{n},
\end{equation}
which, when plugged into (\ref{RBMellad}), produces the fixed point equation, ie the DSE for the anomalous dimension 
\begin{equation}\label{RBMellad'}
\gamma(a) = -\frac{a}{2} \left[ 1 - \sum_{n=1}^{\infty} (-1)^{n} \ \gamma_{n}(a) \right] 
= -\frac{a}{2} \left[ 1 + \sum_{n=1}^{\infty} \ \gamma(a)^{n} \right] \hs{1} (\mbox{rainbow})
\end{equation}
and hence (\ref{Mellineq}). The scaling solution (\ref{ansatz}) is then obtained by directly integrating the RG equation (\ref{RGeq}) with $\Rn=-\gamma(a)$. 

\subsection{Nonlinear DSE: Kilroy approximation}\label{subsec:Kil}
The above rainbow DSE (\ref{RBY}) falls into the class of so-called \emph{linear} DSEs. This denomination comes from the fact that the self-energy does not appear in 
squared form or in higher powers. As soon as nontrivial powers are involved, we speak of \emph{nonlinear} DSEs. An example is the so-called
\emph{Kilroy Dyson-Schwinger equation}\footnote{This was coined by David Broadhurst (see internet for the phrase 'Kilroy was here' to get the idea).} 
which in blob-diagrammatical form reads
\begin{equation}\label{KiloyYb}
 \graph{0}{-0.35}{0.45}{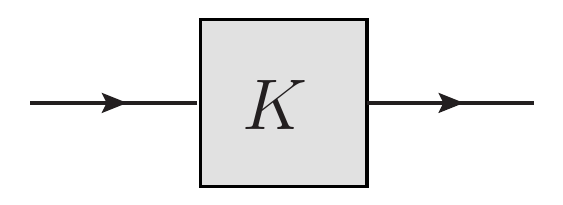}{0.1} =  \graph{0.1}{-0.35}{0.45}{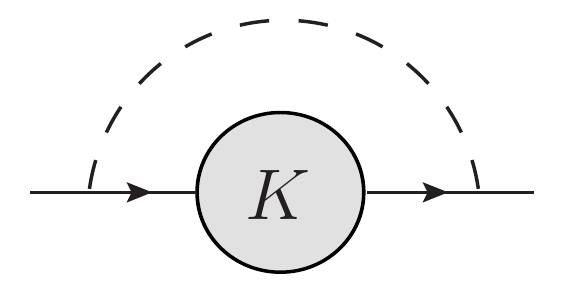}{0.2} .
\end{equation}
The round blob on the right represents the full propagator. This DSE describes an approximation for the self-energy in which graphs of the form
\begin{equation}
 \Gamma = \graph{0}{-0.2}{0.3}{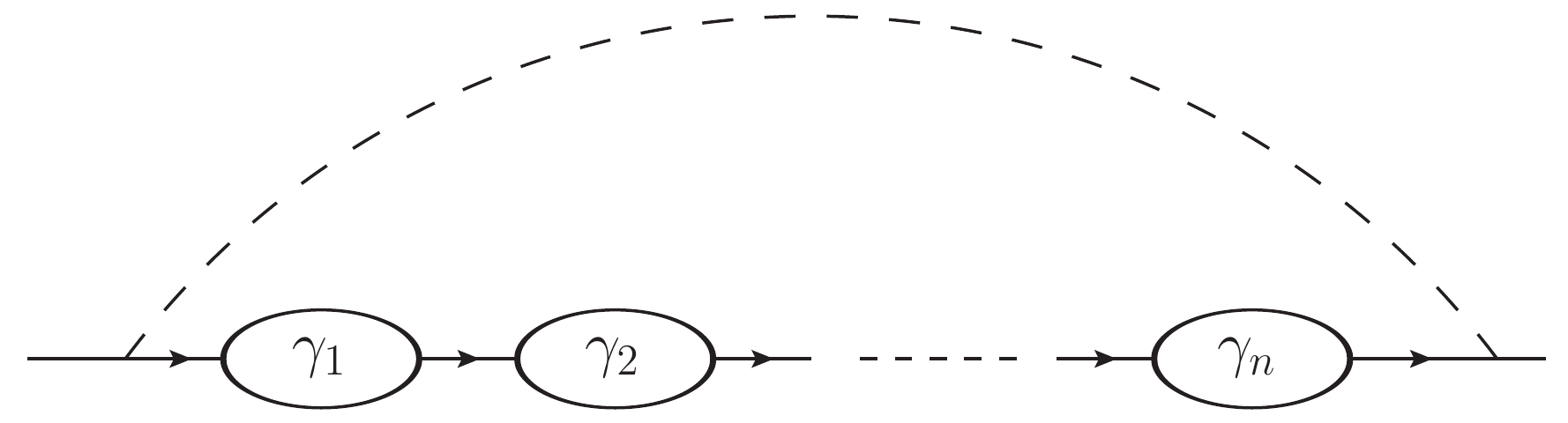}{0.2}  \hs{2} n \geq 0
\end{equation}
emerge. In it, any subgraph $\gamma_{j}$ is either a rainbow graph or a graph one obtains by consecutively inserting any sequence of rainbow graphs into 
a rainbow graph, which means that any chainings and nestings of rainbow subgraphs are involved, eg 
\begin{equation}
\graph{0}{-0.3}{0.3}{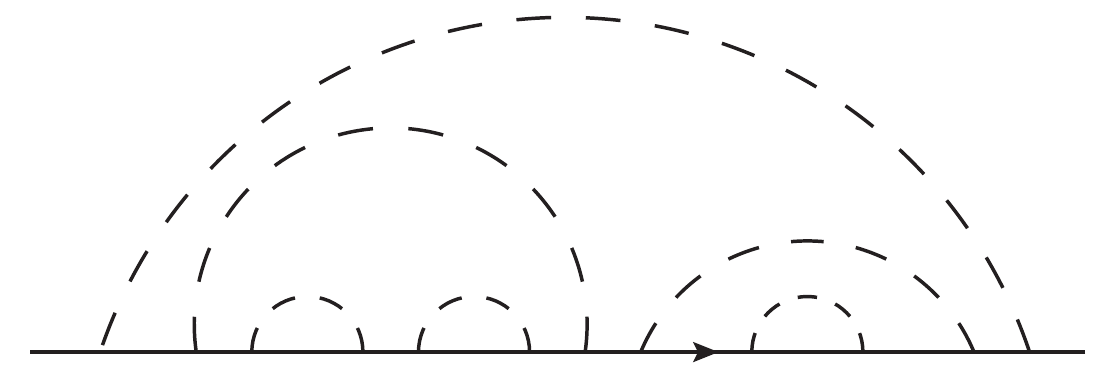}{0.1} \hs{2} (\mbox{Kilroy graph}).
\end{equation}
See \cite{BroK01} for some more examples. Notice that with this approximation we already enter the realm of nontrivial DSEs: the Kilroy approximation 
harbours UV renormalons, the first of which is brought about by the subseries
\begin{equation}
\left(\graph{0}{-0.35}{0.45}{Kiloyl}{0}\right)_{\mbox{\tiny ren}} :=  \graph{0.02}{0}{0.4}{1yup1}{0.1}  
+ \sum_{n \geq 1} \graph{0.02}{-0.35}{0.4}{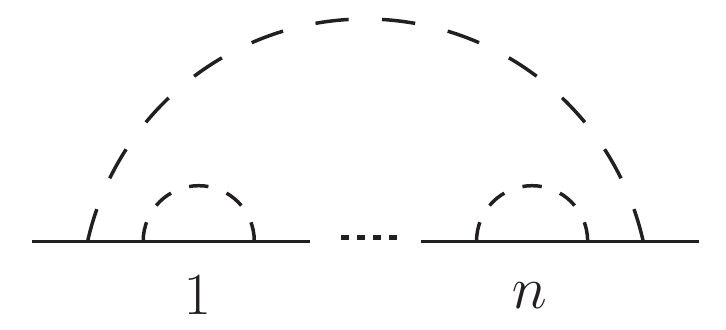}{0.1} .
\end{equation}
We will compute it in a moment, when it suits us. In its analytical form, the Kilroy DSE in Yukawa theory is given by
\begin{equation}\label{KilroyDSE}
-i\Sigma_{K}(q) = (-ig)^{2} \int \frac{d^{4}k}{(2\pi)^{4}} \frac{i}{\slashed{k} + i \epsilon} \left[ 1 
+ i \Sigma_{K}(k)\frac{i}{\slashed{k} + i \epsilon } \right]^{-1} \frac{i}{(q-k)^{2} + i \epsilon} \ .
\end{equation}
After Wick rotating and some algebra, one finds that the analogue of (\ref{RBYR2}) is given by 
\begin{equation}\label{KilroyDSE'}
G(a,L) = 1 + \frac{a}{2} \lim_{\rho \rightarrow 0} G(a,-\partial_{\rho})^{-1} \int \frac{d^{4}k_{E}}{\pi^{2}} 
\left\{ \frac{1}{k_{E}^{2}(q_{E}-k_{E})^{2}} - \frac{1}{k_{E}^{2}(\ti{q}_{E}-k_{E})^{2}} \right\} \left( \frac{k^{2}_{E}}{\mu^{2}} \right)^{-\rho} ,
\end{equation}
which presents a real jump in complexity as compared to the rainbow case. 
We use the Mellin transform of the Kilroy DSE to cast this in the form \cite{KrY06}
\begin{equation}\label{KilroyMell}
G(a,L_{q}) = 1 +  a \lim_{\rho \rightarrow 0} G(a,-\partial_{\rho})^{-1}  ( e^{-\rho L_{q}} - 1 ) F(\rho) \hs{0.8}  (\mbox{Kilroy DSE in Mellin guise}).
\end{equation}
Not surprisingly, this equation cannot be solved by a simple scaling ansatz like in the rainbow case.
However, taking the first derivative of (\ref{KilroyMell}) with respect to $L_{q}$ and then setting $L_{q}=0$ gives the desired DSE for the anomalous 
dimension:
\begin{equation}\label{Kilan}
 \gamma_{1}(a) = -a  \lim_{\rho \rightarrow 0} G(a,-\partial_{\rho})^{-1} (-\rho) F(\rho) 
 = a  \lim_{\rho \rightarrow 0} G(a,-\partial_{\rho})^{-1} \frac{1}{2(1-\rho)}, 
\end{equation}
where we have used the Mellin transform (\ref{MellinT}). 
The crucial difference to the rainbow case is that $G(a,-\partial_{\rho})$ shows up inverted making the differential operator extra nasty:
\begin{equation}\label{GreenD}
\G_{\rho}:= G(a,-\partial_{\rho})^{-1} = 1 + \sum_{r \geq 1} \left(\sum_{n \geq 1} \frac{\gamma_{n}(a)}{n!}(-\partial_{\rho})^{n} \right)^{r}
= 1 + \sum_{r \geq 1} \sum_{n \geq r} (\gamma_{\bullet}^{\star r})_{n} (-\partial_{\rho})^{n} 
\end{equation}
where the expression 
\begin{equation}\label{conv2}
 (\gamma^{\star r}_{\bullet})_{n} := \sum_{n_{1} + \dotsc + n_{r}=n} \frac{\gamma_{n_{1}}(a)}{n_{1}!} \dotsc \frac{\gamma_{n_{r}}(a)}{n_{r}!}  
 \hs{2} r,n \geq 1
\end{equation}
is a shorthand notation (the sum is only over $n_{1}, \dotsc, n_{r} \geq 1$). We set $(\gamma^{\star 0}_{\bullet})_{0} := 1$ and 
$(\gamma^{\star 0}_{\bullet})_{n}=0$ for $n \geq 1$. The motivation for this way of writing the sum is that for a fixed coupling, 
$n \mapsto \gamma_{n}(a)$ is a real-valued function on $\N$. 
The sum in (\ref{conv2}) can then rightfully be seen as the value of an $r$-fold convolution product at the argument $n \in \N$.

The task of computing the rhs of (\ref{Kilan}) is now less trivial, but still a nice exercise which lets us arrive at
\begin{equation}\label{Kilanom1}
 \gamma(a) = a \sum_{r \geq 0,n \geq r} C_{n} (\gamma^{\star r}_{\bullet})_{n}  
 = C_{0} a + a \sum_{r \geq 1,n \geq r} C_{n} (\gamma^{\star r}_{\bullet})_{n}  \hs{2} (\mbox{Kilroy DSE}) ,
\end{equation}
with coefficients $C_{n}:=(-1)^{n}\frac{n!}{2}$. The rainbow DSE (\ref{RBMellad}) can also be shoehorned into this form, 
where the sum in this case just extends over two index values of $r$,
\begin{equation}\label{Rain}
 \gamma(a) = a \sum_{r=0}^{1} \sum_{n \geq r} C_{n} (\gamma^{\star r}_{\bullet})_{n} 
 = C_{0} a + a \sum_{n \geq 1} C_{n} (\gamma^{\star 1}_{\bullet})_{n}  = C_{0} a + a \sum_{n \geq 1} C_{n} \gamma_{n}   \hs{0.2} (\mbox{rainbow DSE}),
\end{equation}
with in this case $C_{0}=-\frac{1}{2}$ and $C_{n}=(-1)^{n}\frac{1}{2}$ for $n\geq 1$. However, in full glory, (\ref{Kilanom1}) reads
\begin{equation}\label{Kilanom}
 \gamma(a) = \frac{a}{2} \left[ 1
 + \sum_{r \geq 1} \sum_{n \geq r} (-1)^{n} \ \binom{n}{n_{1}, \dotsc ,n_{r} } \sum_{n_{1} + \dotsc + n_{r}=n} 
 \gamma_{n_{1}}(a) \dotsc \gamma_{n_{r}}(a)\right]  \hs{1} 
(\mbox{Kilroy}) ,
\end{equation}
where $\gamma(a)=\gamma_{1}(a)$ and the innermost sum is over all $n_{1}\geq 1, \dotsc, n_{r} \geq 1$.
Combining this with the RG recursion $\gamma_{n}(a)=\Rn^{n-1}\gamma(a)$, 
\begin{equation}\label{KiladDSE}
 \gamma(a) = \frac{a}{2} \left[ 1
 + \sum_{r \geq 1} \sum_{n \geq r} (-1)^{n} \ \binom{n}{n_{1}, \dotsc ,n_{r}} \sum_{n_{1} + \dotsc + n_{r}=n} 
 \Rn^{n_{1}-1}(\gamma(a)) \dotsc \Rn^{n_{r}-1}(\gamma(a)) \right]  ,
\end{equation}
brings out this equation's character as a fixed point equation for the anomalous dimension. If we now compare this with the rainbow case in
(\ref{RBMellad}) and (\ref{RBMellad'}), we see the dramatic change: the rhs of (\ref{KiladDSE}) has an infinite number of differential operators 
hidden in the powers of the RG operator $\Rn$, whereas the rainbow DSE has none. (\ref{KiladDSE}) will be the key equation studied in the transseries 
setting. 

To compute the first UV renormalon, we truncate the differential operator in (\ref{GreenD}) 
\begin{equation}
1 + \sum_{r \geq 1} (-c_{1}a \partial_{\rho})^{r}, 
\end{equation}
$c_{1}$ being the first-order coefficient in $\gamma(a)=c_{1}a + \mathcal{O}(a^{2})$, and replace it in (\ref{Kilan}):
\begin{equation}\label{KilanR}
 \gamma_{\text{ren}}(a) =  a  \lim_{\rho \rightarrow 0} \left( 1 + \sum_{r \geq 1} (-c_{1} a \partial_{\rho})^{r} \right) \rho F(\rho) 
 = a + \sum_{r \geq 1}(- c_{1})^{r} r! a^{r}.
\end{equation}
The Borel transform of this series is given by $\hat{\gamma}_{\text{ren}}(\zeta)=(1+c_{1}\zeta)^{-1}$ and has a pole at $\zeta_{*}=-c_{1}^{-1}$. 
This pole is the first UV renormalon of the Kilroy series. Unless this divergence is cancelled in some mysterious way by the remainder of the
anomalous dimension's perturbation series, the anomalous dimension of the Kilroy approximation has a divergent perturbation series! 
This sets it apart from the results of the rainbow approximation and makes it 'more physical'. But this first and all higher renormalons aside, 
considering the \emph{growth} in the number of Kilroy diagrams with loop order alone \cite{BroK00}, this is exactly what one would expect.

\subsection{ODE for Kilroy} However, the good news is that one can do some more to tackle the Kilroy case \cite{KrY06}: differentiating (\ref{KilroyMell}) 
twice with respect to the parameter $L_{q}$ and then setting it to zero leads to 
\begin{equation}
  \gamma_{2}(a) = -a  \lim_{\rho \rightarrow 0} \G_{\rho} \rho^{2} F(\rho) = a  \lim_{\rho \rightarrow 0} \G_{\rho} \frac{- \rho}{2(1-\rho)}
\end{equation}
where we have used the explicit expression of the Mellin transform in (\ref{MellinT}). Comparing this to (\ref{Kilan}) urges us to add both expressions which produces 
$\gamma_{1}(a) + \gamma_{2}(a) = \frac{a}{2}$. Combined with the RG recursion $\gamma_{2}(a)=\Rn \gamma_{1}(a)$, this gives an ODE,
\begin{equation}\label{Kilode}
 \frac{a}{2} = \gamma(a) + \Rn \gamma(a) = \gamma(a) + \gamma(a)(2a \partial_{a}-1)\gamma(a),  
\end{equation}
where $\Rn = 2 a\gamma(a)\partial_{a} - \gamma(a)$ is the RG operator of the Kilroy approximation whose beta function is given by 
$\beta(a)=2a \gamma(a)$ \cite{Y11}. 

Broadhurst and Kreimer have investigated the Kilroy DSE (\ref{KilroyDSE}) for both Yukawa and $(\varphi^{3})_{6}$ theory in \cite{BroK01}, 
where this ODE has also been derived, albeit down a completely different route and slightly differing conventions. The authors find that the anomalous 
dimension $\gamma(a)$ of the Kilroy approximation (\ref{KilroyDSE}) satisfies the implicit 
equation\footnote{Nowadays, there is modern computer algebra software like Maple 16 which turns (\ref{Kilode}) into (\ref{David}).}
\begin{equation}\label{David}
 \sqrt{\frac{a}{\pi}} e^{-Z(a)} = 1 + \mbox{erf}(Z(a)),
\end{equation}
where $Z(a)=(\gamma(a)-1)/\sqrt{a}$ and $\mbox{erf}(x)$ is the famous error function. They solved this
equation\footnote{The authors used different conventions, one has to replace $a \rightarrow \frac{a}{2}$ and 
$\gamma(a) \rightarrow -\frac{1}{2} \gamma(a)$ to find agreement.} numerically for $\gamma(a)$ by an algorithm of the Newton-Raphson type \cite{BroK01}. 
The Kilroy model may therefore be seen as exactly solved. 

Our contention that its perturbation series must be divergent is supported by Broadhurst and Kreimer 
in \cite{BroK00} by the growth of the Kilroy model's coefficients up to 30 loops, which turned out to behave as
\begin{equation}
 c_{n} \sim 2^{n-1} \Gamma(n+1/2).
\end{equation}

\subsection{Approximative DSE for the photon} 
In terms of blob diagrams, the DSE for the photon's self-energy in quantum electrodynamics (QED) is given by 
\begin{equation}\label{phoDSE}
  \graph{0}{-0.1}{0.17}{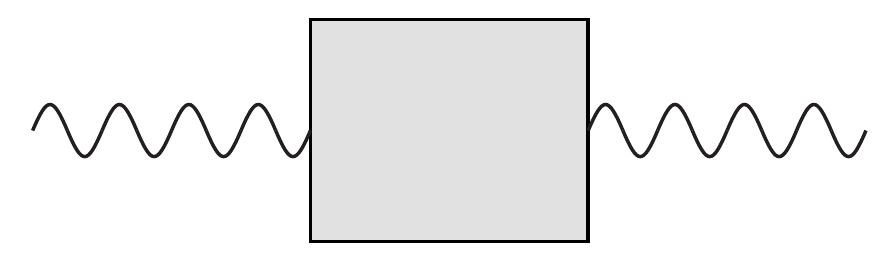}{0.2} =   \graph{0.2}{-0.5}{0.15}{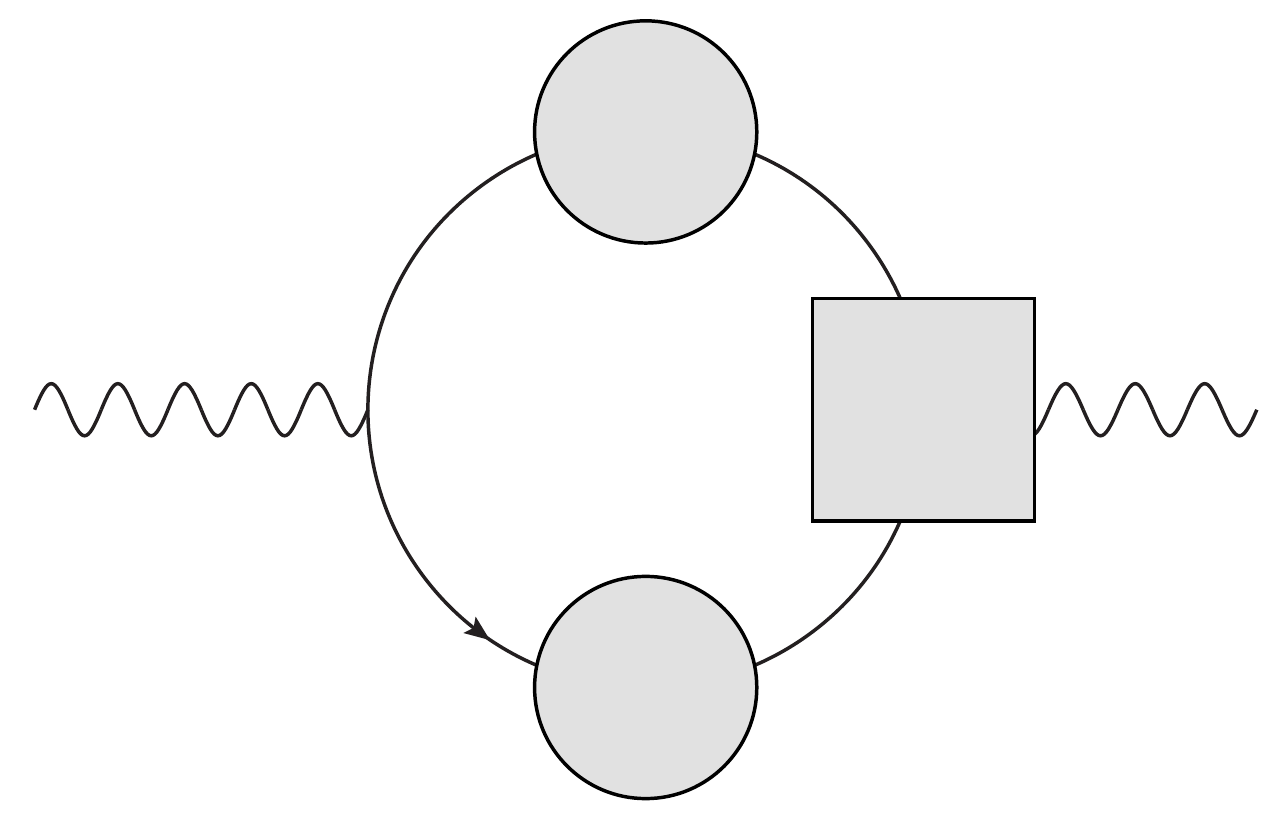}{0.2}   \hs{2.7} (\mbox{photon self-energy}),
\end{equation}
in which the self-energy is hidden on the rhs inside the fermion and vertex blobs. Their DSEs are 
\begin{equation}\label{ferDSE}
 \graph{0}{-0.12}{0.16}{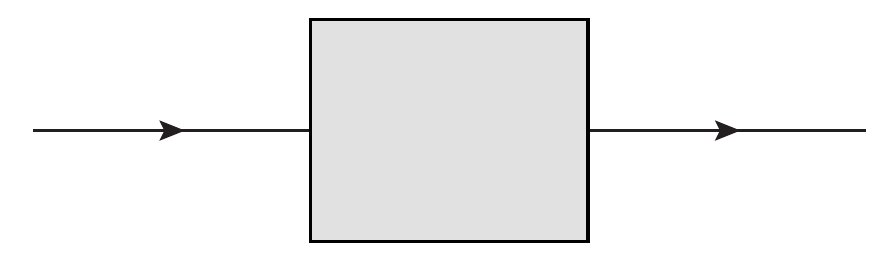}{0.2} =   \graph{0.2}{-0.12}{0.16}{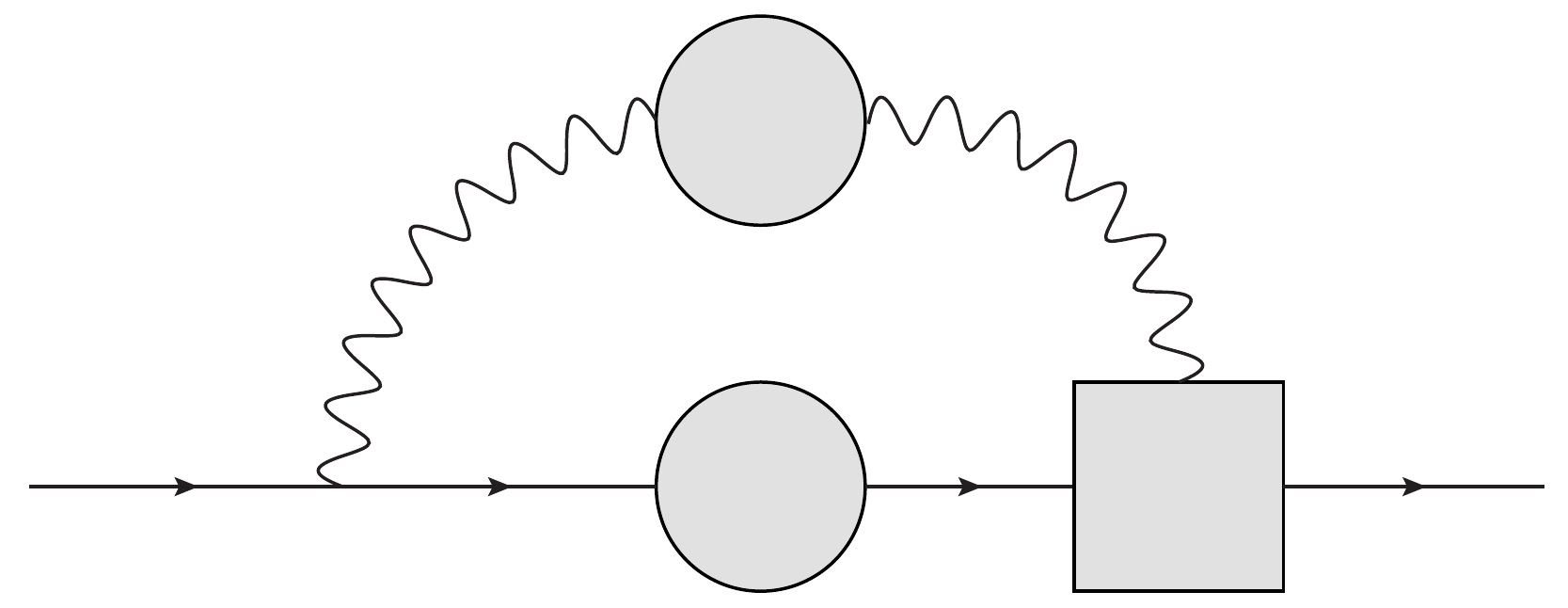}{0.2}   \hs{2} (\mbox{fermion self-energy})
\end{equation}
and 
\begin{equation}\label{verDSE}
\graph{0}{-0.5}{0.2}{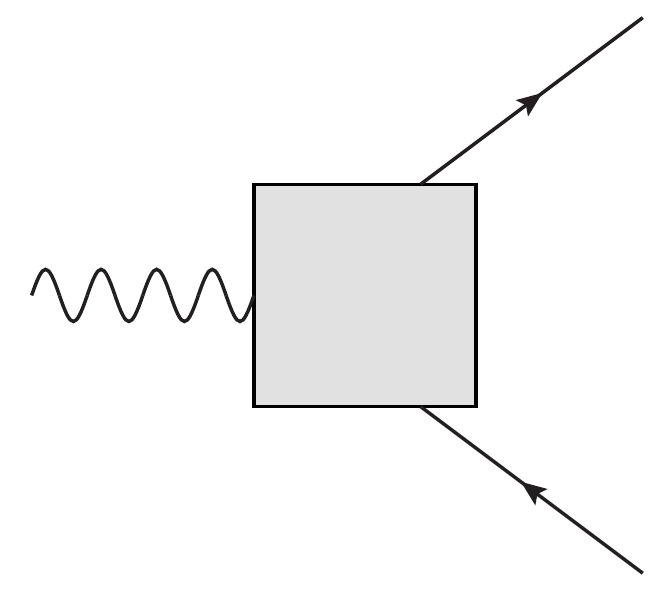}{0}  =   \graph{0.1}{-0.4}{0.2}{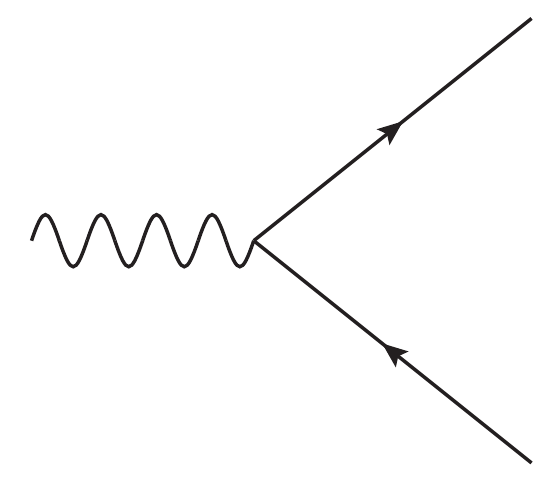}{0.1}  +  \graph{0.1}{-0.6}{0.2}{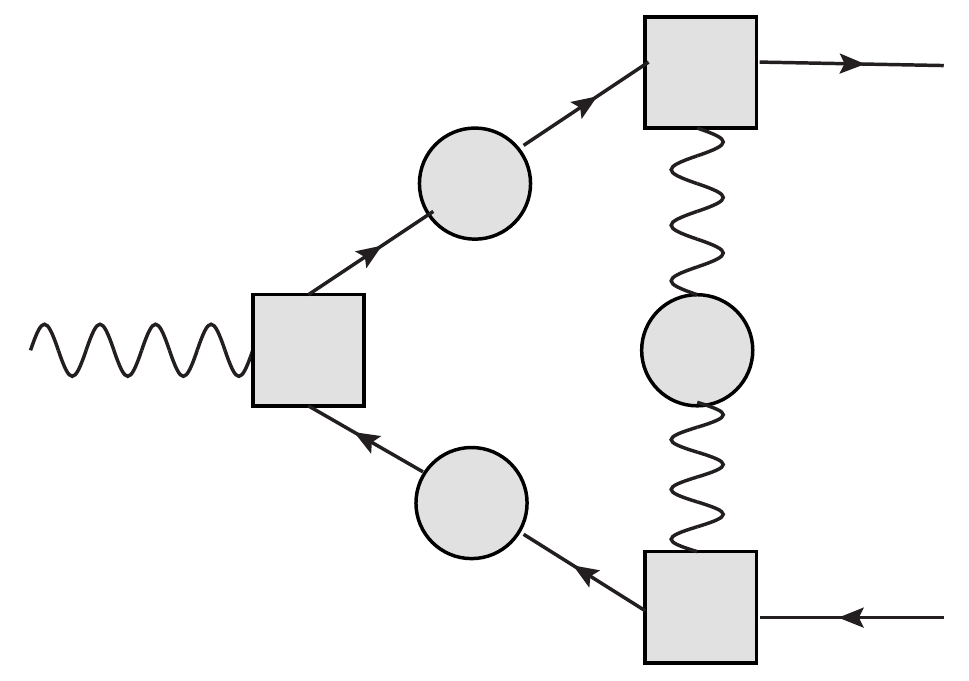}{0.1}  +  \graph{0.1}{-0.6}{0.2}{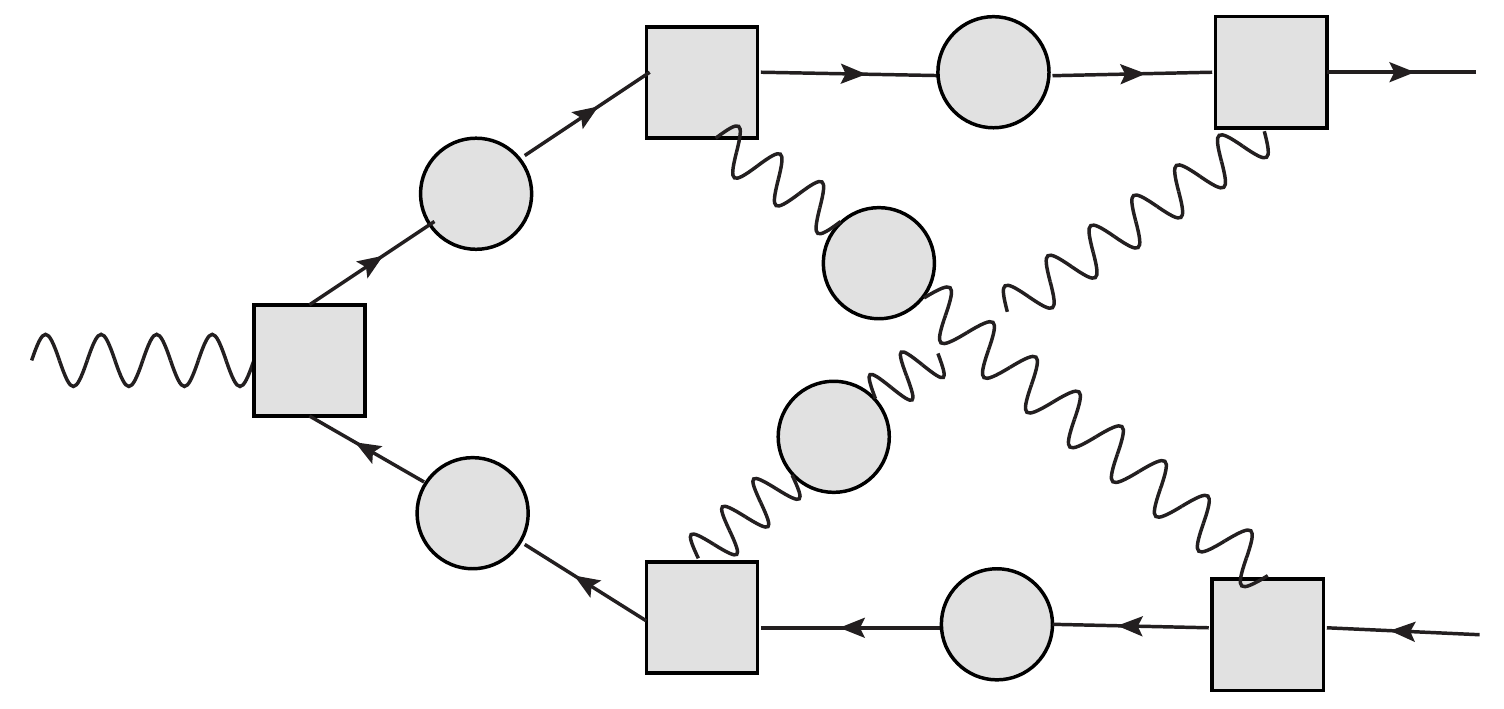}{0.2}
+  \ \ \dotsc      \hs{0.5} (\mbox{vertex function}) .                                                                  
\end{equation}
We abstain from writing these out in their analytical form, the interested reader is referred to the classical source \cite{BjoDre65}. The DSE for 
the vertex (\ref{verDSE}) cannot be read as a nonperturbative equation unless it is truncated. This is the price to pay if one tries to cut the three top 
pieces off the infinite tower that DSEs in fact are \cite{BjoDre65}. The trouble with the infinite skeleton expansion (\ref{verDSE}) is that it comprises
by itself a divergent series even without inserted blobs.     

However, it is moreover possible to decouple the photon DSE (\ref{phoDSE}) from all the other DSEs courtesy 
of the Ward identity for the renormalisation constants, ie $Z_{1}=Z_{2}$ (charge renormalisation = electron wave function renormalisation)\cite{Wa50}. 
Before we elaborate on this aspect, let us first note that the decoupling is achieved combinatorially by constructing the entire perturbation 
series using only photon propagator corrections as follows. 

Again, as explained in the case of the vertex DSE (\ref{verDSE}), the price to pay is that
the single skeleton in (\ref{phoDSE}) is replaced by an infinite sum of skeletons, where the skeleton diagrams are the photon diagrams 
of \emph{quenched QED},
\begin{equation}\label{pho}
\graph{0}{-0.25}{0.3}{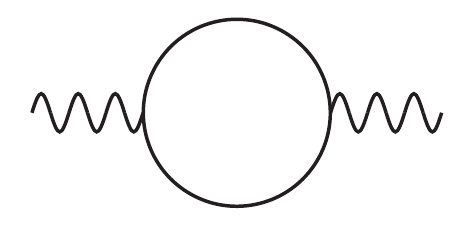}{0} + \graph{0}{-0.25}{0.3}{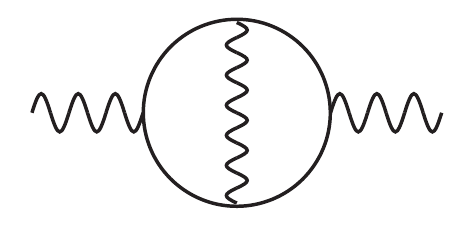}{0} + \graph{0}{-0.25}{0.3}{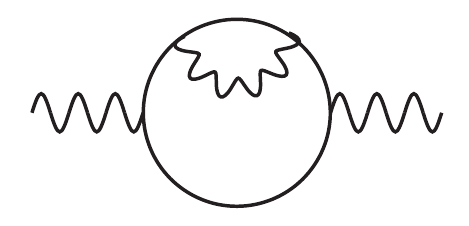}{0} + \graph{0}{-0.25}{0.3}{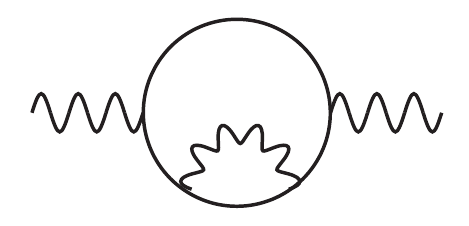}{0}
+  \graph{0}{-0.25}{0.3}{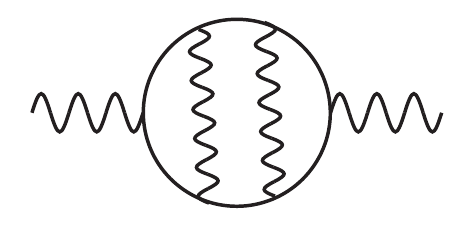}{0} + \graph{0}{-0.25}{0.3}{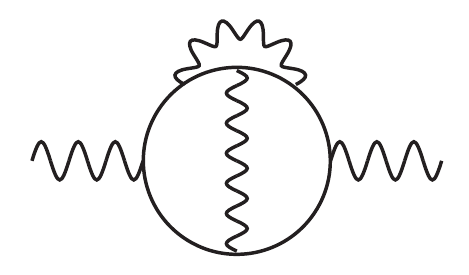}{0} + \graph{0}{-0.4}{0.3}{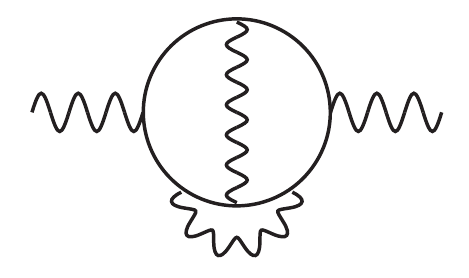}{0} + \  \dotsc 
\end{equation}
with the defining property of featuring only bare photon lines. These photon lines are then dressed with full photon propagators and one arrives at 
\begin{equation}\label{phoDSE1}
\graph{0}{-0.1}{0.17}{DSEphol}{0} = \graph{0}{-0.25}{0.3}{ph0}{0} + \graph{0}{-0.25}{0.3}{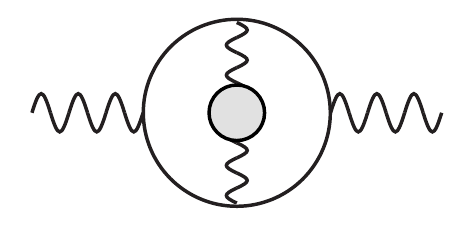}{0} + \graph{0}{-0.25}{0.3}{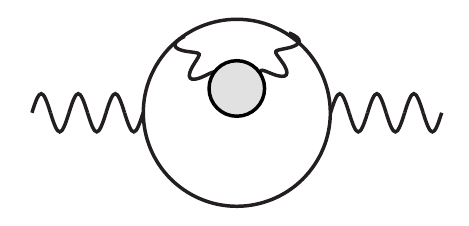}{0} 
+ \graph{0}{-0.25}{0.3}{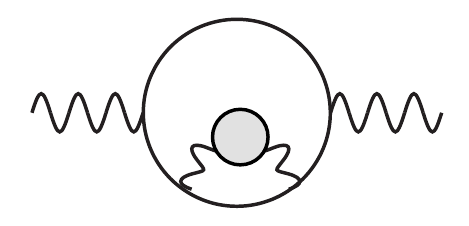}{0} +  \graph{0}{-0.25}{0.3}{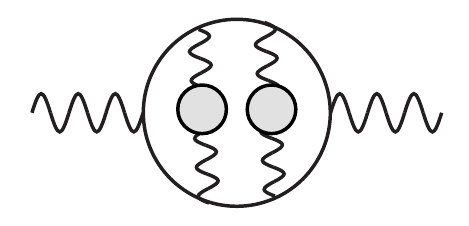}{0} + \graph{0}{-0.25}{0.3}{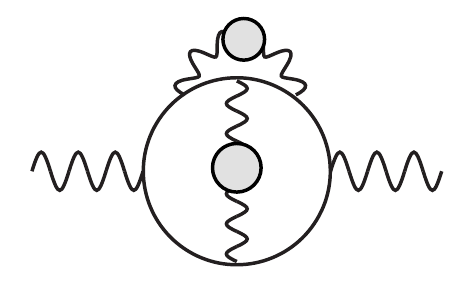}{0}  +  \  \dotsc 
\end{equation}
which combinatorially provides all contributions to the photon's self-energy. If we truncate this skeleton series, we can rightfully interprete it as 
a nonperturbative equation for the self-energy. But this comes at the price of having to acknowledge that it is an approximation. 

As with the vertex series in (\ref{verDSE}), things become somewhat fuzzy when we pass over to the limit of an infinite number of skeletons because the 
skeleton expansion (\ref{pho}) is itself a divergent series! 
Let us have a closer look at this and write (\ref{phoDSE1}) in the form
\begin{equation}\label{Ske}
\sum_{\ell \geq 1} S_{\ell}[\Pi(\alpha)](Q) \alpha^{\ell} = \graph{0}{-0.25}{0.3}{ph0}{0} + \graph{0}{-0.25}{0.3}{ph1d}{0} 
+ \graph{0}{-0.25}{0.3}{ph1od}{0} 
+ \graph{0}{-0.25}{0.3}{ph1ud}{0} +  \graph{0}{-0.25}{0.3}{ph2d}{0} + \graph{0}{-0.25}{0.3}{ph2od}{0}  +  \  \dotsc
\end{equation}
where $S_{\ell}[\Pi(\alpha)](Q)$ is the nonlinear (!) integral operator of the $\ell$-th skeleton graph in (\ref{phoDSE1}) with external kinematics $Q$, 
mapping the self-energy $\Pi(\alpha)$ of the photon to the 'blobbed' expression. If we set $\Pi(\alpha)=0$ in all these terms, we get the perturbation series of the self-energy in 
quenched QED:
\begin{equation}\label{Ske0}
\sum_{\ell \geq 1} S_{\ell}[0](Q) \alpha^{\ell} = \graph{0}{-0.25}{0.3}{ph0}{0} + \graph{0}{-0.25}{0.3}{ph1}{0} + \graph{0}{-0.25}{0.3}{ph1o}{0} 
+ \graph{0}{-0.25}{0.3}{ph1u}{0} + \graph{0}{-0.25}{0.3}{ph2}{0} + \dotsc
\end{equation}
And here is the obstruction: this perturbation series is a divergent series itself! At the very least, it has renormalons from fermion propagator corrections. It therefore needs a nonperturbative completion, say some transseries representation. 
The trouble is that there is no self-consistent equation and cannot be any to help fix such representation: any attempt to find a single Dyson-Schwinger equation for 
the quenched photon is nipped in the bud by the taboo to insert photon blobs into bare photon lines. 

However, that is not to say (\ref{phoDSE1}) is a hopeless case and can never be given a meaning. Not quite so, in fact, the theory of multisummability is
actually more ambitious than what one might think: \cite{Ba00} treats formal power series with coefficients in a Banach algebra! Alas, these 
coefficients do not depend on the coupling whereas those in the DSE (\ref{phoDSE1}) do.
So, the truth of the matter is that this equation acquires a perturbative character in the limit of infinitely many skeletons. And to retain its 
nonperturbative value, we will think of it as truncated at some arbitrary large skeleton loop order and acknowledge it to be an interesting approximation. 
But, as explained before, we may think of (\ref{phoDSE1}) as a sequence of DSEs in which
\begin{equation}
 \graph{0}{-0.1}{0.17}{DSEphol}{0} = \sum_{\ell = 1}^{N} S_{\ell}[\Pi(\alpha)](Q) \alpha^{\ell} 
\end{equation}
is the $N$-th element. Then, presuming that each DSE in this sequence has a solution, we might view the limit of this sequence of solutions as the 
sought-after self-energy.

Besides, the combinatorics of Feynman diagrams in (\ref{phoDSE1}) is fine: the decoupled photon series produces all contributions and is 
therefore, let us say, combinatorially self-consistent. 
The Ward identity is crucial in this decoupling procedure. The photon DSE (\ref{phoDSE1}) would make little sense without it, even when truncated. 
The point is that if we choose a renormalisation scheme such that $Z_{1}=Z_{2}$ holds true, the sum of all contributions in quenched QED 
at each given loop order is primitive, ie needs only one single subtraction. For example, at two-loop level, the sum of graphs
\begin{equation}\label{phot}
\Gamma_{2} = \graph{0}{-0.25}{0.3}{ph1}{0} + \graph{0}{-0.25}{0.3}{ph1o}{0} + \graph{0}{-0.25}{0.3}{ph1u}{0} 
\end{equation}
requires only one counterterm. Although each individual Feynman graph needs more than one subtraction on account of its subdivergence, all individual 
counterterms curing these subdivergences cancel each other out and only a single counterterm for the overall divergence is 
necessary\footnote{As regards the origin of the photon DSE (\ref{phoDSE1}), the author tried to track down the person who first came up with this idea, but failed, although
a number of people made it onto the shortlist with Donald Yennie on top. The author himself learnt it from his supervisor Dirk Kreimer (see also the 
acknowledgement section).}.

However, in what follows, we will denote the sum of all skeletons at loop order $\ell$ by $\Gamma_{\ell}$. 
To write (\ref{phoDSE1}) in terms of Mellin transforms, first note that all skeleton graphs at loop order $\ell \geq 1$ have $(\ell-1)$ internal photon 
lines. For $\ell \geq 2$, let us augment each internal photon propagator with a convergence factor
\begin{equation}
 \left( \frac{k^{2}_{j}}{\mu^{2}} \right)^{-\rho_{j}} = e^{-\rho_{j} L_{k_{j}}},
\end{equation}
where $k_{j}$ is the Euclidean momentum flowing through the $j$-th photon line. From now on, we tacitly take all momenta to be Euclidean. 
Next, we denote by 
\begin{equation}
\int \Int_{\ell}(\rho_{1}, \dotsc, \rho_{\ell-1};q)
\end{equation}
the so-regularised scalar Feynman integral associated to the skeleton sum $\Gamma_{\ell}$ with external momentum $q$, ie all contributing Feynman 
integrands are brought under one integral sign (the coupling constant is excluded from this skeleton Feynman integral). The skeleton is 
renormalised by the subtraction
\begin{equation}
\int \Int_{\ell}(\rho_{1}, \dotsc, \rho_{\ell-1};q) - \int \Int_{\ell}(\rho_{1}, \dotsc, \rho_{\ell-1};\ti{q}) 
\end{equation}
with reference momentum $\ti{q}$ such that $\ti{q}^{2}=\mu^{2}$. Acting $G(\alpha,-\partial_{\rho_{j}})^{-1}$ on this expression inserts a full 
renormalised propagator, by means of the 'mechanism' 
\begin{equation}\label{insert}
G(\alpha,-\partial_{\rho_{j}})^{-1} \left( \frac{k^{2}_{j}}{\mu^{2}} \right)^{-\rho_{j}} 
= \left( \frac{k^{2}_{j}}{\mu^{2}} \right)^{-\rho_{j}} G(\alpha,L_{k_{j}})^{-1} ,
\end{equation}
where $\alpha$ is the fine-structure constant of QED which we naturally choose as coupling parameter and $G(\alpha,L_{k_{j}})^{-1}$ is the analogue of 
(\ref{GreenD}). To be more precise, what we mean is the form factor of the photon's self-energy: let 
\begin{equation}\label{Pro}
\Pi_{\mu \nu}(q) = \frac{g_{\mu \nu} - q_{\mu}q_{\nu}/q^{2}}{q^{2}[1-\Pi(\alpha,q^{2}/\mu^{2})]} 
\end{equation}
be the transversal part of the full renormalised photon propagator in massless QED. Then we have $G(\alpha,L_{q})=1-\Pi(\alpha,q^{2}/\mu^{2})$.
However, to make the crucial step towards a formulation using Mellin transforms, we define the Mellin transform $F_{\ell}$ of the skeleton 
$\Gamma_{\ell}$ by
\begin{equation}
F_{\ell}(\rho_{1}, \dotsc, \rho_{\ell-1}) := \left. \int \Int_{\ell}(\rho_{1}, \dotsc, \rho_{\ell-1};q) \right|_{q^{2}=1},
\end{equation}
analogously to (\ref{MellinT}). One can then write \cite{Y11}
\begin{equation}
[e^{-(\rho_{1}+\dotsc+\rho_{\ell-1})L_{q}} - 1 ]F_{\ell}(\rho_{1}, \dotsc, \rho_{\ell-1}) 
= \int \Int_{\ell}(\rho_{1}, \dotsc, \rho_{\ell-1};q) - \int \Int_{\ell}(\rho_{1}, \dotsc, \rho_{\ell-1};\ti{q}) 
\end{equation}
and finally produce the rhs of (\ref{phoDSE1}) by applying the differential operators (\ref{insert}) to dress the bare internal photon lines and obtain
\begin{equation}\label{phodse}
G(\alpha,L_{q})=1 - \alpha A_{0}L_{q} + \lim_{\{ \rho \} \rightarrow 0} \sum_{\ell \geq 1} \alpha^{\ell+1}  \G_{\rho_{1}} \dotsc 
\G_{\rho_{\ell}}
\left[e^{-(\rho_{1}+\dotsc+\rho_{\ell})L_{q}} - 1 \right]F_{\ell+1}(\rho_{1}, \dotsc, \rho_{\ell}) ,
\end{equation}
where the second term on the rhs is the first-loop contribution to the self-energy in momentum scheme ($A_{0}=1/3\pi$, see \cite{GoKLaS91}) and 
$\G_{\rho_{j}}$ are the formal differential operators 
\begin{equation}
 \G_{\rho_{j}}:=  G(\alpha,-\partial_{\rho_{j}})^{-1} = 1 + \sum_{r \geq 1} \sum_{n \geq r} \left( \sum_{n_{1} + \dotsc + n_{r}=n} 
 \frac{\gamma_{n_{1}}(\alpha)}{n_{1}!} \dotsc \frac{\gamma_{n_{r}}(\alpha)}{n_{r}!} \right) (-\partial_{\rho_{j}})^{n} ,
\end{equation}
defined as in (\ref{GreenD}). The shorthand notation of (\ref{conv2}) puts us in a position to write the above product of differential operators in 
the more compact form  
\begin{equation}
\G_{\rho_{1}} \dotsc \G_{\rho_{\ell}} = \sum_{r_{1}\geq 0,n_{1} \geq r_{1}} \dotsc \sum_{r_{\ell}\geq 0,n_{\ell} \geq r_{\ell}}
(\gamma^{\star r_{1}}_{\bullet})_{n_{1}}  \dotsc (\gamma^{\star r_{\ell}}_{\bullet})_{n_{\ell}} 
(-\partial_{\rho_{1}})^{n_{1}}\dotsc (-\partial_{\rho_{\ell}})^{n_{\ell}} 
\end{equation}
and extract the DSE for the anomalous dimension out of (\ref{phodse}): comparing the first-order 
term in $L_{q}$ on both sides, we find
\begin{equation}\label{phodse1}
 \gamma(\alpha) = \alpha A_{0} + \sum_{\ell \geq 1} \alpha^{\ell+1} \sum_{r_{1}\geq 0, n_{1} \geq r_{1}} \dotsc \sum_{r_{\ell}\geq 0, n_{\ell} \geq r_{\ell}} 
  C_{(n_{1}, \dotsc, n_{\ell})} (\gamma^{\star r_{1}}_{\bullet})_{n_{1}}  \dotsc (\gamma^{\star r_{\ell}}_{\bullet})_{n_{\ell}}  \hs{0.2} (\mbox{photon}) ,
\end{equation}
where $\gamma(\alpha)=\gamma_{1}(\alpha)$ and the numbers
\begin{equation}
 C_{(n_{1}, \dotsc, n_{\ell})} :=\lim_{\{\rho\} \rightarrow 0} (-\partial_{\rho_{1}})^{n_{1}} (-\partial_{\rho_{\ell}})^{n_{\ell}} 
 (\rho_{1} + \dotsc + \rho_{\ell})F_{\ell +1}(\rho_{1}, \dotsc, \rho_{\ell}) 
\end{equation}
can be assumed to exist because the Mellin transforms $F_{\ell +1}$ have only simple first order poles. The limit may have to be taken with extra care, 
though.

\section{Brief introduction to grid-based transseries}\label{sec:ReTra} 		  
Before we apply the framework of transseries to the nonperturbative formulae derived in the previous section, we will give a brief mathematical 
introduction to transseries in the spirit of \cite{Ed09, vH06}, which have been our main sources for the material garnered here. 
Readers will find the details including proofs there. 
For convenience we adopt the habit of mathematicians to replace the coupling constant $x \in \{a,\alpha\}$ by its inverse $z=x^{-1}$ which allows us to 
make the string of signs in (\ref{Borelmach}) a trifle more appealing, ie
\begin{equation}
 \sum_{k \geq 0} a_{k}z^{-k-1} = \int_{0}^{\infty} e^{-z \zeta} (\sum_{k\geq 0} \frac{a_{k}}{k!} \zeta^{k}) d\zeta . 
\end{equation}
But appearance aside, the reason for our choice is that we have no intention to the change the entire terminology of an established theory (namely 
that of transseries). The field for the transseries' coefficients we first use is $\R$. We will then explain when complex numbers are fine. 
We would like to warn the reader that the following material is extremely concise. It may thus be useful to take notes.  

\subsection{Grid-based Hahn series}\label{subsec:Hahn}
Let $(\Mo, \preccurlyeq )$ be a totally ordered abelian group with the neutral element denoted by $1$. 
For $\g,\g' \in \Mo$ with $\g \neq \g'$ and $\g \preccurlyeq \g'$, we write $\g \prec \g'$. 
If $\g \prec 1$, then we refer to $\g$ as \emph{small}, if $\g \succ 1$, we call it \emph{large}. 
The elements of the \emph{monomial group} $\Mo$ are called \emph{Hahn monomials}. A \emph{Hahn series} is a map $T \colon \Mo \rightarrow \R$ written in the 
form of a formal series,
\begin{equation}
 T = \sum_{\g \in \Mo} T_{\g} \g,  \hs{2} (\mbox{'Hahn series'})
\end{equation}
where $T_{\g} := T(\g) \in \R$ is the image of $\g$ under $T$. Whenever we speak of the \emph{support} of a Hahn series $T$, we mean the 
subset $\supp(T):= \{ \g \in \Mo \, |\, T_{\g} \neq 0\}$.

An example of a monomial group is $\Mo_{0}:=\{ z^{a} \, |\, a \in \R \}$ with group law $z^{a} \cdot z^{b} := z^{a+b}$ and order relation $z^{a} \preccurlyeq z^{b}$ if $a \leq b$.
The elements of $\Mo_{0}$ are called \emph{log-free transmonomials of height 0} (zero). One may take them either as 
abstract symbols or as functions. In this latter interpretation, small monomials vanish in the limit $z \rightarrow \infty$, while large ones diverge and
the order relation tells us what their quotients are doing. Hahn series built from these transmonomials are called \emph{log-free transseries of height 0}. 

For a Hahn series $T$, the element $\magn(T) \in \supp(T)$ defined through the property 
\begin{equation}\label{state}
 \g \preccurlyeq \magn(T) \hs{2} \text{for all } \g \in \supp (T)
\end{equation}
is called \emph{magnitude} or \emph{dominating monomial} of $T$. We write $T \prec 1$ if $\magn(T) \prec 1$ and say that $T$ is \emph{small}.
If $\magn(T) \succ 1$, we say that $T$ is \emph{large} and write $T \succ 1$. If, moreover, $\g \succ 1$ for all $\g \in \supp(T)$ and not just the large 
ones, then $T$ is called \emph{purely large}. By convention, the zero transseries is also purely large because the property is trivially satisfied by 
the statement in (\ref{state}) not being false since $\supp(0)=\emptyset$. As an example, consider 
\begin{equation}\label{trex}
 T = z^{2} + 1 + \sum_{n \geq 1} z^{-\sqrt{n}}
\end{equation}
which has dominating monomial $\magn(T)=z^{2} \succ 1$ and is therefore large but not purely large on account of its small and constant part. 
It can nonetheless be decomposed into a purely large, constant and small piece, respectively given by $T_{\succ}=z^{2}$, $T_{\asymp}=1$ and $T_{\prec}=  \sum_{n \geq 1} z^{-\sqrt{n}}$. 

If the leading coefficient is positive (nonnegative), in signs $T(\magn(T)) > 0$ ($\geq 0$), we say that $T$ is positive (nonnegative) and write 
$T > 0$ ($\geq 0$), mutatis mutandis for negative Hahn series. Together with addition, this allows us to impose an ordering on the set of transseries 
by defining the ordering relation through $T \geq S : \Leftrightarrow T-S \geq 0$.

We pick a finite subset $\mo  = \{ \mo_{1} , \dotsc , \mo_{n} \} \subset \{ \g \in \Mo \, |\, \g \prec 1\}$
of small transmonomials and define for every multi-index $l=(l_{1}, \dotsc ,l_{n}) \in \Z^{n}$ the group element 
$\mo^{l} := \mo_{1}^{l_{1}} \dotsc \mo_{n}^{l_{n}} \in \Mo$. 
These Hahn monomials generate the subgroup $\J^{\mo}:=\{ \mo^{l} \, |\, l \in \Z^{n} \} \subset \Mo$. 
Recall that the integer lattice $\Z^{n}$ is partially ordered: given $k,l \in \Z^{n}$, we write $k \geq l$ if $k_{j}\geq l_{j}$ for all 
$j \in \{1,...,n\}$. With this partial ordering, we define \emph{grids} to be subsets of the form 
\begin{equation}
 \J^{\mo}_{k} := \{ \mo^{l} \, |\, l \geq k \} \subsetneq \J^{\mo}.  
\end{equation}
It will be useful to visualise the subgroup $\J^{\mo}$ as $\Z^{|\mo|}$ and the grids $\J^{\mo}_{k}$ as the corresponding 
subsets of $\Z^{|\mo|}$ as illustrated in \textsc{Figure} \ref{grid} for the grid $\J^{\mo}_{(3,2)}$ with $\, |\,\mo\, |\,=2$.

\begin{figure}[ht]
\begin{center} \includegraphics[height=4cm]{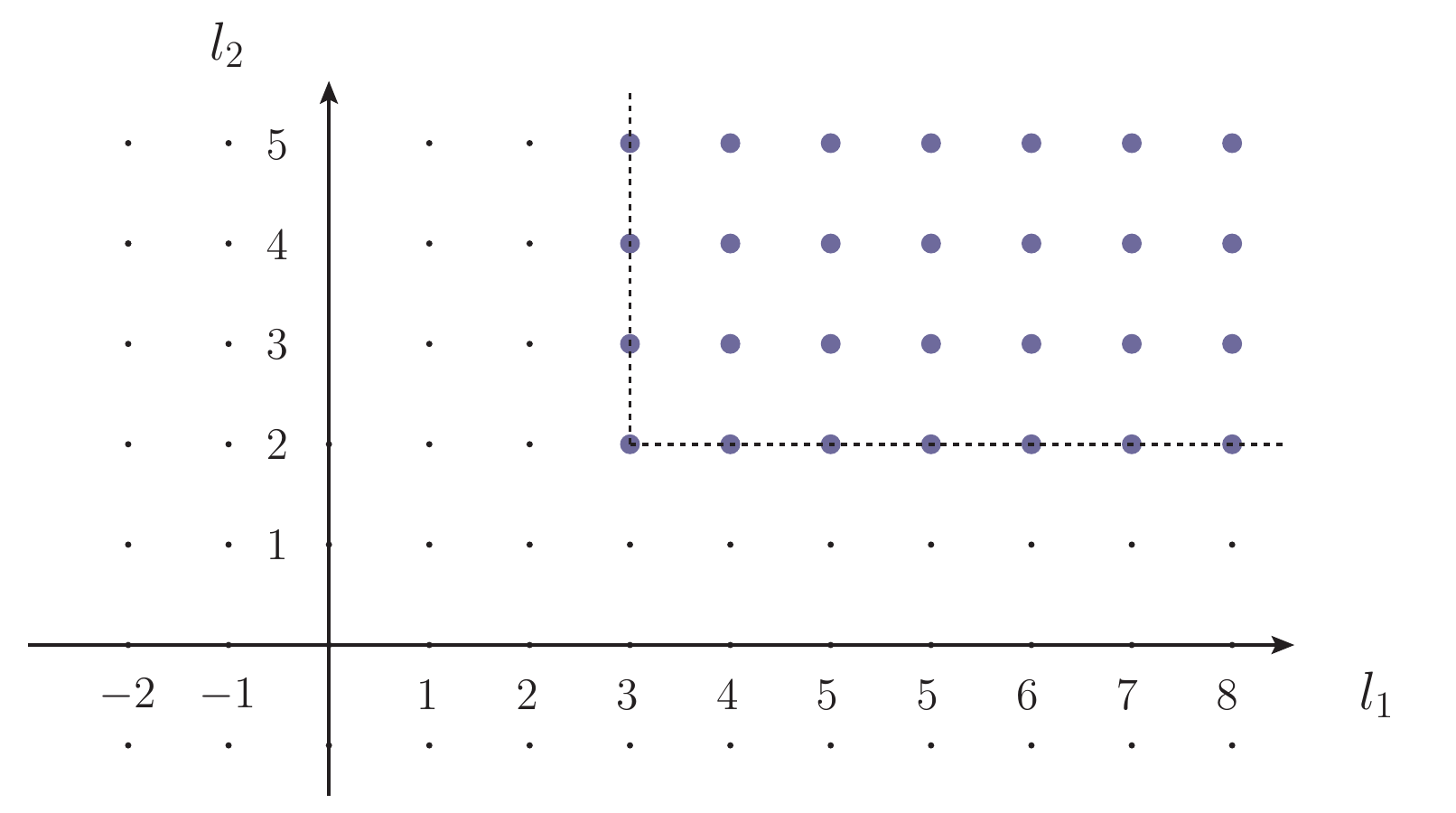} \end{center} 
\caption{\small An illustration of the subgroup $\J^{\mo} \subset \Mo$ for $\, |\,\mo\, |\,=2$, identified with $\Z^{2} \cong \J^{\mo}$ and the 
grid $\J^{\mo}_{(3,2)}$ (big dots). The thin dotted line indicates the boundary of the grid. }
\label{grid}
\end{figure}

A Hahn series $T$ is called \emph{grid-based} if there exists a generator set $\mo= \{ \mo_{1} , ..., \mo_{n} \}$ and a multi-index $k \in \Z^{n}$ such 
that $\supp(T) \subset \J^{\mo}_{k}$. One says that $T$ is supported by the grid $\J^{\mo}_{k}$ or that this grid is \emph{supportive} for $T$. 
Subsets of a grid are referred to as \emph{subgrids}. 

The set of grid-based Hahn series with monomial group $\Mo$ is denoted by $\Tm{\R}$. The nice thing about such series is that their product is well-defined:  
\begin{equation}
 T \cdot S =  (\sum_{\mathfrak{a} \in \Mo}T_{\mathfrak{a}}\mathfrak{a} ) \cdot (\sum_{\mathfrak{b} \in \Mo}S_{\mathfrak{b}}\mathfrak{b} )
 := \sum_{\mathfrak{g} \in \Mo} (\sum_{\mathfrak{a} \mathfrak{b} = \g} T_{\mathfrak{a}} S_{\mathfrak{b}} ) \g = 
 \sum_{\g \in \Mo} (T \cdot S)_{\g} \g, 
\end{equation}
which is due to the fact that the set $\{ (\mathfrak{a}, \mathfrak{b}) \in \supp(T) \times \supp(S) : \mathfrak{a} \mathfrak{b} = \mathfrak{g} \}$ is 
finite. It promotes the vector space of grid-based Hahn series to an algebra. Moreover, grid-based series can be written in the form
\begin{equation}\label{HS}
 T = \sum_{l \in \Z^{n} } t_{l} \mo^{l}    
\end{equation}
and the product of two Hahn series $T, S$ in this notation reads
\begin{equation}
T \cdot S = (\sum_{l' \in \Z^{n} } t_{l'} \mo^{l'} ) \cdot ( \sum_{l'' \in \Z^{n} } s_{l''} \mo^{l''} ) 
=  \sum_{l \in \Z^{n}} (\sum_{l'+l''=l} t_{l'}s_{l''})\mo^{l}.
\end{equation}

\subsection{Log-free transseries}
We now dispose of the notions necessary to define grid-based log-free transseries. First, we denote by $\Ts_{0}:=\Tlm{\R}{\Mo_{0}}$ the algebra of 
grid-based log-free transseries of height 0, ie grid-based Hahn series supported by grids in the group $\Mo_{0}=\{ z^{a} \, | \, a \in \R \}$. 
Notice that these series are already generalisations of formal power series in $\R[[z^{-1}]]$ but that the example (\ref{trex}) is not grid-based. 

Purely large transseries will play a special role in what follows. It is therefore vital to understand them. Consider the zero-height transseries
\begin{equation}\label{pula}
 \sum_{n \geq -\ell} a_{n} z^{-n} = a_{-\ell}z^{\ell} + \dotsc + a_{-2}z^{2} + a_{-1}z + a_{0} + \sum_{n \geq 1}a_{n}z^{-n}.
\end{equation}
It is large but purely large only if $a_{n}=0$ for all $n \geq 0$. This is of course a very special and certainly not the most general transseries of zero height,
we could have chosen any finite number of transmonomials in $\Mo_{0}$. But it teaches us something: since the supportive grids of transseries 
in $\Ts_{0}$ all have a maximal element by the very definition of grids, a purely large transseries of height zero can only be made up of a finite 
number of large transmonomials and no constant!  

Let us denote the subalgebra of purely large transseries by $\Ts_{0}^{\succ} \subsetneq \Ts_{0}$ and define 
\begin{equation}
 \Mo_{1}:=\left\{ z^{a} e^{A} \, |\, a \in \R, A \in \Ts_{0}^{\succ} \right\}
\end{equation}
to be an ordered abelian group with the obvious group law given by $(z^{a} e^{A})\cdot (z^{b} e^{B}):=z^{a+b} e^{A+B}$ and 
ordering $z^{a} e^{A} \preccurlyeq z^{b} e^{B}$ if  $A \leq B$ or $A=B$ but then $a \leq b$. The elements of $\Mo_{1}$ are called 
\emph{log-free transmonomials of height 1}. Note that because $0 \in \Ts_{0}$ is by definition also purely large, $\Mo_{0} \subsetneq \Mo_{1}$ is a
proper subgroup and $\Ts_{0} \subsetneq \Ts_{1}$ a proper subalgebra. A transmonomial $\mathfrak{t} \in \Mo_{1}$ is said to be of \emph{exact} height 1, if 
$\mathfrak{t} \notin \Mo_{0}$.

The reader may guess it, $\Ts_{1}:=\Tlm{\R}{\Mo_{1}}$ denotes the algebra of \emph{grid-based log-free transseries of height 1}. An example is the 
small transseries given by 
\begin{equation}
 T =  4 z^{-2} e^{-3z} + z^{-1} e^{-z + 2z^{2} - z^{3}}.
\end{equation}
In this example, we have $-z + 2z^{2} - z^{3} < -3z$ and therefore $z^{-1} e^{-z + 2z^{2} - z^{3}} \prec 4 z^{-2} e^{-3z}$. Note that now, with nontrivial
exponential height, purely large transseries, ie the elements of $\Ts_{1}^{\succ} \subsetneq \Ts_{1}$, need no longer have a finite number of terms. If we modify the above example (\ref{pula}), then 
\begin{equation}\label{pula1}
\sum_{n \geq -\ell} a_{n} z^{-n} e^{z} 
\end{equation}
is a purely large transseries of exact height 1 even in case it is an infinite series.

To proceed towards general grid-based log-free transseries, the game is now played inductively on height: the algebra of \emph{grid-based log-free transseries of height $N$} is given by
$\Ts_{N}=\Tlm{\R}{\Mo_{N}}$ where 
\begin{equation}
 \Mo_{N} := \left\{z^{a} e^{A} \,|\, a \in \R, A \in \Ts^{\succ}_{N-1} \right\}
\end{equation}
is the monomial group of log-free transmonomials of height $N$ and, finally, the monomial group of \emph{grid-based log-free transmonomials} 
$\Mo_{\bullet}:=\bigcup_{N \in \N}\Mo_{N}$ then gives rise to the set of \emph{grid-based log-free transseries} $\Ts_{\bullet}=\Tlm{\R}{\Mo_{\bullet}}$. 

\subsection{Transseries with logarithms} As alluded to in §\ref{sec:Intro}, we shall not use transseries with logarithms. Besides
simplicity, adding logs in the form of log polynomials does not affect our results, as we will explain when the need arises in later sections. 
Nonetheless, for the sake of completeness, we shall give a lightening introduction.
 
One defines additional Hahn monomials involving logarithms as follows. Let 
\begin{equation}
 \log_{M}z:= (\log \circ \dotsc \circ \log)(z)
\end{equation}
be the $M$-fold composition of the natural logarithm, either seen as a symbol or as a function. 
For a zero height transmonomial $\g = z^{a} \in \Mo_{0}$, we define the composition with a logarithm by 
$\g \circ \log_{M}= z^{a} \circ \log_{M}:=(\log_{M} z)^{a}$ and for a transseries of zero 
height, we set
\[
 T \circ \log_{M} = (\sum_{\g \in \Mo_{0}} T_{\g} \g )\circ \log_{M}:= \sum_{\g \in \Mo_{0}} T_{\g} (\g \circ \log_{M}),
\]
ie a termwise composition. Then inductively on height, one sets
\[
 (z^{a}e^{A}) \circ \log_{M}:= (z^{a} \circ \log_{M}) e^{A \circ \log_{M}}
\]
for a monomial $z^{a}e^{A} \in \Mo_{N}$. Hahn monomials obtained this way are called \emph{transmonomials} of height $N$ and depth $M$, where depth 
refers to 'logarithmic depth', while height - the reader has probably by now realised it - refers to 'exponential height'. The set of such monomials is 
denoted by $\Mo_{NM} := \{ \g \circ \log_{M} \, |\, \g \in \Mo_{N} \}$ and the set of all transmonomials is given 
by $\Mo_{\bullet \bullet} = \bigcup_{N,M \in \N }\Mo_{N M}$. Finally, $\Ts_{\bullet \bullet}=\Tlm{\R}{\Mo_{\bullet \bullet}}$ is the set of all 
grid-based transseries. As regards the order relation for these new transmonomials, let $Q,P \in \Mo_{N}$ be both log-free, then one sets
\begin{equation}
 Q \circ \log_{M} \prec P \circ \log_{M} :\Leftrightarrow Q \prec P
\end{equation}
and $Q \circ \log_{M} \prec Q \circ \log_{N}$ if $N < M$. For example, by $z \prec e^{z}$ we find $\log z \prec z$ due to 
\begin{equation}
 \log z = z \circ \log \prec e^{z} \circ \log = e^{\log z} = \exp (\log z) = z \circ (\log_{-1} \circ \log) = z ,
\end{equation}
where $\log_{-N}z$ for negative $-N$ is defined to be an iterated exponential, ie in this case the $N$-fold composition 
$\exp_{N}z = (\exp \circ \dotsc \circ \exp)(z)$. 
The reader can see now what class the transseries in (\ref{pertexT}) belongs to: 
it is a \emph{logarithmic transseries} of height 1 and logarithmic depth 1.

\section{Renormalisation and resurgent transseries in quantum field theory}\label{sec:RenorTra}       
As alluded to in the introduction, the correct type of transseries needed to characterise the observables of a given fully fledged 
renormalisable (or better: renormalised) QFT are as yet totally unknown. We shall present in this section our tentative view on this 
issue and argue why the negative results of this work are not entirely unexpected, the root cause being renormalisation. 

\subsection{Resurgent transseries for quantum field theory}\label{subsec:trans} 
The ideas currently being floated suggest that the transseries of QFTs should be of exponential height 1 and logarithmic depth 1, that is, it is 
expected to be of the form
\begin{equation}\label{transa}
 f(z) = \sum_{\omega \in \Omega \cup \{ 0\} } e^{-\omega z} f_{\omega}(z),
\end{equation}
with $f_{\omega}(z) \in \C[[z^{-1}]][\log z]$, ie polynomials in the variable $u=\log z$ with formal power series in $z^{-1}$ as coefficients. 
The set $\Omega$ contains all singularities of the Borel transform of the perturbative series $f_{0}(z) \in \C[[z^{-1}]]$, where the assumption is
that $\omega \in \Omega$ implies $\text{Re}(\omega)\neq 0$, ie none of the singularities are purely imaginary. 

This is an important point since otherwise, there would be an ordering problem for the support of the transseries in (\ref{transa}) when comparing two 
exponentials. Generally, in the construction of transseries, the order relation $>$ which compares the leading
coefficients of two transseries is altered to compare their real parts when the field is changed from $\R$ to $\C$. Complications only arise when 
transseries with purely imaginary leading coefficients are exponentiated \cite{Co09}. 

Notice that (\ref{transa}) is at face value not in general grid-based but only \emph{well-based}, which for practical purposes means that the support 
has a maximum. The problem with well-based but not grid-based transseries is that it is not known whether and how they can be treated with \'Ecalle's 
formalism of accelero-summation \cite{Ed09}. Hence only transseries of the grid-based type can so far safely be called \emph{resurgent}. 
In a worst-case scenario, the observables of a renormalsied QFT cannot be represented by resurgent transseries and are therefore not analysable functions.

However, such bleak view on things is unfounded. The evidence unearthed so far suggests that grid-based transseries may actually be sufficient: 
several quantum-mechanical examples \cite{DunU14}, the genus expansions in topological string theory \cite{CESVo15, CESVo16}, 
minimal (super)string theory \cite{ASVo12,SchiVa14}, the coupling expansions of the SUSY Yang-Mills theories studied in \cite{ARuS15}
and the model dealt with in \cite{DunU13} \emph{all} exhibit singularities in the Borel plane at places of the form $\omega_{n}=n \sigma $, where 
$\sigma \in \C$ and $n \in \N^{*}$, ensuring that (\ref{transa}) is grid-based. 
Moreover, the renormalon analysis of \cite{Ben99} shows that the renormalons of QCD are situated on the real line $\R$ at 
$\omega^{\pm}_{n} = \pm n/\beta_{0}$, where $\beta_{0}$ is the one-loop contribution to the beta function. 

\subsection{Renormalisation as a game changer}\label{subsec:Recha}
But as asserted in §\ref{subsec:transQFT}, the idea that the semi-classical expansion of the partition function can be written in terms
of an action of the form
\begin{equation}\label{action}
 S(\phi) = \frac{1}{\alpha} \int d^{d}x \  \La(\phi,\partial \phi)
\end{equation}
is not true for the \emph{renormalised version of this theory}. Let us quickly have a look at it in the case of Euclidean $(\phi^{4})_{4}$ theory and see 
why we may have to expect the transseries of a renormalised QFT to be of a fundamentally different nature. 

Despite the fact that the following argument is completely heuristic, we still think it of value, especially in the light of the negative result 
obtained in this work. We begin with the bare action 
\begin{equation}\label{bare}
 S(\varphi)= \int d^{4}x \left( \frac{1}{2}\varphi [-\nabla^{2}+m_{0}^{2}]\varphi + \frac{\lambda_{0}}{4!}\varphi^{4} \right),
\end{equation}
where $\lambda_{0}$ and $m_{0}$ are the bare coupling and mass, respectively. The field $\varphi$ can in this case indeed be
scaled to take the inverse of $\alpha=\lambda_{0}^{2}$ to the front of the integral as in (\ref{action}), in the simplest form, one has
\begin{equation}\label{action2}
 S(\varphi) = \frac{1}{\lambda_{0}} \int d^{4}x \left( \frac{1}{2} \lambda_{0}^{1/2}\varphi [-\nabla^{2}+m_{0}^{2}] \lambda_{0}^{1/2} \varphi 
 + \frac{1}{4!}(\lambda_{0}^{1/2}\varphi)^{4} \right) =
 \frac{1}{\lambda_{0}} \bar{S}\left(\lambda_{0}^{1/2} \varphi \right) = \frac{1}{\lambda_{0}} \bar{S}(\phi), 
\end{equation}
where the new field is $\phi := \lambda_{0}^{1/2} \varphi$ and $\bar{S}$ has no explicit coupling dependence. 
However, when we renormalise the theory, this procedure is no longer possible. The renormalised version of (\ref{bare}) is given by 
\begin{equation}
 S_{R}(\varphi,\lambda)=  \underbrace{\int d^{4}x \left(\frac{1}{2}\varphi [-\nabla^{2}+m^{2}]\varphi \right)}_{=:S_{0}(\varphi)} 
 + S_{\text{int}}(\varphi,\lambda),
\end{equation}
with $S_{0}(\varphi)$ being the free part, $\lambda$ the physical coupling and 
\begin{equation}
 S_{\text{int}}(\varphi,\lambda) = \int d^{4}x \left( \frac{1}{2}[Z(\lambda)-1](\nabla \varphi)^{2} 
 + \frac{1}{2}m^{2}[Z_{m}(\lambda)-1]\varphi^{2} + \frac{\lambda}{4!}Z_{c}(\lambda)\varphi^{4} \right)
\end{equation}
the interaction part for the renormalised field $\varphi = Z(\lambda)^{-1/2}\phi$, sporting all necessary renormalisation Z factors: 
$Z_{m}(\lambda)$ for the mass and $Z_{c}(\lambda)$ for the coupling. The task is now to compute the partition function
\begin{equation}
 \mathcal{Z}(J,\lambda) = \int \De \varphi \ e^{-S_{R}(\varphi,\lambda)+\int J \cdot \varphi} 
\end{equation}
with external source field $J$. All quantities including the anomalous dimension are extracted from this function.   
The anomalous dimension is usually defined via the wavefunction renormalisation constant $Z(\lambda)$ which is assumed 
to exists as a nonperturbative quantity (for finite cut-off). Alas, the 3 renormalisation Z factors are in principle only known in terms of their  
divergent perturbation series, ie even for a finite cut-off are the 3 series
\begin{equation}
 Z(\lambda)=\sum_{l\geq 0}a_{l}\lambda^{l} \ , \hs{1} Z_{m}(\lambda)=\sum_{l\geq 0}b_{l}\lambda^{l}, \hs{1} Z_{c}(\lambda) =\sum_{l\geq 0}c_{l}\lambda^{l} 
\end{equation}
badly divergent, where $a_{0}=b_{0}=c_{0}=1$ (we have suppressed both cut-off and scale dependence). Never mind, let us study the critical points, that is, the solutions of
\begin{equation}
 \frac{\delta S_{R}(\varphi,\lambda)}{\delta \varphi(x)} = \left[- Z(\lambda) \nabla^{2} + m^{2} Z_{m}(\lambda) \right] \varphi(x) 
 + \frac{\lambda}{3!}Z_{c}(\lambda)\varphi(x)^{3} = 0
\end{equation}
for vanishing external field $J=0$. We see that even the constant solution $\overline{\varphi} \neq 0$ given by 
\begin{equation}
 \overline{\varphi} = \pm i  \sqrt{\frac{3!m^{2}Z_{m}(\lambda)}{\lambda Z_{c}(\lambda)}}
\end{equation}
has a critical-point action with highly nontrivial coupling dependence, 
\begin{equation}\label{sadact}
S_{R}(\overline{\varphi},\lambda) = -  \frac{3m^{4}Z_{m}(\lambda)^{2}}{2\lambda Z_{c}(\lambda)},
\end{equation}
in particular, nothing proportional to $\lambda^{-1}$ as encountered in (\ref{action2}). At the very least, this tells us that we should not expect 
(\ref{transa}) to be a transseries appropriate for renormalised QFTs, where $z=\lambda^{-1}$ as explained above. Our stance is that we take the canonical
formalism of renormalisation very seriously because it leads to valid results and harbours a combinatorial truth.

To carry the discussion on renormalisation a bit further, let us briefly recount how perturbation theory proceeds from here.
Yet we will not use functional derivatives with respect to the external source field $J$ as this would only obscure our point\footnote{We also do not need 
Wick's theorem to make our point.}. In a first step one expands the interaction part of the renormalised action 
\begin{equation}\label{expan}
 S_{\text{int}}(\varphi,\lambda) = \sum_{l\geq 1} S_{l}(\varphi) \lambda^{l} := \sum_{l\geq 1} \int d^{4}x \left( \frac{a_{l}}{2}(\nabla \varphi)^{2} 
 + \frac{b_{l}}{2} m^{2}\varphi^{2} + \frac{c_{l-1}}{4!}\varphi^{4} \right) \lambda^{l} 
\end{equation}
and in a second step plugs it back into the exponential of the Euclidean damping factor of the path integral so that one gets
\begin{equation}\label{expans}
 \mathcal{Z}(0,\lambda) = \int \De \varphi \ e^{-S_{0}(\varphi)} e^{-\sum_{l\geq 1} S_{l}(\varphi) \lambda^{l}}.
\end{equation}
Notice what form this object has: it is the infinite sum of exponentials with a formal power series upstairs! This smacks of an integral 
(ie sum) over transseries of height 1. It is not a legal transseries, though: the power series in the second exponential is meant to be asymptotic with 
respect to the limit $\lambda \rightarrow 0$ (ie $\lambda^{-1} = z \rightarrow \infty$) and is therefore not purely large. 
However, the story goes on with expanding the second exponential, 
\begin{equation}\label{Z}
\mathcal{Z}(0,\lambda) = \sum_{N=0}^{\infty} \frac{1}{N!} \int \De \varphi \ e^{-S_{0}(\varphi)}
\left(-\sum_{l\geq 1} S_{l}(\varphi) \lambda^{l} \right)^{N}
\end{equation}
which is now a sum over powers of legal (yet divergent) transseries, subjected to an integration procedure that ignores the coupling $\lambda$. 
Of course, this is a naive picture, considering that we are dealing with a functional integral (according to the canonical narrative). 

There is not much to be found in the literature on the complications brought about by the need to renormalise a theory. Stingl acknowledges 
the drastic change when the coupling is replaced by a 'divergent coupling renormalisation' in \cite{Sti02} but says that ''... we encounter no changes 
in the form of the corresponding Lagrangians, ...'' (ibidem p.33). This, in our mind, leads to a wrong kind of thinking about the 
issue\footnote{In \cite{Sti02}, Stingl has a lot of interesting things to say on resurgence in QFT which we do not discuss here. In his 
conclusion section on p.126, he essentially says that using the 'trademark' $\beta(\alpha) \partial_{\alpha}$ ''... would be much more after the heart of 
the resurgence theorist, ...''. That is what we are doing in this work.}: as a matter of fact, perturbation theory never treats a renormalised Lagrangian as if it were of the same 'form' but instead singles out the dramatically 
changed interaction part (\ref{expan}) and lets it take part in a wild game of divergence cancellations, albeit Hopf algebra governed \cite{Krei02}! 

In Chapter 40 of his monumental monograph \cite{Zi02}, Zinn-Justin briefly discusses instanton contributions for a renormalised scalar theory and how to deal with 
the counterterm power series of (\ref{expan}) in the path integral (\ref{expans}). In terms of (\ref{expan}), his argument, as we understand it, goes as 
follows. First write the renormalised action as 
\begin{equation}
S(\varphi) = \frac{1}{\lambda} S_{0}\left(\lambda^{1/2} \varphi \right) + S_{1}\left(\lambda^{1/2} \varphi \right) 
+ S_{2}\left(\lambda^{1/2} \varphi \right) \lambda + S_{3}\left(\lambda^{1/2} \varphi \right) \lambda^{2} + \dotsc , 
\end{equation}
then compute the effective action including the counterterms order by order in $\lambda$. To our mind, this perturbative procedure generates exponentials 
with nontrivial coupling dependence, an aspect certainly worth being investigated further. 

Another pertinent find that we prefer to only mention for completeness is \cite{We79} which has a result on instantons in effective Yang-Mills theories with a trace anomaly. The author shows that if 
the energy-momentum tensor has a trace anomaly, for example incurred by renormalisation as described in \cite{CoDuJo77}, then there are no instanton 
solutions for which any effective action would be finite. 

\subsection{Exponential size of Borel transforms}\label{expsi}
Finally, we would like to express our suspicion that renormalisation is very likely to increase 
the \emph{exponential size} of the Borel transforms. 
This means the following. Let $p$ be a path in $\C$ from the origin to infinity and $\Nei_{p}$ a neighbourhood of $p$. 
A function $f \colon \Nei_{p} \rightarrow \C$ is then said to be \emph{of exponential size} (or order) $m \geq 0$ if $m$ is the smallest real number
for which there are constants $c,A >0$ such that 
\begin{equation}
|f(\zeta)| \leq A e^{c|\zeta|^{m}} \hs{2} \forall \zeta \in \Nei_{p}.
\end{equation}
This is a slightly generalised version of the usual definition, as given in \cite{Ba00}, much better suited for resurgence\footnote{Balser demands this growth bound 
be fulfilled in a sector (see \cite{Ba00}, pp.232), whereas resurgence theory seeks endlessly continuable germs along paths that steer around the 
singularities of an arbitrary set $\Omega$ of isolated singularities \cite{Sa14}. }. If $m>1$, then condition (iii) 
in §\ref{subsec:Bor}, which guarantees the existence of the Borel-Laplace transform in (\ref{LaBo}), is no longer satisfied. As mentioned in 
§\ref{subsec:Bor}, this requires one to consider the theory of multisummable series \cite{Ba00, Ba09}. 

The reason why we suspect that (\ref{Z}) may produce something whose Borel transform might be of higher exponential size is rather simple: we know that
in dimension $d=2$, $\varphi^{4}$ theory is superrenormalisable and (\ref{Z}) 'simplifies' to \cite{Sal99}
\begin{equation}\label{Z'}
\mathcal{Z}_{0}(0,\lambda) = \sum_{N=0}^{\infty} \frac{1}{N!} \int \De \varphi \ e^{-S_{0}(\varphi)}
\left(-S_{1}(\varphi) \lambda \right)^{N}
\end{equation}
but already produces a series of type Gevrey 1 \cite{Ja65,GuZi90} which, on account of its Borel summability \cite{EMaS75}, 
has a Borel transform of exponential size at most $m=1$. 

In contrast to (\ref{Z'}), the series of Gaussian averages in (\ref{Z}) contains asymptotic power series for each $N \geq 1$. In the light of this, 
although we are aware that subtle cancellations intended by renormalisation take place, it strikes us like a rather strong statement to say that (\ref{Z}) produces a formal power series
no more severe than (\ref{Z'}) does.

\section{Transseries ansatz}\label{sec:Transatz}					  
In the light of these deliberations, we have to reckon with severe problems in finding the right transseries ansatz in a renormalised theory 
because a transseries of the form (\ref{pertexT}) cannot be expected to be the answer. However, because our thoughts in §\ref{sec:RenorTra} do not prove 
anything, we have deemed it worthwhile to investigate the nonperturbative equations deduced in §\ref{sec:DySch} for the anomalous dimension $\wg$ with a 
fairly generic state-of-the-art transseries ansatz, to be described in what follows next. 

\subsection{Transseries ansatz}
We shall use a transseries ansatz in which the two transmonomials 
\begin{equation}
 \mo_{1} = z^{-c} e^{-P(z)}  ,  \hs{1} \mo_{2}=z^{-1}
\end{equation}
are the basic building blocks; $c \in \R$ is a fixed constant and $P(z)=\sum_{j=1}^{m}b_{j}z^{j}$ a polynomial in which $b_{m} >0$, making for purely
large transseries of height 0. We shall sometimes refer to $P(z)$ as \emph{instanton polynomial}. In fact, because the action may at saddle points have 
a nontrivial coupling dependence as suggested by (\ref{sadact}), it may very well itself be a transseries of height 1. But because this is 
at this stage far beyond the scope of our investigations, we will stick to the generator set $\mo = \{ \mo_{1}, \mo_{2} \} \subset \Mo_{1}$ and 
assume $\J^{\mo}_{(0,0)}$ to be the supportive grid for the anomalous dimension $\wg$. 

We shall now drop $\mo$ from our notation and write grids as $\J_{k}:=\J^{\mo}_{k}$ for $k \in \Z^{2}$ and the subalgebra of transseries supported by 
such grids as
\begin{equation}
 \Tu := \{ \ T \in \Tlm{\C}{\Mo_{1}} \ \, |\, \ \exists k \in \Z^{2} \ \mbox{\small such that} \ \supp(T) \subseteq \J_{k} \ \} 
 = \Tlm{\C}{\mo^{l} : l \in \Z^{2}} .
\end{equation}
Our transseries ansatz for the anomalous dimension is given by
\begin{equation}\label{transan}
 \wt{\gamma}(z) = \sum_{l \geq (0,0)} \wt{\gamma}_{(l_{1},l_{2})} z^{- l_{1}c} e^{-l_{1} P(z)} z^{-l_{2}} 
 \hs{0.5} (\mbox{transseries ansatz for anomalous dimension}),
\end{equation}
where $\wt{\gamma}_{(l_{1},l_{2})} \in \C$ are the coefficients and for the summation, we make use of the partial ordering on the double-index 
set $\Z^{2}$. We set $\wg_{(0,0)}=0$ because perturbative QFT informs us that the perturbative series of the anomalous dimension $\wg$ has no constant term. 
We deem $\J_{(0,0)}$ to be a reasonable (albeit tentative) choice of \emph{supportive grid} for $\wg$. The supportive subgrid containing the support 
of $\wg$ is depicted in \textsc{Figure} \ref{support}.

\begin{figure}[ht]
\begin{center} \includegraphics[height=4.5cm]{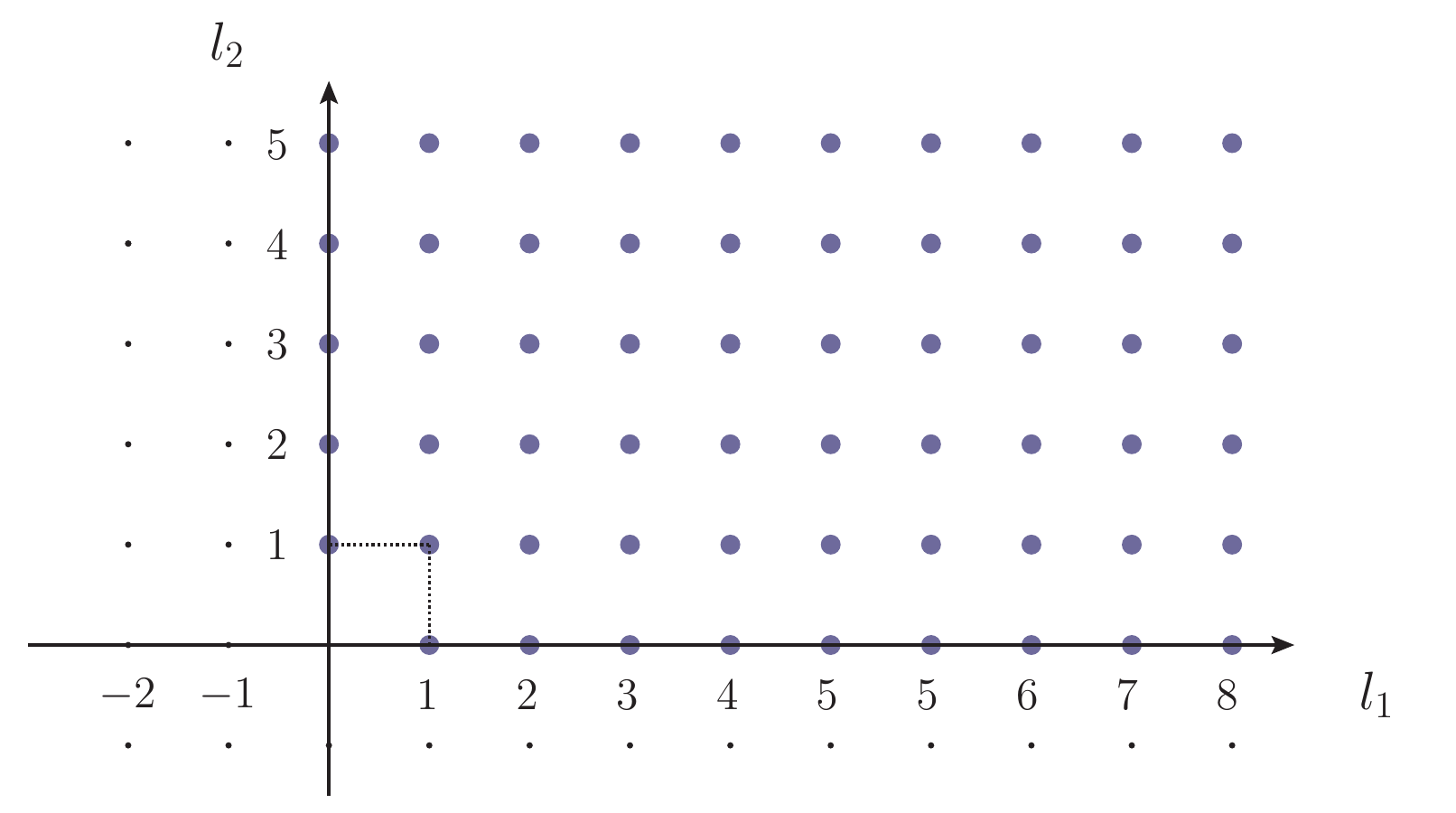} \end{center} 
\caption{\small Assumed supportive subgrid of the anomalous dimension's transseries $\wg$ as explored in this work (big dots), ie $\supp(\wg)$ is assumed to be 
a subgrid of this set.}
\label{support}
\end{figure}

To have a more precise characterisation of the supportive grids used in this work, we introduce for $n \geq 1$ the subgrids of 'perturbative'
transmonomials $\Pe_{n} := \{ \mo^{(0,s)} = \mo_{2}^{s} : s \geq n \}$ and write 
\begin{equation}
\supp(\wg) \subseteq \Pe_{1} \cup \J_{(1,0)} \subset \J_{(0,0)}.
\end{equation}
If we imagine adding polynomials of logarithms, then the support would simply be enlarged by a finite number of transmonomials for each sector in the direction  of 
a third dimension, ie
\begin{equation}
 \{ \mo_{1}^{\sigma} \} \times \{ \mo_{2}^{s} : s \in \Z \} \times \{ (\log z)^{r} : 0 \leq r \leq n_{\sigma} \}
\end{equation}
where $n_{\sigma} \in \N$ may be different for every sector $\sigma$, in particular $n_{0}=0$. This can be visualised in \textsc{Figure} \ref{support} 
by thinking of the additional logarithmic transmonomials as a finite number of dots behind (or above) every dot extending into (or out of) the page.   

\subsection{Sectors}
As is customary in physics, we distinguish between the different \emph{sectors} of a transseries. Let us define them properly now. They are given as 
subseries for a fixed first index. Let $f \in \Tu$ be a transseries with $\supp(f) \subset \J_{k}$, then we write
\begin{equation}
 f = \sum_{\sigma \geq k_{1}} \mo_{1}^{\sigma} \sum_{t \geq k_{2}} f_{(\sigma,t)} \mo_{2}^{t} 
 =: \sum_{\sigma \geq k_{1}} \mo_{1}^{\sigma} \W_{\sigma}[f] = \sum_{\sigma \geq k_{1}} z^{- \sigma c} e^{- \sigma P(z)} \W_{\sigma}[f],  
\end{equation}
where the sectors are given by the subseries $\mo_{1}^{\sigma}\W_{\sigma}[f] \in \mo_{1}^{\sigma}\Tlm{\C}{\mo_{2}^{s} : s \in \Z}$. We will at times 
also write $\W_{\sigma}f$ as a measure against notational overload. The terminology is as follows:
\begin{itemize}
 \item for $\sigma=0$, $\W_{0}[f]$ is referred to as \emph{perturbative sector} of $f$, while
 \item in the case $\sigma \neq 0$ the subseries $\mo_{1}^{\sigma} \W_{\sigma}[f] =  z^{-\sigma c} e^{- \sigma P(z)} \W_{\sigma}[f]$ is called 
      \emph{$\sigma$-th instanton sector}, or \emph{$\sigma$-th nonperturbative sector} of $f$.
\end{itemize}
The various sectors are represented by the discrete vertical lines $\{ \sigma \} \times \Z \subset \Z^{2}$ for $\sigma \in \N$ in \textsc{Figure} \ref{support}: 
they carry the perturbative $\sigma=0$ and the nonperturbative sectors $\sigma > 0$. We write
\begin{equation}
 f = \sum_{l \geq k} f_{l}\mo^{l}
\end{equation}
if the transseries $f \in \Tu$ is supported by the grid $\J_{k}$, ie $\supp(f) \subseteq \J_{k}$.

Note that multiplying two transseries may in general alter the support, a trivial fact that will be of importance in the following sections 
§\ref{sec:DynSys},\ref{sec:DSE}: if we consider the product of two transseries $f,g \in \Tu$ with $\supp(f) \subseteq \J_{k'}$ and 
$\supp(g) \subseteq \J_{k''}$, 
\begin{equation}
 f \cdot g = (\sum_{l' \geq k'}f_{l'} \mo^{l'}) \cdot (\sum_{l'' \geq k''}g_{l''} \mo^{l''}) 
 = \sum_{l \geq k'+ k''} (\sum_{l' + l'' =l} f_{l'}g_{l''}) \mo^{l}, 
\end{equation}
we see that $\supp(f \cdot g)=\supp(f)\cdot \supp(g) \subseteq \J_{k'} \cdot \J_{k''} = \J_{k'+k''}$, ie every element of one grid is multiplied by every 
of the other.

\section{RG recursion as a discrete dynamical system}\label{sec:DynSys}			  
We will in this section study the RG recursion (\ref{RGrec}) in the algebra of transseries $\Tu$ and make extensive use of the terminology introduced in 
§\ref{sec:ReTra} and the previous section. It is in our mind convenient to view the RG recursion as a \emph{discrete dynamical system} which is why we will also bring 
in a little jargon from that side\footnote{Readers should be familiar with it.}.    

\subsection{Discrete dynamical system}
If we denote the coupling by $x$, then we recall from §\ref{sec:DySch} that the RG functions are produced by repeatedly applying the RG 
operator $\Rn$ to the anomalous dimension $\gamma(x)$,
\begin{equation}
 \gamma_{n+1}(x) = \Rn^{n} \gamma(x) = [\beta(x)\partial_{x} - \gamma(x)]^{n} \gamma(x).
\end{equation}
For the two cases we investigate, the beta function is given by \cite{Y11}
\begin{equation}
\mbox{Yukawa model:} \hs{0.5} \beta(a) = 2 a \gamma(a) , \hs{2} \mbox{QED:} \hs{0.5} \beta(\alpha) = \alpha \gamma(\alpha) .
\end{equation}
We will therefore write $\beta(x)= s x \gamma (x)$ and $\Rn = \gamma(x)[s x \partial_{x} - 1]$ which is the generic form capturing both cases. We do not
find it necessary to use extra notation for the two models, nowhere will there be potential for confusion: either the coupling's symbol or the context 
will make it clear. 

However, in the transseries setting with variable $z=x^{-1}$, we write the RG operator as   
\begin{equation}
\Rn = \wt{\gamma}(z)[s \der - 1]
\end{equation}
with $\der:=-z \partial_{z}$ on account of $x\partial_{x} = -z \partial_{z}$. 

We think of this differential operator as a derivation $\der \colon \Tu \rightarrow \Tu$ on the algebra $\Tu$,
which acts on the transmonomials according to
\begin{equation}\label{tra}
\der(\mo^{l}) = \der \left( z^{- l_{1}c} e^{-l_{1} P(z)} z^{-l_{2}} \right) 
= (cl_{1} + l_{2}) \mo_{1}^{l_{1}}\mo_{2}^{l_{2}} + l_{1} \sum_{i=1}^{m} i b_{i} \mo_{1}^{l_{1}}\mo_{2}^{l_{2}-i}
\end{equation}
and thus takes a transseries $f= \sum_{l \geq  (s,t)}f_{l} \mo^{l}$ supported by the grid $\J_{(s,t)}$ and maps it to 
\begin{equation}\label{deri}
 \der f = \sum_{l \geq  (s,t-m)} (\der f)|_{l} \mo^{l} 
 := \sum_{l \geq  (s,t-m)} \left[ (l_{1}c +l_{2}) f_{l} + l_{1} \Sum_{i=1}^{m} i b_{i} f_{(l_{1},l_{2}+i)} \right] \mo^{l}.
\end{equation}
This shows that $\der$ may change the support: (\ref{deri}) tells us that $\supp(\der f) \subseteq \J_{(s,t-m)}$, where, of course, 
$\J_{(s,t)} \subseteq \J_{(s,t-m)}$ which implies that the support may grow.

The RG recursion defines a discrete dynamical system with flow map given by  
\begin{equation}
\Phi \colon \N \times \Tu \rightarrow \Tu \hs{2}  \Phi(n,f) := \Rn^{n}(f)  \hs{2} (n \geq  0),
\end{equation}
where $f \in \Tu$ is the starting point. Of course, the only orbit of interest to us is the very orbit we get if we let the flow start at the 
anomalous dimension $\wg$, ie $\Phi(\N,\wg) = \{ \Phi(n,\wg) : n \in \N \}$, which is also part of the operator $\Rn=\wt{\gamma}(z)[s \der - 1]$.  
The DSEs for the anomalous dimension derived in §\ref{sec:DySch} appear in this context as a condition imposed on the orbit.

In order to compute how the RG recursion operator $\Rn$ changes the supportive grid, we use the (straightforward) rules 
\begin{equation}\label{rules}
\Pe_{n} \cdot \Pe_{m} = \Pe_{n+m}, \hs{1} \Pe_{n} \cdot \J_{(k_{1},k_{2})} = \J_{(k_{1},k_{2}+n)} ,\hs{1}  \J_{k} \cdot \J_{k'} = \J_{k+k'}.
\end{equation}
We first note that
\begin{equation}\label{suppf}
 \supp (\der f) \subseteq \Pe_{v} \cup \J_{(s,t-m)}
\end{equation}
for $f \in \Tu$ with $\supp(f) \subseteq \Pe_{v} \cup \J_{(s,t)}$. Consequently with $\supp(\wg) \subseteq \Pe_{1} \cup \J_{(1,t)}$ the first RG step 
yields 
\begin{equation}\label{suppfa}
\begin{split}
 \supp(\Rn \wg) &= \supp(\wg) \cdot \supp([s\der-1] \wg) \subseteq \left(\Pe_{1} \cup \J_{(1,t)} \right) \cdot \left(\Pe_{1} \cup \J_{(1,t-m)}\right) \\
 &= \Pe_{2} \cup \J_{(1,t+1-m)} \cup \J_{(1,t+1)} \cup \J_{(2,2t-m)} = \Pe_{2} \cup \J_{(1,t+1-m)} \cup \J_{(2,2t-m)}.
\end{split}
\end{equation}
For $t=0$ and $m=1$ (the degree of the instanton polynomial), this is
\begin{equation}
 \supp (\wg_{2}) = \supp (\Rn \wg) \subseteq \Pe_{2} \cup \J_{(1,0)} \cup \J_{(2,-1)}
\end{equation}
and, along the same lines, one gets
\begin{equation}
 \supp (\wg_{3}) = \supp (\Rn \wg_{2}) \subseteq \Pe_{3} \cup \J_{(1,0)} \cup \J_{(2,-1)} \cup \J_{(3,-2)}
\end{equation}
in the second step. This creates a staircase pattern for the lower boundary of the support, as \textsc{Figure} \ref{Rsupport} shows. 

\begin{figure}[ht]
\begin{center} \includegraphics[height=4cm]{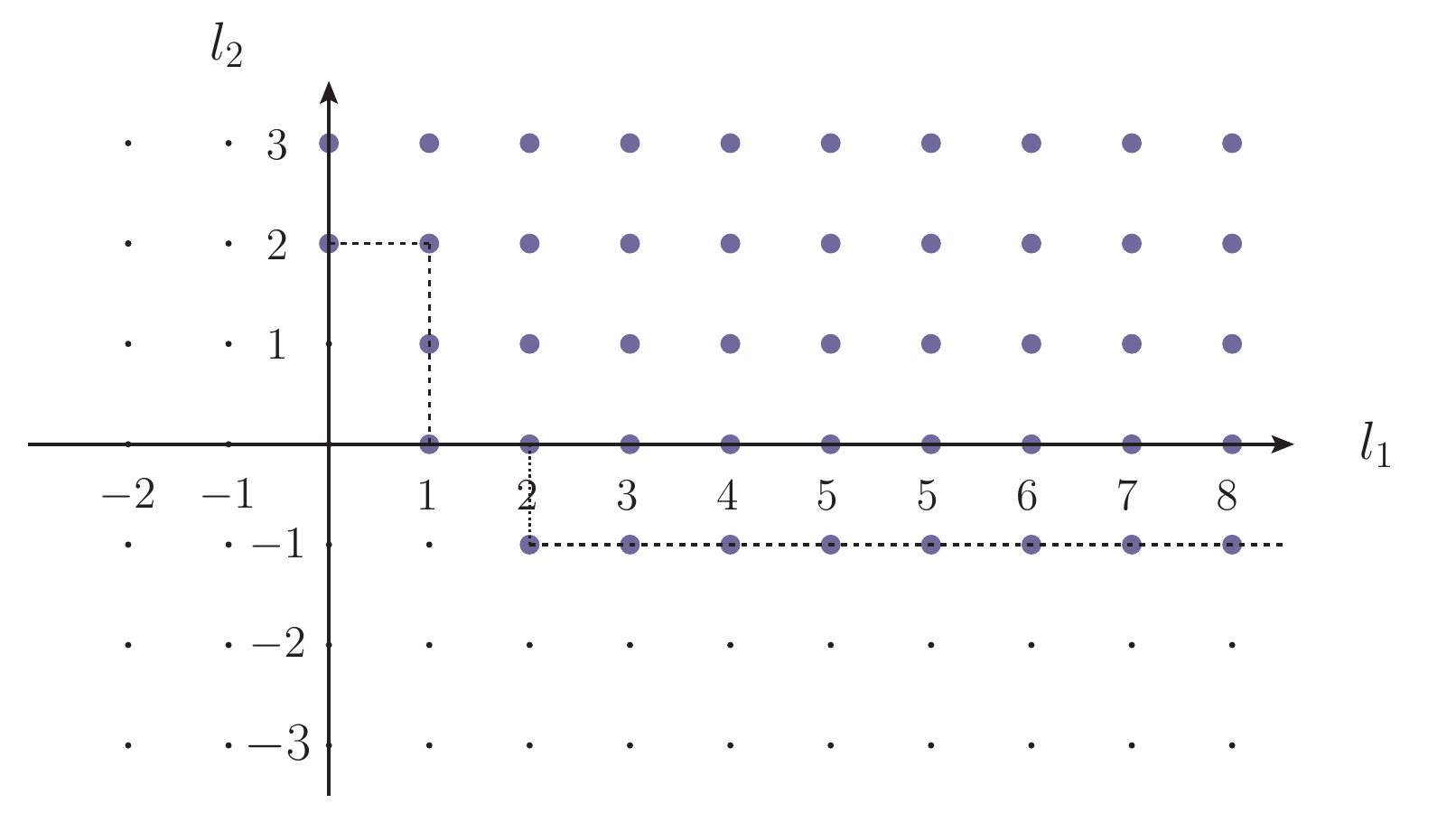} \raisebox{2cm}{$\stackrel{\Rn}{\longrightarrow}$} \includegraphics[height=4.8cm]{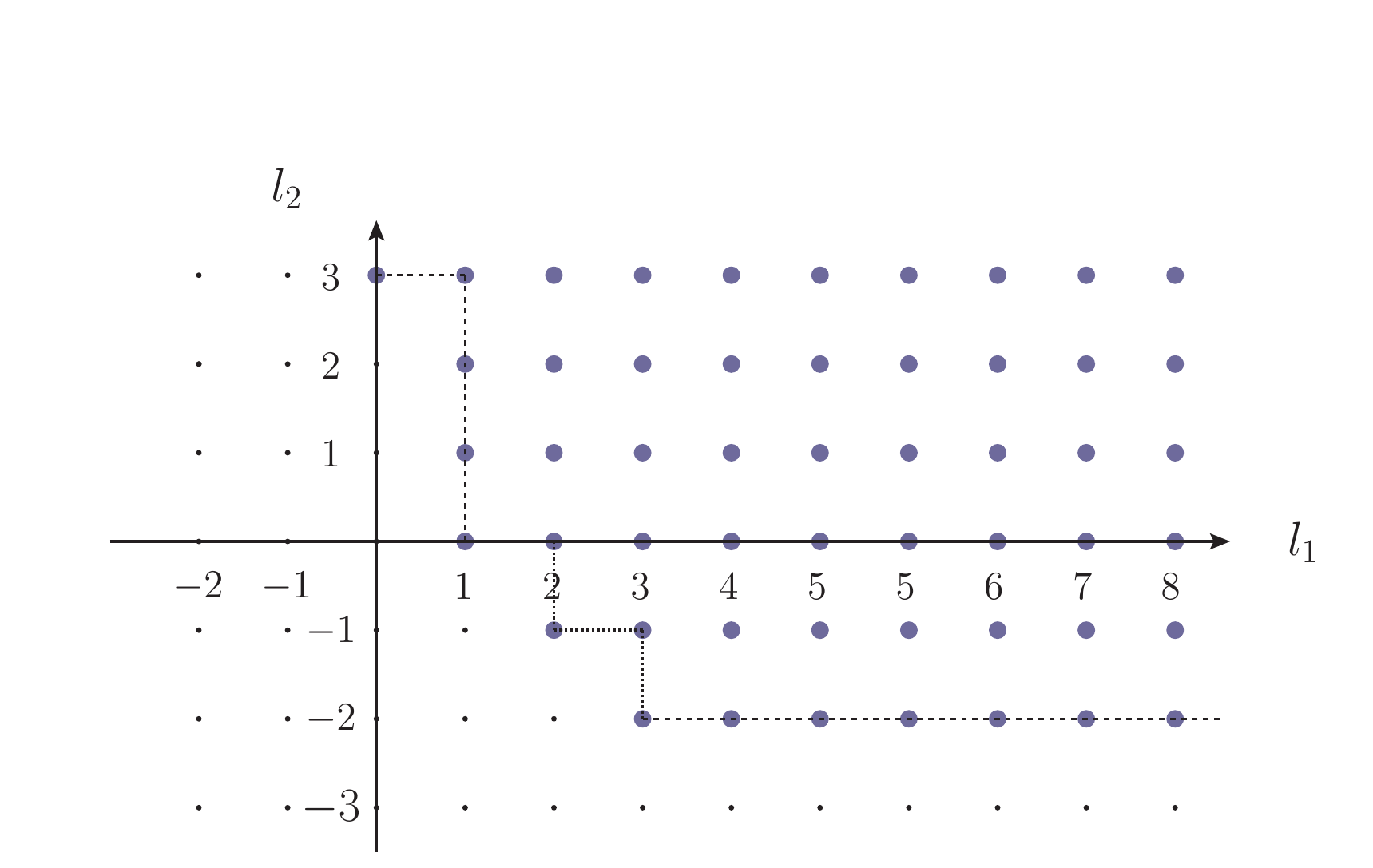} 
\end{center} 
\caption{\small An example of how the RG recursion changes the support along the orbit of the discrete dynamical system for $\deg P(z)=1$. 
The diagrams show how an additional step is produced when $\Rn$ acts on $\wg_{2}$ (left) to arrive at $\wg_{3}$ (right). }
\label{Rsupport}
\end{figure}

As we will see in more detail next, the depth of the individual steps crucially depends on the degree $m$ of the instanton polynomial $P(z)$ which sits
in the exponent of the nonperturbative transmonomial $\mo_{1}$ and changes the support when it comes down upon differentiation. 

\subsection{RG recursion limit}
We will now address the question whether the sequence of transseries $(\wg_{n})_{n \geq 1} \subseteq \Tu $ converges in some sense. In fact, we will 
see that the support of $\wg$ is crucial. Before we study this issue, we review the pertinent notion of convergence from \cite{Ed09}. 

\begin{Definition}\label{Co}
Let $(T_{i})_{i \in \N}$ be a family of transseries in $\Tu$. 
We say of the family $(T_{i})$ that it converges to $0$ and write $T_{i} \rightarrow 0$, if 
\begin{itemize}
 \item [(i)] there exists a fixed grid $\J_{k}$ such that $\suppi(T_{i}) \subseteq \J_{k}$ for all $i \in \N$ and
 \item [(ii)] the sequence $(\suppi(T_{i}))_{i \in \N}$ of supports is point-finite, which means that for any transmonomial $\g$, one 
              has $\g \in \suppi(T_{i})$ for at most a finite number of $i \in \N$.
\end{itemize}
We say $(T_{i})$ converges to $T \in \Tu$ and write $T_{i} \rightarrow T$ if $T_{i} - T \rightarrow 0$. 
\end{Definition}

The convergence $T_{i} \rightarrow 0$ means $\supp(T_{i}) \rightarrow \emptyset$, ie the transseries within the family will eventually be depleted of 
their coefficients as the index $i$ grows and condition (i) ensures that one has some control over the supports, otherwise pathological behaviour 
might occur \cite{Ed09}.  

The next assertion tells us what condition we must impose on the support of $\wg$ if we want the RG recursion to stand a chance of 
converging.

\begin{Proposition}\label{Psupp}
The subspace $\X_{t}=\{ f \in \Tu \, | \, \suppi (f) \subseteq \Pe_{1} \cup \J_{(1,t)} \}$ is an $\Rn$-stable subalgebra if and only if 
$t \geq \deg P(z) =1$. In general, for $f \in \X_{t}$, one has 
\begin{equation}\label{su}
 \suppi (\Rn^{n}f) \subseteq \Pe_{n+1} \cup \bigcup_{j=1}^{n+1} \J_{(j,j[t-1]+1-[m-1]n)} \ , \hs{1}  n \geq 1,
\end{equation}
where $\Rn=f(s\der-1)$ and $m$ is the degree of the instanton polynomial $P(z)$. 
\end{Proposition}
\begin{proof}
We first prove (\ref{su}) by induction whose start $n=1$ has been already performed in (\ref{suppfa}). For the induction step we find
\begin{equation}\label{supp4'}
\begin{split}
 \supp &(\Rn \Rn^{n}f) \subseteq \left(\Pe_{1} \cup \J_{(1,t)} \right) \cdot \left( \Pe_{n+1} \cup \bigcup_{j=1}^{n+1} \J_{(j,j[t-1]+1-[m-1]n-m)} \right) \\
 &= \Pe_{n+2} \cup \bigcup_{j=1}^{n+1} \J_{(j,j[t-1]+1-[m-1][n+1])} \cup \J_{(1,t+n+1)} \cup \bigcup_{j=1}^{n+1} \J_{(j+1,[j+1][t-1]+1-[m-1][n+1])} .
\end{split}
\end{equation}
We perform an index shift $j \rightarrow j-1$ and use
\begin{equation}
 \J_{(j,j[t-1]+1-[m-1][n+1])}|_{j=1}  = \J_{(1,t+n+1-m[n+1])} \supseteq \J_{(1,t+n+1)}
\end{equation}
to get 
\begin{equation}
\supp (\Rn \Rn^{n}f) \subseteq \Pe_{n+2} \cup \bigcup_{j=1}^{n+2} \J_{(j,j[t-1]+1-[m-1][n+1])}.
\end{equation}
The result (\ref{su}) makes one thing very clear: only if $t \geq m=1$ does the supportive subgrid not grow. 
So we see that the subspace $\X_{t}$ is stable under $\Rn$, ie $\Rn \X_{t} \subseteq \X_{t}$ if the condition $t\geq 1=m$ is satisfied. Multiplication 
poses no problem if $t \geq 0$  since
\begin{equation}
\left(\Pe_{1} \cup \J_{(1,t)}\right) \cdot \left(\Pe_{1} \cup \J_{(1,t)}\right) = \Pe_{2} \cup \J_{(1,t+1)} \cup \J_{(1,t+1)} \cup \J_{(2,2t)} 
= \Pe_{2} \cup \J_{(1,t+1)}
\end{equation}
and multiplication would only produce larger supports and thereby lead out of $\X_{t}$ if $t<0$. 
\end{proof}

If we had logs, this result would not change: an additional transmonomial $\mo_{3}=\log z$ leads to an additional term in (\ref{tra}), ie
\begin{equation}\label{tra'}
 \der \left( \mo_{1}^{l_{1}}\mo_{2}^{l_{2}}\mo_{3}^{l_{3}} \right)  
= (cl_{1} + l_{2}) \mo_{1}^{l_{1}}\mo_{2}^{l_{2}}\mo_{3}^{l_{3}} - l_{3} \mo_{1}^{l_{1}}\mo_{2}^{l_{2}}\mo_{3}^{l_{3}-1} \\
+ l_{1} \sum_{i=1}^{m} i b_{i} \mo_{1}^{l_{1}}\mo_{2}^{l_{2}-i}\mo_{3}^{l_{3}}
\end{equation}
but does not change the first two indices and yields 
\begin{equation}\label{logs}
 (\der f)|_{(l_{1},l_{2},l_{3})} = (cl_{1}+l_{2})f_{(l_{1},l_{2},l_{3})} - (l_{3}+1)f_{(l_{1},l_{2},l_{3}+1)} 
 + l_{1} \sum_{i=1}^{m} i b_{i} f_{(l_{1},l_{2}+i,l_{3})} .
\end{equation}
for the coefficients of a transseries $f$, which is why the results of Proposition \ref{Psupp} are untouched as long as one agrees that $\J_{k}$ 
encompasses the logarithmic monomials attached to it. \\

So what lesson can we draw from this? One is certainly 
\begin{Corollary}\label{Rstable}
The RG recursion $(\wg_{n})_{n \geq 1}$ fails to satisfy condition (i) of Definition \ref{Co} unless 
$\suppi(\wg) \subseteq \Pe_{1} \cup \J_{(1,1)}=\J_{(0,1)}$
and $m= \deg P(z) = 1$ in which case one has 
\begin{equation}\label{supp3}
 \suppi(\wg_{n}) \subseteq \Pe_{n} \cup \J_{(1,1)} \subset \J_{(0,1)}
\end{equation}
for all $n \geq 1$.
\end{Corollary}
\begin{proof}
This assertion is a consequence of Proposition \ref{Psupp} which guarantees stability under the RG recursion only if $t \geq m = 1$ which entails for (\ref{su})
\begin{equation}\label{su'}
\supp (\Rn^{n}\wg) \subseteq \Pe_{n+1} \cup \bigcup_{j=1}^{n+1} \J_{(j,j[t-1]+1)} \subset \Pe_{n+1} \cup \bigcup_{j=1}^{n+1} \J_{(j,1)} = \Pe_{n+1} \cup \J_{(1,1)} ,
\end{equation}
for all $n \geq 1$.
\end{proof}

So we conclude that in case the supportive subgrid is not contained in the grid $\J_{(0,1)}$, it may grow badly.
In \textsc{Figure} \ref{Rsupport} we have seen an example of this behaviour for $m=1$ and $t=0$: the RG step from $\wg_{2}$ to $\wg_{3}$ leads to 
$\supp(\wg_{1}) \subsetneq \supp(\wg_{2}) \subsetneq \supp(\wg_{3})$. 
This growth of support may preclude convergence because it has the potential to run out of control. To prevent this from happening, restrictions must be imposed on the 
support of $\wg$: 
\begin{itemize}
 \item the nonperturbative sectors must have power series without a constant term, ie the support has no dots on the $l_{1}$-axis in 
       \textsc{Figure} \ref{support},
 \item the polynomial $P(z)$ must be of degree 1. 
\end{itemize}
However, the support may be distributed over its supportive grid in a way that renders the RG recursion convergent, say by some subtle cancellations.  

Another lesson we draw from Proposition \ref{Psupp} is how the dominating transmonomial changes along the orbit. To understand this, note that 
the lower-sector transmonomials are always larger than higher-sector ones. The perturbative sector's transmonomials are therefore always the largest.
Among them, the one with the smallest power of $\mo_{2}=z^{-1}$ is the largest. 

If we focus on the nonperturbative sectors, then (\ref{su}) informs us that for $\wg_{n}$, the dominating transmonomials of all nonvanishing sectors
$\sigma >1$ must fulfil 
\begin{equation}\label{mag}
\magn ( \mo_{1}^{\sigma} \W_{\sigma}\wg_{n}) \prec \magn ( \mo_{1} \W_{1} \wg_{n}) \preccurlyeq \mo_{1}\mo_{2}^{t-(m-1)(n-1)} 
= z^{-c}e^{-P(z)} z^{(m-1)(n-1)-t}    
\end{equation}
which can be read off from the rhs of (\ref{su}). We will need this result in §\ref{sec:DSE}. 

However, up to this point of the story, and hence excluding the DSEs from the discussion, one has to concede that a growing support along the orbit is 
not problematic as the RG recursion does not have to converge. The RG functions simply are the derivatives of the self-energy with respect to the 
momentum logarithm $L_{q}$ evaluated at $L_{q}=0$, nothing more. Of course, we would wish for the momentum log expansion 
\begin{equation}
 G(x,L_{q}) = 1 - \sum_{n \geq 1} \gamma_{n}(x)L_{q}^{n}
\end{equation}
to converge as we do not desire to consider a nonperturbative completion of this expression with respect to the parameter $L_{q}$. 
However, this requirement can be met even if the RG recursion itself does not converge. But one should be aware that the limit of $\Rn^{n}(\wg)$ as 
$n \rightarrow \infty$ is not grid-based in case the support grows without bounds along the orbit of the dynamical system. 

To get a feel for how the support may grow for $t=0$ in the general case $m \geq 1$ , let us see what concrete form the expression (\ref{su}) takes 
for the first few RG functions:
\begin{equation}\label{su'}
\begin{split}
 \supp (\wg_{1}) & \subseteq \Pe_{1} \cup \J_{(1,0)} \\
 \supp (\wg_{2}) & \subseteq \Pe_{2} \cup \J_{(1,-m+1)} \cup \J_{(2,-m)}, \\ 
 \supp (\wg_{3}) & \subseteq \Pe_{3} \cup \J_{(1,-2m+2)}  \cup \J_{(2,-2m+1)} \cup \J_{(3,-2m)},  \\
 \supp (\wg_{4}) & \subseteq \Pe_{4} \cup \J_{(1,-3m+3)} \cup \J_{(2,-3m+2)} \cup \J_{(3,-3m+1)} \cup \J_{(4,-3m)} ,
\end{split}
\end{equation}
from which we see that, as already mentioned, a staircase pattern with a deep first step is created and all sectors get additional 
transmonomials. \textsc{Figure} \ref{rsupport} illustrates this situation.

\begin{figure}[ht]
\begin{center} \includegraphics[height=4.1cm]{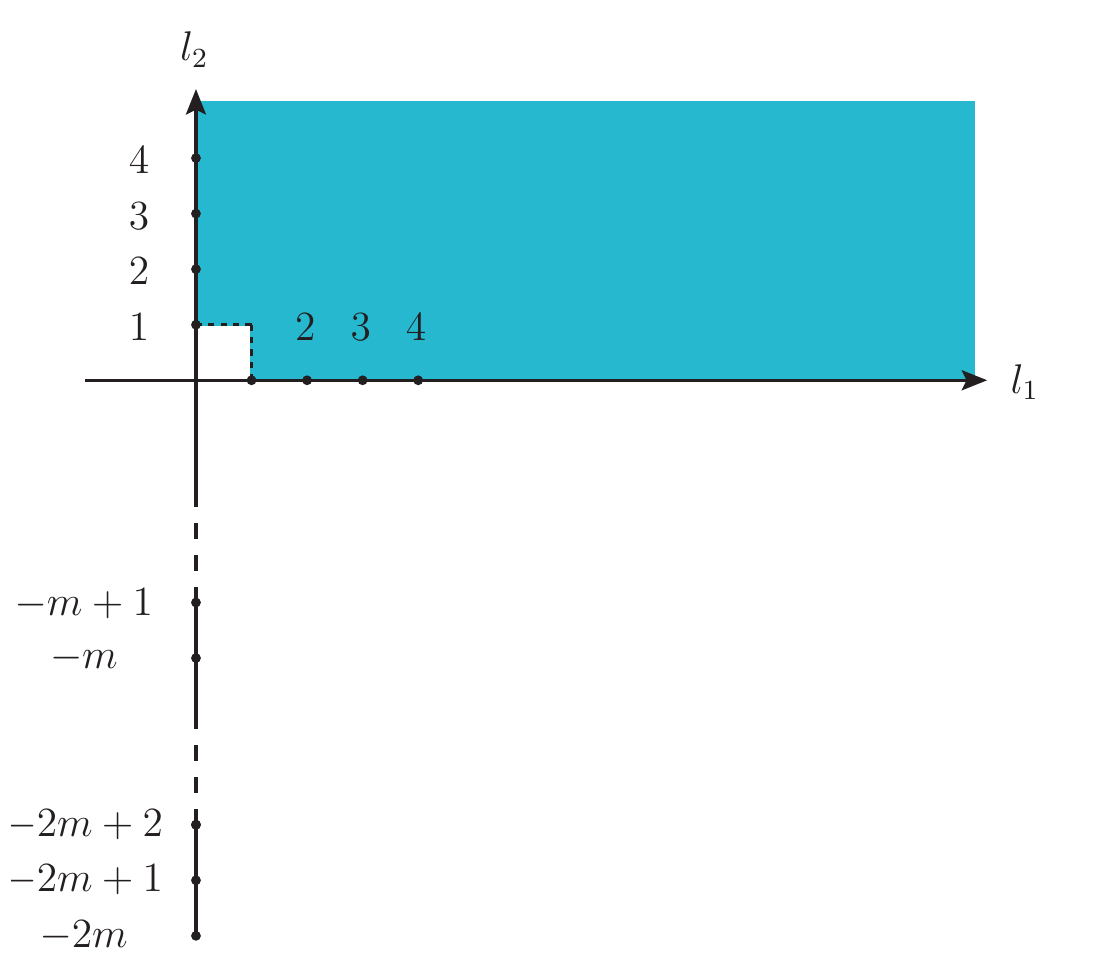} \raisebox{2cm}{$\stackrel{\Rn}{\longrightarrow}$} 
\includegraphics[height=4.1cm]{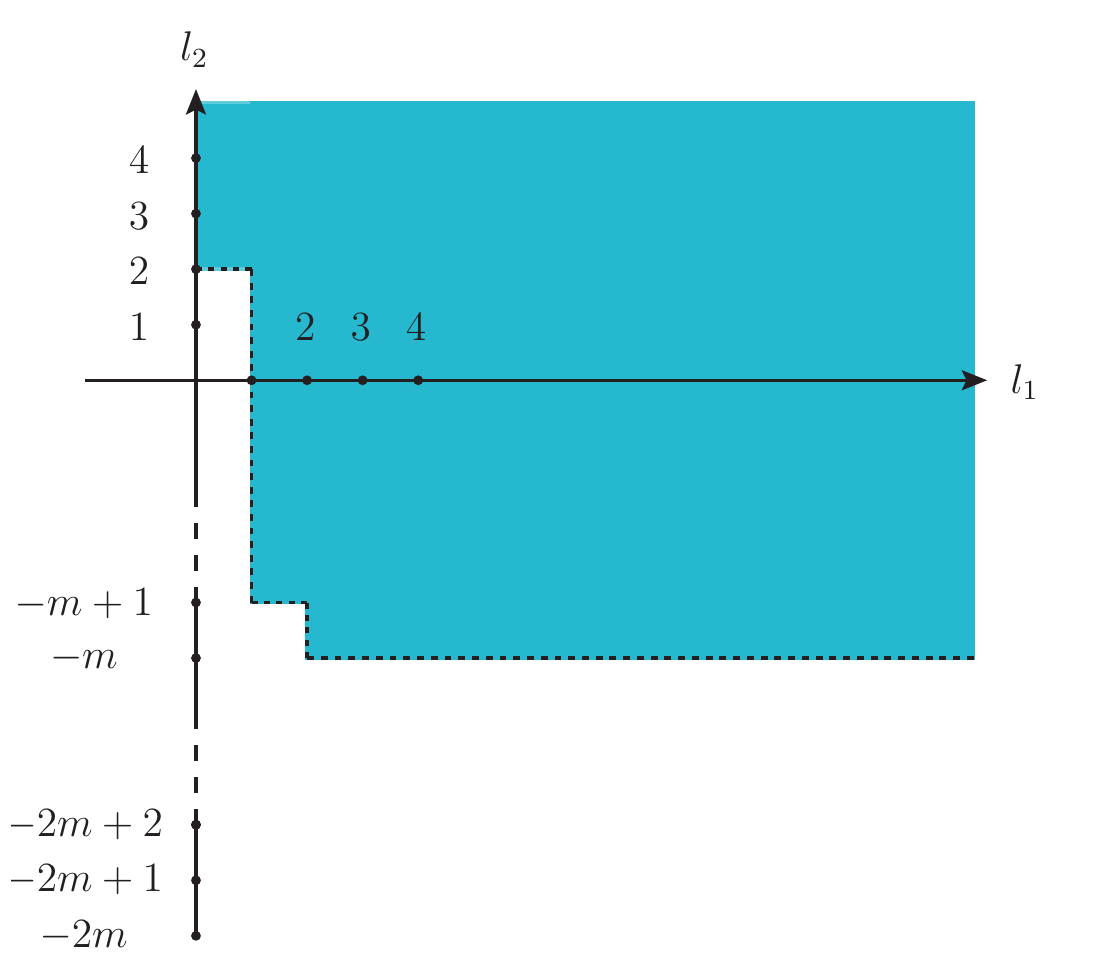} \raisebox{2cm}{$\stackrel{\Rn}{\longrightarrow}$} 
\includegraphics[height=4.1cm]{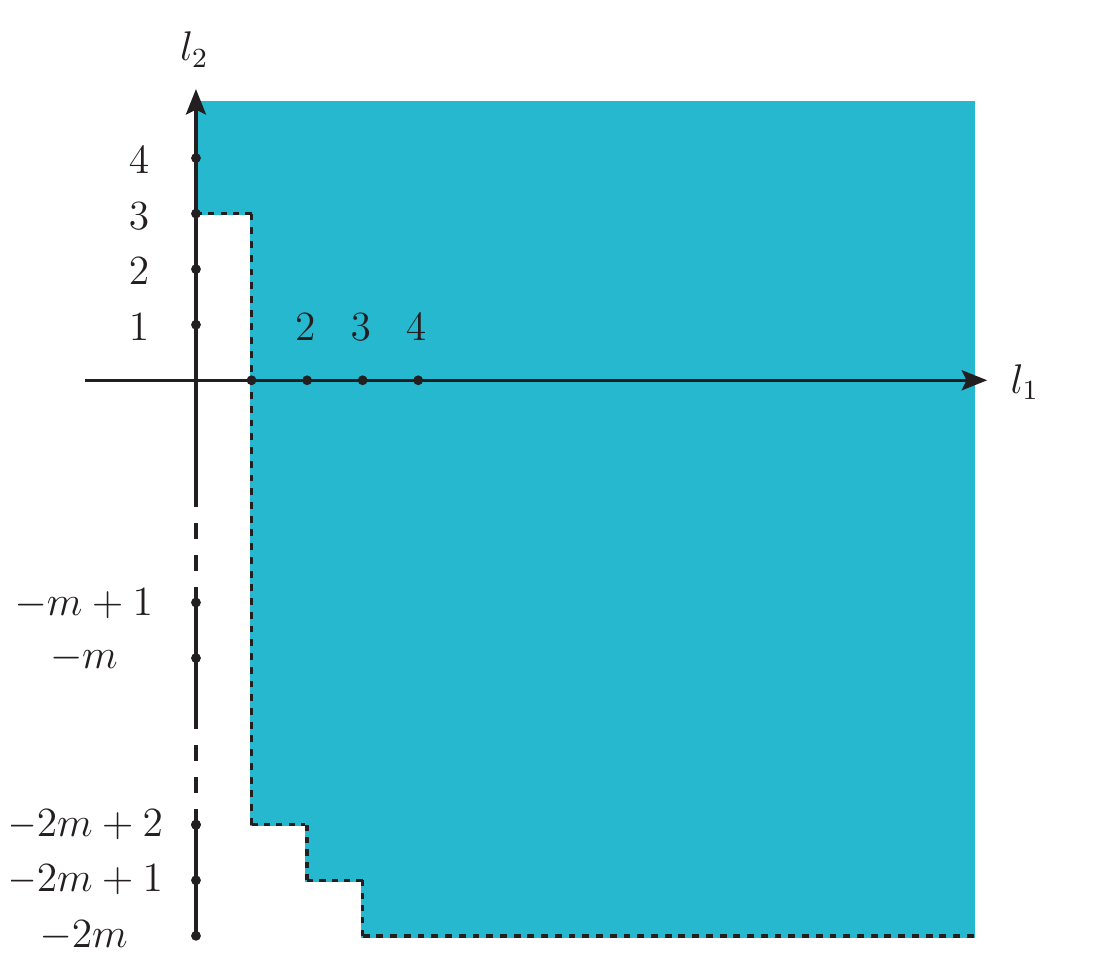} \end{center} 
\caption{\small Bird's eye view on how the RG recursion changes the support along the orbit of the discrete dynamical system in the case of a more
general transseries ansatz from $\wg=\wg_{1}$ (left) through $\wg_{2}$ (centre) to $\wg_{3}$ (right). 
The depth of the first step depends on $m=\deg P(z)$, ie the degree of the 'instanton polynomial' $P(z)$. See also \textsc{Figure} \ref{Rsupport}.}
\label{rsupport}
\end{figure}

However, an RG-stable support is only a necessary condition anyway: for true convergence in the sense of Definition \ref{Co},
we need $\supp(\wg_{n}-\wg_{n+1}) \rightarrow \emptyset$ (through point finiteness) which is an even stronger requirement. 

\section{Nonuniqueness of fixed points}\label{sec:DSE}                                    
In this section, we will study the DSEs derived in §\ref{sec:DySch} and discuss whether a pertinent fixed point theorem for grid-based transseries applies.
The peculiar situation we are in is that we seek a truly nonperturbative solution that goes beyond the known perturbative one but at the same time 
encompasses it. Since a fixed point theorem guaranteeing a unique solution can only thwart this goal, we prefer to stay clear of suchlike.  

\subsection{Fixed point equation as an asymptotic constraint}
We take the photon DSE (\ref{phodse1}) and recast it in the transseries setting as a fixed point equation in the algebra $\Tu$ to get 
\begin{equation}\label{dyson}
\wg = A_{0}\mo_{2} + \sum_{\ell=1}^{N} \mo_{2}^{\ell+1} \sum_{r_{1}\geq 0, n_{1} \geq r_{1}} \dotsc \sum_{r_{\ell}\geq 0, n_{\ell} \geq r_{\ell}} 
  C_{(n_{1},\dotsc, n_{\ell})} (\wg^{\star r_{1}}_{\bullet})_{n_{1}} \dotsc (\wg^{\star r_{\ell}}_{\bullet})_{n_{\ell}} \hs{1} (\text{photon}),
\end{equation}
where the convolutions can be written by means of the RG recursion $\wg_{n}=\Rn^{n-1}(\wg)$ as
\begin{equation}
 (\wg^{\star r}_{\bullet})_{m} =  \sum_{m_{1} + \dotsc + m_{r}=m} \frac{1}{m_{1}!} \Rn^{m_{1}-1}(\wg) \dotsc \frac{1}{m_{r}!} \Rn^{m_{r}-1}(\wg)
\end{equation}
for $m \geq 1$. Note that this is a finite linear combination of grid-based transseries in $\Tu$ and hence the assignment
$\wg \mapsto (\wg^{\star r}_{\bullet})_{m}$ is a well-defined operator in $\Tu$ , which is then also true of the operator 
\begin{equation}\label{Qa}
\begin{split}
 \Qa_{M}(\wg) := \sum_{n=0}^{M} \sum_{r=0}^{n} C_{n} (\wg^{\star r}_{\bullet})_{n}   = C_{0} + C_{1} \wg 
 & + C_{2} \left[ \frac{\wg_{2}}{2!}  + \wg \cdot \wg \right] + C_{3} \left[  \frac{\wg_{3}}{3!} +  2 \frac{\wg_{2}}{2!} \cdot \wg + 
 \wg \cdot \wg \cdot \wg \right] \\ 
 & + \dotsc + C_{M} \left[ \left(\wg^{\star 1}_{\bullet}\right)_{M} + \dotsc + \left(\wg^{\star M}_{\bullet}\right)_{M} \right].
 \end{split}
\end{equation}
The case $N=1$ of the photon DSE in (\ref{dyson}) arises as the limit
\begin{equation}\label{Kil}
 \wg = A_{0} \mo_{2} + \mo_{2}^{2} \lim_{M \rightarrow \infty}\Qa_{M}(\wg)
 = A_{0} \mo_{2} + \mo_{2}^{2} \left( C_{0} + C_{1} \wg + C_{2} \left[ \frac{\wg_{2}}{2!} + \wg \cdot \wg \right] + \dotsc \right) \hs{0.1} (\text{photon})
\end{equation}
where only the 2-loop skeleton (\ref{phot}) is taken into account. Notice that this is very close to the Kilroy DSE for the anomalous 
dimension (\ref{Kilanom1}),  
\begin{equation}\label{Kil'}
 \wg = \mo_{2} \lim_{M \rightarrow \infty}\Qa_{M}(\wg) = \mo_{2} \left( C_{0} + C_{1} \wg 
 + C_{2} \left[ \frac{\wg_{2}}{2!} + \wg \cdot \wg \right] + \dotsc \right) \hs{1} (\mbox{Yukawa fermion}) .
\end{equation}
If one desires to take more skeletons into account for the photon, one has to consider the operator family 
\begin{equation}
 \Qa_{M_{1},\dotsc,M_{\ell}}(\wg) := \sum_{n_{1}=0}^{M_{1}} \dotsc \sum_{n_{\ell}=0}^{M_{\ell}} \sum_{r_{1}=0}^{n_{1}} \dotsc \sum_{r_{\ell}=0}^{n_{\ell}}
 C_{(n_{1}, \dotsc, n_{\ell})} (\wg^{\star r_{1}}_{\bullet})_{n_{1}} \dotsc (\wg^{\star r_{\ell}}_{\bullet})_{n_{\ell}}  
\end{equation}
such that the DSE (\ref{dyson}) is given by the limit  
\begin{equation}\label{dyson1}
\wg = A_{0}\mo_{2} + \lim_{M_{1},\dotsc,M_{N} \rightarrow \infty} \sum_{\ell=1}^{N} \mo_{2}^{\ell+1} \Qa_{M_{1},\dotsc,M_{\ell}}(\wg) 
\hs{1} (\mbox{photon DSE in $\Tu$}).
\end{equation}
This is a fixed point equation for $\wg$ which, in the picture of the RG recursion as a dynamical system, is at the same time an asymptotic constraint
imposed on the orbit starting at $\wg$ (the initial value of the dynamical system): the operators $\Qa_{M_{1},\dotsc,M_{\ell}}$ produce finite linear 
combinations of products of transseries from the orbit $\Phi(\N, \wg) \subset \Tu$ up to time step $M = \max \{ M_{1},\dotsc,M_{N} \}$, 
ie $\Qa_{M_{1},\dotsc,M_{\ell}}(\wg)$ is a polynomial in $\wg_{1}, \dotsc, \wg_{M}$, which, from a structural viewpoint, is nothing but a product of 
several $\Qa_{M}(\wg)$'s. Although the coefficients $C_{0}, C_{1} , \dotsc$ are in the Yukawa case different from those in the photon case, we are not 
interested in this aspect as it is of no relevance for our investigation. We therefore do not use extra notation for $\Qa_{M}$ (Yukawa) and 
$\Qa_{M_{1},\dotsc,M_{\ell}}$ with $\ell=1$ (QED). 

\subsection{Evasion of fixed point theorems} 
We define two families of 'Dyson-Schwinger' (DS) operators on $\Tu$, one for the Kilroy DSE  
\begin{equation}\label{DSKil}
 \Psi_{M}(f) := \mo_{2} \Qa_{M}(f)   \hs{5} (\text{Yukawa fermion})
\end{equation}
and one for the photon DSE
\begin{equation}\label{DSPho}
\Psi^{\graph{0}{0.04}{0.08}{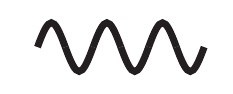}{0}}_{M_{1},\dotsc,M_{N}}(f) :=  A_{0}\mo_{2} + \sum_{\ell=1}^{N} \mo_{2}^{\ell+1} \Qa_{M_{1},\dotsc,M_{\ell}}(f)
\hs{2} (\text{photon}).
\end{equation}
There is a wish we have about these operators, concerning their fixed points: if $f,h \in \Tu$ are fixed points, ie $\Psi_{M}(f)=f$
and $\Psi^{\graph{0}{0.04}{0.08}{pho}{0}}_{M_{1},\dotsc,M_{N}}(h)=h$, then we demand 
\begin{equation}
\Psi_{M}(\W_{0}f)=\W_{0}f, \hs{2} \Psi^{\graph{0}{0.04}{0.08}{pho}{0}}_{M_{1},\dotsc,M_{N}}(\W_{0}h) = \W_{0}h ,
\end{equation}
that is, its perturbative part must satisfy the DSE by itself, without its nonperturbative part. The message of this wish is clear: we do not 
desire the DS operators to have a unique fixed point because that would be $\W_{0}\wg \in \Tu$, the perturbative series of the anomalous dimension.

Before we tackle the question whether we can preclude a unique fixed point, we need some more terminology from \cite{Ed09}. 

We have introduced the order relation $\leq$ for Hahn series in §\ref{sec:ReTra} based on the coefficient of the dominating transmonomial. Here is another 
one which introduces a partial order: for two transseries $S,T \in \Tu$, we write $S \preccurlyeq T$ if their dominating transmonomials satisfy
\begin{equation}
\magn(S) \preccurlyeq \magn(T) .
\end{equation}
If $\magn(S) = \magn(T)$, one writes $S \asymp T$. If $\magn(S) \neq \magn(T)$ and $S \preccurlyeq T$, then, of course $\magn(S) \prec \magn(T)$ and 
we write $S \prec T$ and say that $T$ dominates over $S$.

\begin{Definition}[Contractivity]\label{Cont}
Let $\U \subset \Tu$ be some subset. A map $\Psi$ from $\U$ into itself is called \emph{contractive} if for any $S,T \in \U$ 
with $S \neq T$ one has
\begin{equation}
 \Psi(S)-\Psi(T) \prec S-T .
\end{equation}
\end{Definition}

Notice that this property is precisely the last we thing we want for the above two DS operator families. The reason is easy to understand: assume $\Psi$
is a DS operator and 
\begin{equation}
\wg = \W_{0}\wg + \sum_{\sigma \geq 1}\mo_{1}^{\sigma}\W_{\sigma}\wg
\end{equation}
a fixed point such that moreover $\Psi(\W_{0}\wg)=\W_{0}\wg$, then
\begin{equation}
 \Psi(\wg) - \Psi(\W_{0}\wg) = \wg - \W_{0}\wg    
\end{equation}
which contradicts contractivity because $\wg \neq \W_{0}\wg$. This means if we find a subset that contains the perturbative solution and on which the 
DS operators are contractive, then we can be sure that it does \emph{not} contain a nonperturbative completion of the anomalous dimension. 

The next assertion tells us when a contractive operator has a fixed point. It is a general result about grid-based transseries from \cite{Ed09} which also 
holds for our transseries algebra $\Tu$.   

\begin{Theorem}[Fixed point theorem]\label{Fix}
Fix $k \in \Z^{2}$ and let $\U \subset \Tu$ be a subset of transseries $T$ with $\emph{\supp}(T) \subset \J_{k}$. If $\Psi$ is contractive 
in $\U$, as in Definition \ref{Cont}, then there exists a unique fixed point $T \in \U$ of $\Psi$, ie $T = \Psi(T).$
\end{Theorem}
\begin{proof}
See \cite{Ed09}, Theorem 4.22. 
\end{proof}

The results of §\ref{sec:DynSys}, in particular Proposition \ref{Psupp}, show that the proviso for the supportive grid is easily violated because the 
support seems to have a strong tendency to grow along the orbit if the degree of the instanton polynomial is at least quadratic. 
This should clearly effect contractivity, especially because multiplying by $\mo_{2}$, itself a contractive operation, cannot compensate for the 
support-increasing effect of the RG recursion. 
There are two ways to respond to this situation: either 
\begin{itemize}
 \item one considers a subset on which the DS operators are contractive and which excludes the perturbative solution or  
 \item one focusses on subsets which contain it at the cost of losing contractivity. 
\end{itemize}
So the task is to find a subset on which contractivity is given or not given. However, this turned out to be harder than originally expected. The intuition 
gained in §\ref{sec:DynSys} from Proposition \ref{Psupp} suggests that the DS operator may produce very large transmonomials if $m=\deg P(z)>1$. 
Let us pick an example to see how contractivity may easily be broken and that for this reason we might as well stick to the unconventional idea of 
evading the above fixed point theorem.

If we assume $\ti{\gamma}_{0} \in \C[[z^{-1}]]$ is a perturbative fixed point, ie the solution in principle given in terms of Feynman diagrams,
then we are interested in the subset of nonperturbative completions of $\ti{\gamma}_{0}$,
\begin{equation}
\text{Com}(\ti{\gamma}_{0}):= \left\{ \ \wg \in \Tu \, | \, \W_{0}\wg = \ti{\gamma}_{0} , \exists \sigma \geq 1 : \W_{\sigma} \wg \neq 0 \ \right\},
\end{equation}
ie of precisely those transseries with nonvanishing nonperturbative part whose perturbative sector agrees with the perturbative solution $\ti{\gamma}_{0}$. 
This set does not contain the perturbative solution and may in a case of contractivity have a unique solution. But it does not take much to find that 
this property is hard to be fulfilled. 
Let us take two candidates $f,h \in \text{Com}(\ti{\gamma}_{0}) \cap \X_{t}$, then
\begin{equation}
  \magn(f-h) \preccurlyeq \mo_{1} \mo_{2}^{t},
\end{equation}
ie whichever the largest transmonomial they disagree on, it is smaller than $\mo_{1} \mo_{2}^{t}$. (\ref{mag}) obtained from Proposition \ref{Psupp} 
promises us that the RG recursion may produce very large transmonomials, at least we can say that 
\begin{equation}
\magn \{ \Rn^{n} f \} = \magn \{ (f[s \der-1])^{n}f \} \succcurlyeq \mo_{1} \mo_{2}^{t-m+1} \succcurlyeq \mo_{1} \mo_{2}^{t}
\end{equation}
is certainly not too strong a condition for some large enough $n \geq 1$. Suppose 
\begin{equation}\label{large}
\magn(\Rn^{n} f- \Rn^{n} h ) \succcurlyeq \mo_{1} \mo_{2}^{t-m+1} 
\end{equation}
again for some large enough $n \geq 1$, where $m = \deg P(z)$ (see Proposition \ref{Psupp}). This is all we need for the following argument. Let $v \geq 1$
and take $M \geq n$ to be large enough, then 
\begin{equation}
\begin{split}
 \magn[\mo_{2}^{v} & (\Qa_{M}f-\Qa_{M}h)]  = \mo_{2}^{v} \magn (C_{1}[f-h] + C_{2}[\Rn f- \Rn h] + \dotsc) \\
 &= \mo_{2}^{v} \magn (C_{2}[\Rn f- \Rn h] + \dotsc) \succcurlyeq \mo_{1}\mo_{2}^{t-m+1+v} \succcurlyeq \mo_{1}\mo_{2}^{t} = \magn(f-h)
\end{split}
\end{equation}
if $m \geq 1+v$, refuting contractivity. The dots $\dotsc$ involve the higher powers of $\Rn$ which create the additional larger 
transmonomials that cannot be cancelled, even if some smaller ones disappear by cancellations. The resulting monomial is therefore larger or as 
large as the first term.

The case $N>1$ for QED is not much different. The operator $\Qa_{M_{1},\dotsc,M_{\ell}}(\wg)$ starts with $\Qa_{M}(\wg)$ and essentially adds powers 
thereof times higher powers of $\mo_{2}$. Driving up the parameters $M_{1},\dotsc,M_{\ell}$ introduces additional ever larger transmonomials down the abyss
as it evolves in \textsc{Figure} \ref{rsupport} that cannot be cancelled and made small by the $N$-th (ie the largest) power of $\mo_{2}$, thus nothing 
can interfere with the '$\succcurlyeq$' conclusion.

Because the case $v=1$ occurs in the Yukawa model and $v \geq 2$ in the QED model, we conclude that the conditions 
\begin{enumerate}
 \item [(C1)] $m=\deg P(z) \geq 2$ (Yukawa fermion)
 \item [(C2)] $m=\deg P(z) \geq 3$ (photon)
\end{enumerate}
for large enough $M, M_{1}, \dotsc, M_{N} \in \N$ and any $t \in \Z$ suffice for the DS operator \emph{not to be contractive}.

Finally, let us now come to another subtle point: the limits
\begin{equation}
\wg = \lim_{M \rightarrow \infty} \Psi_{M}(\wg) , \hs{2} 
\wg = \lim_{M_{1},\dotsc,M_{N} \rightarrow \infty} \Psi^{\graph{0}{0.04}{0.08}{pho}{0}}_{M_{1},\dotsc,M_{N}}(\wg)
\end{equation}
which we must take in the end, are somewhat fishy. The case $N=1$ already has everything it takes for us to grasp what the problem is: whether or not the 
support of $\Qa_{M}(\wg)$ grows with $M$, this object becomes an infinite sum of transseries. For an infinite sum of transseries to be well-defined, 
one normally requires the sequence of its terms to vanish: for a sequence of transseries $(A_{j})_{j \geq 1} \in \Tu$, and 
\begin{equation}
A=\sum_{j \geq 1} A_{j} 
\end{equation}
one prefers $A_{j} \rightarrow 0$, to ensure that every coefficient of the resulting transseries $A$ actually exists, ie one demands that the sum 
$A(\mo) = \sum_{j \geq 1} A_{j}(\mo) \in \C$ terminate. But, as ever so often, physics does not do us this favour, as the subseries 
\begin{equation}
\sum_{n=1}^{M} C_{n} \left(\wg_{\bullet}^{\star 1}\right)_{n} =  \sum_{n=1}^{M} \frac{C_{n}}{n!} \wg_{n}
\end{equation}
of $\Qa_{M}(\wg)$ in (\ref{Qa}) shows: at least apriori, every RG function contributes a number to each coefficient of the sum. This may or may not 
produce a number in the limit $M \rightarrow \infty$. 
However, we take the physicist's viewpoint and assume that is must, since after all physics has given us this equation. 

In summary, we have in this section found that $\deg P(z) \geq 2$ for the Yukawa and $\geq 3$ for the QED model ensures at least that the DS operator is 
not contractive, leaving room for a transseries solution \emph{beyond the perturbative one}.

\section{Discrete RG flow of (non)perturbative data}\label{sec:RGflow}  	          

In what follows next, we shall investigate how the RG recursion passes on information from the sectors of $\wg$ to the sectors of the higher
RG functions' transseries $\wg_{n}$. We will see in this section that it happens in a very specific and orderly fashion. Why it is worth our attention
will then become clear when we come back to the DSE in §\ref{sec:PinNP}. For the moment suffice it to say that it is necessary in order to eventually 
monitor how the sectors of $\wg$ communicate amongst each other.

\subsection{Graded algebra of coefficients} To monitor the flow of information along the orbit of the RG recursion 
$\wg = \wg_{1} \rightsquigarrow \wg_{2} \rightsquigarrow \dotsc $, we take the coefficients of $\wg$ and view them as abstract generators of a free 
commutative algebra. 

Let $\Ge= \{ c_{l} \}_{l \in \N^{2}}$ be the set of generators and $\A:= \C[\Ge]$ the corresponding free polynomial algebra over $\C$. One may think of
them as an infinite set of polynomial variables, indexed by pairs $l=(l_{1},l_{2}) \in \N^{2}$ (the convention here is $0 \in \N$). 

Let $\Theta \colon \A \rightarrow \C$ be an algebra morphism such that the generators are mapped to the corresponding coefficients of $\wg$, 
in signs, $\Theta(c_{l}) = [\mo^{l}]\wg$. In particular, the anomalous dimension is represented by a transseries
\begin{equation}
 \og := \sum_{l \geq (0,0)}c_{l} \mo^{l}
\end{equation}
such that $\Theta(\og)=\wg$. Let $\Tu_{\A}:= \Tlm{\A}{\mo^{l} : l \in \Z^{2}}$ be the algebra of our transseries with coefficients in 
the algebra $\A$ instead of $\C$. 

We get the higher RG functions by acting the RG operator $\Rn = \og(z) (s \der-1)$ on the abstract transseries $\og$, ie $\og_{n+1}:=\Rn^{n}(\og)$, 
where we do not use extra notation for $\Rn$ in this new transseries algebra. The transseries representation of the RG function $\gamma_{n}$ is written as
\begin{equation}
 \og_{n} = \sum_{l \in \N \times \Z} \og_{n}|_{l} \mo^{l} \hs{2} n \geq  1,
\end{equation}
where $[\mo^{l}]\og = \og_{n}|_{l}$ denotes the $l$-th coefficient, clearly an element in $\A$.  

Because $\Theta$ is an algebra morphism, it maps $\og_{n}$ to the corresponding $n$-th RG function $\wg_{n}$, ie $\Theta(\og_{n})=\wg_{n}$. 
The reason we have introduced transseries with abstract coefficients is that in 
contrast to the RG recursion in the algebra $\Tu=\Tlm{\C}{\mo^{l} : l \in \Z^{2}}$ with its 'oblivious' operations, we will be able to tell 
in $\Tu_{\A}$ which sector of the transseries $\og$ the coefficients of $\og_{n}$ come from. To this end, however, we need another tool.

\begin{Definition}[Instanton grading] Let $\Upsilon \colon \A \rightarrow \A$ be a derivation defined on the generators by $\Upsilon(c_{l}):=l_{1} c_{l}$ and let
\begin{equation}\label{instg}
 \A = \bigoplus_{\sigma \geq 0} \A_{\sigma} = \A_{0} \oplus \A_{1} \oplus \A_{2} \oplus \dotsc
\end{equation}
be the corresponding \emph{grading} it gives rise to, ie the eigenspaces of $\Upsilon$ ($a \in \A_{\sigma} :\Leftrightarrow \Upsilon(a)=\sigma a$). 
We refer to $\Upsilon$ as \emph{instanton grading operator} and (\ref{instg}) as \emph{instanton grading}. The elements of $\A_{\sigma}$ are said to be  
\emph{homogeneous of degree} $\sigma$.
\end{Definition}

For readers unfamiliar with gradings, here is an example: take the element $g = 3 c_{k}c_{l} + c_{u} \in \A$ with multi-indices $k,l,u \in \Z^{2}$, then
\begin{equation}
 \Upsilon(g) = \Upsilon(3 c_{k}c_{l}) + \Upsilon(c_{u}) = 3 \Upsilon(c_{k})c_{l} + 3 c_{k} \Upsilon(c_{l}) + u_{1} c_{u} = 3 (k_{1}+l_{1}) c_{k}c_{l} + u_{1} c_{u}.
\end{equation}
This is homogeneous of degree $u_{1}$ if and only if $u_{1}= k_{1}+l_{1}$ because then one has $\Upsilon(g)=u_{1}g$. Note what this tells us: 
if $g \in \A_{u_{1}}$ is a coefficient of an RG function $\og_{n}$ and $u_{1} < \sigma \in \N$, then we can be sure that the element $g$ does not 
originate from the anomalous dimension's instanton sectors $\geq \sigma$. 

\subsection{Sector-homogeneous transseries}
There is a distinguished subspace $\Ho(\Tu_{\A}) \subset \Tu_{\A}$ of what we call \emph{sector-homogeneous transseries}, 
best characterised by the direct sum decomposition
\begin{equation}
 \Ho(\Tu_{\A}) = \bigoplus_{\sigma \geq 0} \mo_{1}^{\sigma} \Tlm{\A_{\sigma}}{\mo_{2}^{s} : s \in \Z} \subset \Tu_{\A}
\end{equation}
and consisting of transseries $f=\sum_{\sigma \geq  0} \mo_{1}^{\sigma} \W_{\sigma}[f]$ with the property that
$\W_{\sigma}[f] \in \Tlm{\A_{\sigma}}{\mo_{2}^{s} : s \in \Z}$ for all sectors $\sigma \geq 0$, ie the coefficients of $f$ are homogeneous of the 
degree that corresponds exactly to their sector,
\begin{equation}
 f = \sum_{l \in \N \times \Z } f_{l} \mo^{l} \in \Ho(\Tu_{\A}) \hs{1} : \Longleftrightarrow \hs{1} \Upsilon(f_{l})=l_{1}f_{l} \hs{0.5} \forall 
 l =(l_{1},l_{2}) \in \N \times \Z.
\end{equation}
Notice what a distinguished class these transseries belong to: their coefficients, ie elements in $\A$,  
are homogeneous with respect to the instanton grading and, on top, of exactly the degree that corresponds to the transseries' 
sector\footnote{This is not to be confused with \'Ecalle's 
definition of homogeneous transseries in \cite{Eca04}, which in contrast refers to a uniform exponential depth of the 
support ('depth' is called height here in §\ref{sec:ReTra}, adopted from \cite{Ed09}). In QFT, any nonperturbative completion of a perturbative series
is by definition not homogeneous in \'Ecalle's sense.}. 
The next assertion is straightforward yet crucial.

\begin{Lemma}\label{Alg}
The subspace of sector-homogeneous transseries $\Ho(\Tu_{\A}) \subset \Tu_{\A}$ is a $\der$-invariant subalgebra. 
\end{Lemma}
\begin{proof}
 Let $f,g \in \Ho(\Tu_{\A})$, then $f \cdot g \in \Ho(\Tu_{\A})$ on account of 
 \begin{equation}
 \Upsilon(f \cdot g)|_{l} = \sum_{l'+l''=l}\Upsilon(f_{l'}g_{l''}) = \sum_{l'+l''=l} [\Upsilon(f_{l'})g_{l''} + f_{l'}\Upsilon(g_{l''})] 
 = \sum_{l'+l''=l} [l'_{1}+l''_{1}]f_{l'}g_{l''} = l_{1} (f \cdot g)|_{l}. 
 \end{equation}
To see that $\der \Ho(\Tu_{\A}) \subseteq \Ho(\Tu_{\A})$, we use (\ref{deri}) and compute 
\begin{equation}
 \Upsilon(\der f)|_{l}=\Upsilon \left([l_{1}c +l_{2}] f_{l} + l_{1} \sum_{i=1}^{m} i b_{i} f_{(l_{1},l_{2}+i)} \right) =  
 l_{1}[l_{1}c +l_{2}] f_{l} + l_{1}^{2} \sum_{i=1}^{m} i b_{i} f_{(l_{1},l_{2}+i)} =  l_{1} (\der f)|_{l}
\end{equation}
because the index shift does not change the first index, ie $\Upsilon f_{(l_{1},l_{2}+i)}= l_{1}f_{(l_{1},l_{2}+i)}$.  
\end{proof}
Note that logarithmic transmonomials $\mo_{3}=\log z$ do not disrupt sector homogeneity for the same reason: (\ref{tra'}) implies that the presence 
of a log power leads to an additional term due to a shift in the third index, not the first. Therefore, sector homogeneity is not destroyed by polynomials in logs. 

Can sector homogeneity ever be damaged? If we had more complicated sectors, say perhaps characterised by a double index $\sigma=(\sigma_{1},\sigma_{2})$ for 
the two generating transmonomials 
\begin{equation}\label{super}
\mathfrak{n}_{1} = e^{-z} \ , \hs{2} \mathfrak{n}_{2} = e^{-e^{z}},
\end{equation}
then 
\begin{equation}
\der\left( e^{-\sigma_{1} z} e^{-\sigma_{2} e^{z}} \right) = \sigma_{1} z e^{-\sigma_{1} z} e^{-\sigma_{2} e^{z}}  + \sigma_{2} z e^{-(\sigma_{1}-1) z} e^{-\sigma_{2} e^{z}}
\end{equation}
informs us that the nonperturbative sector $(0,\sigma_{2})$ will receive the transmonomial $\sigma_{2} z e^{-\sigma_{2} e^{z}}$ from sector 
$(1,\sigma_{2})$ as a result of the action of $\der$. Who knows whether this might one day be relevant?

However, we carry on and note that Lemma \ref{Alg} entails an important  

\begin{Proposition}[Discrete RG flow]\label{flow}
The subalgebra $\Ho(\Tu_{\A})$ is $\Rn$-invariant, ie 
\begin{equation}
 \Rn \Ho(\Tu_{\A}) \subseteq \Ho(\Tu_{\A}).
\end{equation}
Consequently, the flow of the RG recursion preserves sector homogeneity. 
\end{Proposition}
\begin{proof}
This is a simple consequence of $\og \in \Ho(\Tu_{\A})$, true by definition, and Lemma \ref{Alg} because the RG operator $\Rn$ does only involve 
multiplication with $\og$ and the action of $s \der - 1$. 
\end{proof}

With these results at hand, we have a precise picture of where both perturbative and nonperturbative data from the sectors of the anomalous dimension 
ends up: because $\og_{n} \in \Ho(\Tu_{\A})$, we find that the coefficients of the $\sigma$-th (instanton) sector of $\og_{n}$ can only have data from 
sectors $\leq \sigma$ of $\og$, in signs 
\begin{equation}
 [\mo_{2}^{s}]\W_{\sigma} \og_{n} \in \A_{\sigma}  \hs{2} \forall n \geq 1, s \in \Z,
\end{equation}
where $\A_{\sigma}$ has no coefficients of degree $> \sigma$. To see sector homogeneity concretely, we compute a few nonperturbative coefficients of the 
second RG function: 
\begin{equation}
 (\Rn \og)|_{l} = (\og \M\og)|_{l} = \sum_{l'+l''=l} c_{l'} ( \M\og)|_{l''} 
 = \sum_{l'_{1}+l''_{1}=l_{1}} \ \sum_{l'_{2}+l''_{2}=l_{2}} c_{(l_{1}',l_{2}')} ( \M\og)|_{(l_{1}'',l_{2}'')},
\end{equation}
where by (\ref{deri}) we have $( \M \og)|_{l} = ( s [l_{1}c +l_{2}]-1 )c_{l} + s l_{1} \sum_{i=1}^{m} i b_{i} c_{(l_{1},l_{2}+i)}$ with shorthand 
\begin{equation}\label{Abk}
 \M:= s \der -1
\end{equation}
which we will employ from now on. Notice that $ ( \M\og)|_{(l_{1},l_{2})}=0$ for $l_{2} < -m$. For the first instanton sector of $\wg_{2}$, the RG recursion yields 
\begin{equation}
 (\Rn \og)|_{(1,t)} = \sum_{t'+t''=t} \left[ c_{(1,t')} ( \M \og)|_{(0,t'')} + c_{(0,t')} ( \M\og)|_{(1,t'')} \right].
\end{equation}
Note that $(\Rn \og)|_{(1,t)} = 0$ if $t \leq -m$ on account of $c_{(0,0)}=0$. For $m=1$, the first nonvanishing coefficients are given by 
\begin{equation}\label{1stsec}
 \begin{split}
(\Rn \og)|_{(1,0)} &= c_{(1,0)} ( \M\og)|_{(0,0)} + c_{(0,1)} ( \M\og)|_{(1,-1)} = c_{(0,1)} ( \M\og)|_{(1,-1)} =  s b_{1} c_{(0,1)} c_{(1,0)} \\
(\Rn \og)|_{(1,1)} &= c_{(1,0)} ( \M\og)|_{(0,1)} + c_{(0,1)} ( \M\og)|_{(1,0)} + c_{(0,2)} ( \M\og)|_{(1,-1)} \\
&= [ (sc+s-2) c_{(0,1)} +  s b_{1} c_{(0,2)} ] c_{(1,0)} + sb_{1} c_{(0,1)} c_{(1,1)} \\
(\Rn \og)|_{(1,2)} &= c_{(1,1)} ( \M\og)|_{(0,1)} + c_{(1,0)} ( \M\og)|_{(0,2)} + c_{(0,1)} ( \M\og)|_{(1,1)} + c_{(0,2)} ( \M\og)|_{(1,0)} \\
& \hs{3} + c_{(0,3)} ( \M\og)|_{(1,-1)} \\ 
&= [(sc+2s-2)c_{(0,2)} + sb_{1} c_{(0,3)} ] c_{(1,0)} + [(sc+2s-2)c_{(0,1)} + s b_{1} c_{(0,2)} ] c_{(1,1)}  \\ 
& \hs{3}  + sb_{1} c_{(0,1)}c_{(1,2)},  
 \end{split}
\end{equation}
in which we can clearly see that all are elements in $\A_{1}$. The first coefficients of the perturbative sector read
\begin{equation}\label{pertsec}
 \begin{split}
(\Rn \og)|_{(0,0)} &= (\Rn \og)|_{(0,1)} = 0  , \hs{1} (\Rn \og)|_{(0,2)} = c_{(0,1)} ( \M\og)|_{(0,1)} = (s-1) c_{(0,1)}^{2}  \\
(\Rn \og)|_{(0,3)} &=  c_{(0,1)} ( \M\og)|_{(0,2)} + c_{(0,2)} ( \M\og)|_{(0,1)}  = (3s-2)c_{(0,1)} c_{(0,2)},   
 \end{split}
\end{equation}
and all of them lie in $\A_{0}$. If we had introduced an extra grading with respect to the second index (to end up having a double grading), the RG operator 
would be homogeneity preserving also with respect to this grading, but as (\ref{1stsec}) shows, only on the perturbative sector!  
Some readers may at this stage still wonder what this is about. We ask for patience, towards the end of the next section (§\ref{sec:PinNP}), 
this matter will then finally become clear. 

\subsection{RG-driven flow of data}
To summarise the state of play so far, according to the RG recursion, and hence the RG equation, we note for the record that
\begin{itemize}
 \item both perturbative and nonperturbative sectors of the anomalous dimension $\og$ inform all sectors of the higher RG functions $\og_{n}$
       in such a way that lower sectors of $\og$ pass on data to all higher ones of $\og_{n}$ and, in particular,
 \item the perturbative sectors of the RG functions' transseries $\og_{n}$ do not receive information from $\og$'s nonperturbative sectors,
       while conversely all nonperturbative sectors of $\og_{n}$ are informed by its perturbative sector.
\end{itemize}
The flow of information as driven by the RG recursion is depicted schematically in \textsc{Figure} \ref{RGFlow}, where the non-shaded boxes on the 
bottom represent the sources of data. 
\begin{figure}[ht]
\begin{center} \includegraphics[height=2.5cm]{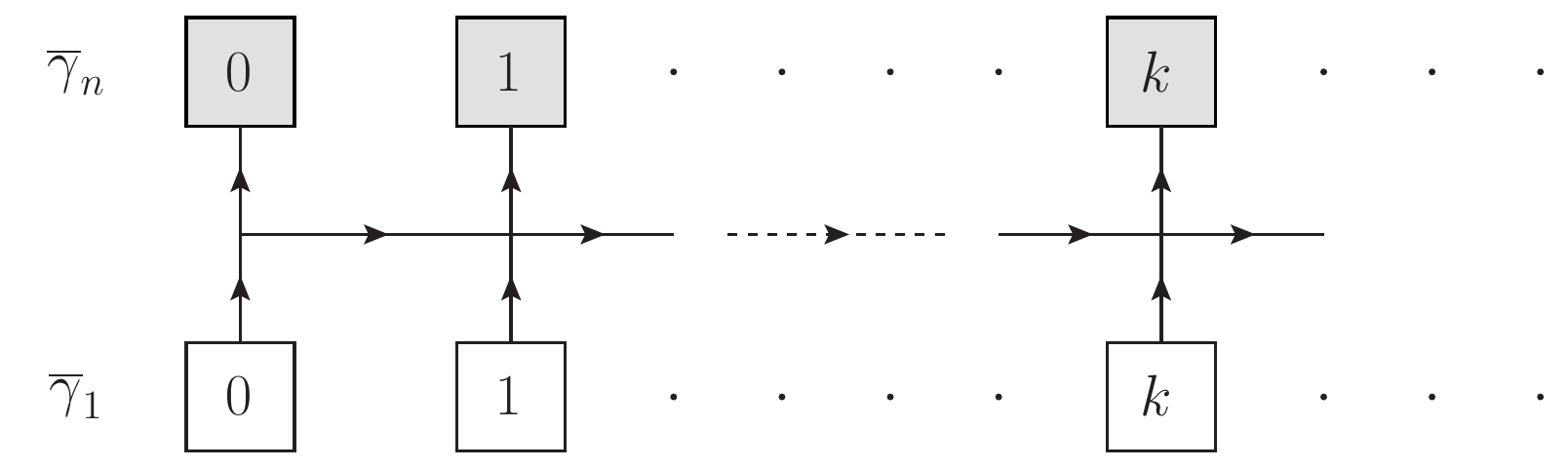} \end{center} 
\caption{\small How the RG recursion passes on perturbative and nonperturbative data
from the anomalous dimension $\og=\og_{1}$ to the instanton sectors $\sigma = 0,1,2,\dotsc,k$ of the $n$-th RG function $\og_{n}$ ($n\geq 2$).
The non-shaded boxes represent the (instanton) sectors of the anomalous dimension $\og$ and are the sources of information. 
The shaded boxes, which stand for the sectors of the $n$-th RG function $\og_{n}$, can only receive data.}
\label{RGFlow}
\end{figure}

\section{Communication between sectors}\label{sec:PinNP}               		          
We shall in this section study how the DSEs prescribe the sectors of the anomalous dimension's transseries $\wg$ to communicate. Resurgence in the 
strict sense that the perturbative sector determines all nonperturbative sectors should in principle only be guaranteed if the transseries ansatz is the 
correct one. However, as we show next, the key property of the Dyson-Schwinger (DS) operator it takes for a one-way transfer of data turns out to 
be \emph{preservation of sector homogeneity}. Although we know that our transseries ansatz is not the right one, we expect the essence of the results 
in this section to be true even for more elaborate transseries, as long as one can identify sectors defined by exponentials of transseries. 

\subsection{DS operator preserves sector homogeneity}
Strictly speaking, the DSEs 
\begin{equation}\label{Deisson}
\og = \Psi_{M}(\og), \hs{2}  \og = \Psi^{\graph{0}{0.04}{0.08}{pho}{0}}_{M_{1},\dotsc,M_{N}}(\og)
\end{equation}
make no sense in $\Tu_{\A}$ for the simple reason that in the algebra $\A$ there are no relations between the coefficients by the very definition
of a free algebra. However, it is clear how we should interprete (\ref{Deisson}): the abstract coefficients on both sides can be compared 
transmonomial-wise so as to yield algebraic equations in which the coefficients play the role of the unknowns. One should then at least in principle 
be able to solve the equation order by order and sector by sector. Before we proceed, we consider a little 

\begin{Lemma}\label{QH}
The operators $\Qa_{M_{1},\dotsc,M_{\ell}} \colon \Tu_{\A} \rightarrow \Tu_{\A}$ preserve sector homogeneity, ie 
\begin{equation}
 \Qa_{M_{1},\dotsc,M_{\ell}} \Ho(\Tu_{\A}) \subseteq \Ho(\Tu_{\A})
\end{equation}
for all $M_{1},\dotsc,M_{\ell} \in \N$ and $1 \leq \ell \leq N$. 
\end{Lemma}
\begin{proof}
 All operations performed by the operators $\Qa_{M_{1},\dotsc,M_{\ell}}$ are those of the differential algebra $(\Tu_{\A},\der)$, yielding linear 
 combinations of products of RG functions. 
 By Proposition \ref{flow}, all RG transseries are sector homogeneous. Because the set of such transseries $\Ho(\Tu_{\A})$ is a differential 
 subalgebra with derivation $\der$, we conclude that $\Qa_{M_{1},\dotsc,M_{\ell}}(\og) \in \Ho(\Tu_{\A})$ due to $\og \in \Ho(\Tu_{\A})$. 
\end{proof}

We remind the reader that $\W_{\sigma}f$ denotes the power series of the $\sigma$-th (non)perturbative sector of the transseries $f$ and that 
$[\mo_{2}^{s}]\W_{\sigma}f = f_{(\sigma,s)}$ is the coefficient associated with the double index $(\sigma, s) \in \N \times \Z$.   

\begin{Proposition}[Resurgence]\label{main}
For all $M_{1},\dotsc,M_{N} \geq 1$, both sides of the equation
\begin{equation}\label{dseq}
\og = \underbrace{A_{0}\mo_{2} + \sum_{\ell=1}^{N} \mo_{2}^{\ell+1} \Qa_{M_{1},\dotsc,M_{\ell}}(\og)}_{
= \Psi^{\graph{0}{0.04}{0.08}{pho}{0}}_{M_{1},\dotsc,M_{N}}(\og)}
\in \Ho(\Tu_{\A})  
\end{equation}
are sector-homogeneous transseries. Consequently, for each fixed sector $\sigma \geq 0$ and each $n \geq 1$, the equations
\begin{equation}\label{sy}
[\mo_{2}^{s}]\W_{\sigma}[\og] 
= [\mo_{2}^{s}]\W_{\sigma}\left[A_{0}\mo_{2} +  \Sum_{\ell=1}^{N} \mo_{2}^{\ell+1} \Qa_{M_{1},\dotsc,M_{\ell}}(\og) \right]
\in \A_{\sigma}   \hs{2} s \leq n
\end{equation}
describe a finite set of $n$ nonlinear difference equations for the coefficients of the anomalous dimension $\og$, where only coefficients of homogeneous
degree $\leq \sigma$ are involved. The analogous statement is true for the Kilroy DSE $\og = \mo_{2} \Qa_{M}(\og)$.   
\end{Proposition}
\begin{proof}
 The first equation for $\sigma=0,s=1$ is given by $c_{(0,1)}=A_{0}$ which determines first perturbative coefficient, where $A_{0} \in \R$ can be 
 identified with $A_{0}1 \in \A_{0}$. Because the rhs of (\ref{dseq}) is sector-homogeneous by Lemma \ref{QH}, one finds that for each fixed 
 sector $\sigma \geq 1$, the system (\ref{sy}) is made up only of coefficients from sectors $\leq \sigma$. The assertion also holds for the 
 Kilroy DSE because the arguments apply equally to $\og = \mo_{2} \Qa_{M}(\og)$.
\end{proof}

This means the perturbative sector is independent of all other sectors, ie for $\sigma=0$, the system of equations
\begin{equation}
[\mo_{2}^{s}]\W_{0}[\og] = [\mo_{2}^{s}]\W_{0}\left[A_{0}\mo_{2} +  \Sum_{\ell=1}^{N} \mo_{2}^{\ell+1} \Qa_{M_{1},\dotsc, M_{\ell}}(\og) \right] \in \A_{0}
 \hs{2} s \leq n,  
\end{equation}
determines the perturbative sector up to perturbative (loop) order $n$ and the infinite tower  
\begin{equation}\label{tow}
\begin{split}
 [\mo_{2}^{s}]\W_{1}[\og] &= [\mo_{2}^{s}]\W_{1}\left[A_{0}\mo_{2} +  \Sum_{\ell=1}^{N} \mo_{2}^{\ell+1} \Qa_{M_{1},\dotsc, M_{\ell}}(\og) \right]
\in \A_{1}  \\
 [\mo_{2}^{s}]\W_{2}[\og] &= [\mo_{2}^{s}]\W_{2}\left[A_{0}\mo_{2} +  \Sum_{\ell=1}^{N} \mo_{2}^{\ell+1} \Qa_{M_{1},\dotsc, M_{\ell}}(\og) \right]
\in \A_{2} \\ \hs{-2} \vdots & \hs{2} \vdots 
\end{split}
\end{equation}
contains (in principle) the corresponding equations for all nonperturbative sectors. To find the transseries solution of the proper anomalous dimension $\wg$, the limit 
$M_{1}, \dotsc, M_{N} \rightarrow \infty$ must be taken.  

\textsc{Figure} \ref{DSEFlow} shows schematically what this \emph{essentially} means: the perturbative sector $\W_{0}[\og]$ of the anomalous dimension $\og$ 
is the sole source of information and passes on purely perturbative data to all of its own instanton sectors $\W_{\sigma}[\og]$ for $\sigma \geq 1$ and 
thereby all instanton sectors of the higher RG function's transseries. 

\begin{figure}[ht]
\begin{center} \includegraphics[height=2.5cm]{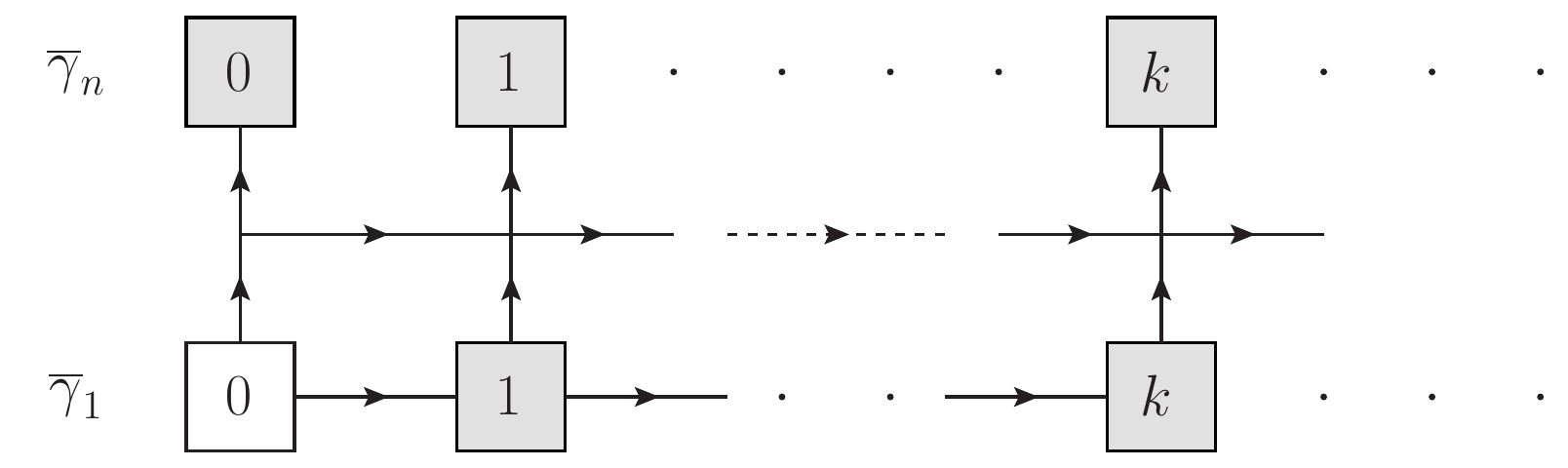} \end{center} 
\caption{\small How the DSE prescribes the perturbative sector of the anomalous dimension $\og_{1}$  
to pass on data to its own nonperturbative sectors and thereby to those of the $n$-th RG function $\og_{n}$ ($n\geq2$). The perturbative sector of 
the anomalous dimension, represented by the single non-shaded box in the lower left corner, is the only source of information.}
\label{DSEFlow}
\end{figure}

To see this concretely, we revisit the Kilroy DSE (\ref{Kil'}) as an equation in $\Tu_{\A}$, 
\begin{equation}\label{Kil''}
 \og = \mo_{2} \lim_{M \rightarrow \infty}\Qa_{M}(\og) = \mo_{2} \left[ C_{0} + C_{1} \og 
 + C_{2} \left( \frac{\og_{2}}{2!} + \og \cdot \og \right) + \dotsc  \right]  .
\end{equation}
Recall that this is a fixed point equation for the anomalous dimension $\og$ because all RG functions $\og_{n}$ depend on $\og$ through the RG 
recursion $\og_{n+1} = \Rn^{n}(\og)$ which on the level of coefficients takes the form 
\begin{equation}\label{RGcoef}
\og_{n+1}|_{(\sigma,t)} = (\Rn \og_{n})|_{(\sigma,t)} 
= \sum_{\sigma'+\sigma''=\sigma} \ \sum_{t'+t''=t}  c_{(\sigma',t')} ( \M \og_{n})|_{(\sigma'',t'')} \in \A_{\sigma},
\end{equation}
where  
\begin{equation}
 ( \M \og_{n})|_{(\sigma'',t'')} = [s(\sigma''c+t'')-1]\og_{n}|_{(\sigma'',t'')} 
 + s \sigma'' \sum_{i=1}^{m} i b_{i} \og_{n}|_{(\sigma'',t''+i)}  
\end{equation}
and the coefficients $c_{(\sigma,t)}$ are those of the anomalous dimension 
\begin{equation}
 \og = \sum_{(\sigma,t) \geq (0,0)} c_{(\sigma,t)} \mo_{1}^{\sigma}\mo_{2}^{t}
\end{equation}
as a transseries in $\Tu_{\A}$. In terms of coefficients, the DSE (\ref{Kil''}) reads
\begin{equation}\label{Kil3'}
 c_{(\sigma,t)} = C_{0} \delta_{\sigma,0} \delta_{t,1} + C_{1} c_{(\sigma, t-1)} 
 + C_{2} \left( \frac{1}{2!} \og_{2}|_{(\sigma, t-1)} + (\og \cdot \og)|_{(\sigma, t-1)} \right) + \dotsc ,
\end{equation}
where the multiplication by $\mo_{2}$ in the DSE manifests itself as an index shift in the second component of the index pair, which is due to
\begin{equation}
 \mo_{2}^{s} f = \mo_{2}^{s} \sum_{(\sigma,t) \geq (\sigma_{0},t_{0})} f_{(\sigma,t)} \mo_{1}^{\sigma} \mo_{2}^{t+s}
 = \sum_{(\sigma,t) \geq (\sigma_{0},t_{0}+s)} f_{(\sigma,t-s)} \mo_{1}^{\sigma} \mo_{2}^{t}
\end{equation}
for a transseries $f =\sum_{(\sigma,t) \geq (\sigma_{0},t_{0})} f_{(\sigma,t)} \mo_{1}^{\sigma} \mo_{2}^{t}$ with 
$\supp(f) \subseteq \J_{(\sigma_{0},t_{0})}$ and entails $(\mo_{2}^{s} f)|_{(\sigma,t)} =0$ if $t < t_{0} +s$.

In combination with the RG recursion (\ref{RGcoef}), (\ref{Kil3'}) lets us clearly see that only coefficients from sectors $\leq \sigma$ are 
in a lucky case involved in determining the coefficient $c_{(\sigma,t)}$, where a 'lucky case' is given if a correct (or less wrong) transseries 
ansatz is being used.

\section{Getting the ansatz right}\label{sec:ansatz} 				          
The attentive reader knows by now that this section's title is deceptive: we prove in this section that our transseries algebra $\Tu$ contains no 
nonperturbative transseries solution for the Yukawa model and probably also none for the photon DSE. 

The parameters we are allowed to play with in $\Tu$ are the supportive grid for $\og$ as well as the instanton
polynomial $P(z)$ in the exponential of the transmonomial $\mo_{1}$.
We shall also explain why adding logarithmic monomials makes no difference. 

\subsection{First sector} 
We start with the first coefficient of the first nonperturbative sector ($\sigma = 1$) and evaluate the Kilroy DSE (\ref{Kil3'}) at $(\sigma,t)=(1,0)$ which gives
\begin{equation}\label{Kil3b}
 c_{(1,0)} = C_{1} c_{(1,-1)} + C_{2} \left( \frac{1}{2!}\og_{2}|_{(1,-1)} + (\og \cdot \og)|_{(1,-1)} \right) + \dotsc 
\end{equation}
The compute the rhs of this equation, we consider
\begin{equation}\label{og2}
\og_{2}|_{(1,t)}= (\Rn \og)|_{(1,t)} 
 = \sum_{t'+t''=t} \left[ c_{(1,t')} ( \M\og)|_{(0,t'')} + c_{(0,t')} ( \M \og)|_{(1,t'')} \right],
\end{equation}
with 
\begin{equation}\label{og2'}
\begin{split}
 ( \M \og)|_{(1,t'')} &= [ s (c +t'')-1 ] c_{(1,t'')} + s \left[ b_{1} c_{(1,t''+1)} + \dotsc + m b_{m} c_{(1,t''+ m)} \right] \\
 (\M \og)|_{(0,t'')}  &= ( s t''-1 )c_{(0,t'')} .
\end{split}
\end{equation}
The rhs of (\ref{Kil3b}) requires us to evaluate (\ref{og2}) at $t=-1$ which has a vanishing first term under the sum sign,
\begin{equation}
 \sum_{t'+t''=-1} c_{(1,t')} ( \M\og)|_{(0,t'')} = 0
\end{equation}
because either $t'$ of $t''$ must be negative. The remainder therefore consists of the sum  
\begin{equation}\label{og2b}
\og_{2}|_{(1,-1)} 
=  c_{(0,0)} ( \M \og)|_{(1,-1)} + c_{(0,1)} ( \M\og)|_{(1,-2)}  + \dotsc + c_{(0,m-1)} ( \M\og)|_{(1,-m)} 
\end{equation}
which vanishes if $m=1$ due to $c_{(0,0)}=0$ but is nonzero apriori if we choose $m \geq 2$. 
The next lemma proves that if we dismiss this latter choice and use $m=1$, the entire first sector vanishes.

\begin{Lemma}[First sector trivial]\label{1sttri}
 If the instanton polynomial of the Kilroy model is given by $P(z)=b_{1}z$, and hence $\deg P(z)=m=1$, then the entire first nonperturbative sector 
 ($\sigma=1$) of the anomalous dimension $\og$ vanishes, that is, $c_{(1,t)}=0$ for all $t \geq 0$. The analogous assertion is true for the photon's 
 anomalous dimension if $\deg P(z) = m \leq 2$. 
\end{Lemma}
\begin{proof}
We draw on (\ref{su}) in Proposition \ref{Psupp} and find for $\supp (\wg) \subseteq \Pe_{1} \cup \J_{(1,q)}$ and $m=1$,   
\begin{equation}\label{su'}
\supp (\og_{n}) \subseteq \Pe_{n} \cup \bigcup_{j=1}^{n} \J_{(j,j[q-1]+1)} \ , \hs{1}  n \geq 1,
\end{equation}
which says that $\og_{n}|_{(1,-1)}=0$ for all $n \geq 1, q \geq 0$ because (\ref{su'}) dictates $\mo_{1} \mo_{2}^{-1} \notin \supp (\og_{n})$ for all 
$n \geq 1$ (see also \textsc{Figures} \ref{Rsupport} \& \ref{rsupport}). 
Because the rhs of (\ref{Kil3'}) has only products of RG functions, it follows $c_{(1,0)}=0$ due to
\begin{equation}\label{ends}
 (\og_{n_{1}} \cdot \dotsc \cdot \og_{n_{\ell}})|_{(1,-1)} = 0
\end{equation}
for all $n_{1}, \dotsc , n_{\ell} \geq 1$. So in summary, for $t=0$ we have found $c_{(1,t)}=0$ due to $\og_{n}|_{(1,t-1)}=0$ for all $n \geq 1$. 
Note that on the level of the coefficients, the RG recursion takes the form
\begin{equation}\label{og}
\og_{n+1}|_{(1,t)}= (\Rn \og_{n})|_{(1,t)} 
 = \sum_{t'+t''=t} \left[ c_{(1,t')} ( \M\og_{n})|_{(0,t'')} + c_{(0,t')} ( \M\og_{n})|_{(1,t'')} \right],
\end{equation}
with 
\begin{equation}\label{og'}
\begin{split}
( \M \og_{n})|_{(1,t'')} &= [ s (c +t'')-1 ]\og_{n}|_{(1,t'')} + s b_{1} \og_{n}|_{(1,t''+1)}  \\
 (\M \og_{n})|_{(0,t'')} &= ( s t''-1 )\og_{n}|_{(0,t'')} .
\end{split}
\end{equation}
So assume now that $c_{(1,t)}=0$ and $\og_{n}|_{(1,t-1)}=0$ for $t \leq t_{*}$ and all $n \geq 1$ (we have so far shown this for $t_{*}=0$). Then 
(\ref{og}) and (\ref{og'}) imply   
\begin{equation}\label{og2''}
\og_{2}|_{(1,t)}=  \sum_{t'+t''=t} \left[ c_{(1,t')} ( \M\og)|_{(0,t'')} + c_{(0,t')} ( \M\og)|_{(1,t'')} \right] = 0  \hs{1} \text{for all }t \leq t_{*},
\end{equation}
because this expression makes only use of coefficients $c_{(1,t)}$ with $t\leq t_{*}$. Moreover, by induction on $n$, these 
equations imply $\og_{n}|_{(1,t)}=0$ for $t \leq t_{*}$ and all $n \geq 1$ by the same argument. This also entails  
\begin{equation}\label{ends'}
 (\og_{n_{1}} \cdot \dotsc \cdot \og_{n_{\ell}})|_{(1,t)} = 0 \hs{2} \mbox{if} \hs{0.4} t \leq t_{*}
\end{equation}
for all $n_{1}, \dotsc , n_{\ell} \geq 1$. Then the DSE in (\ref{Kil3'}) dictates $c_{(1,t+1)}=0$ for all $t \leq t_{*}$ and thus   
$c_{(1,t)}=0$ for all $t \leq t_{*}+1$ as a consequence.
The same thing happens in QED for the photon's anomalous dimension: the photon DSE (\ref{dseq})
\begin{equation}\label{dyson1'}
\og = A_{0}\mo_{2} + \lim_{M_{1},\dotsc,M_{N} \rightarrow \infty} \sum_{\ell=1}^{N} \mo_{2}^{\ell+1} \Qa_{M_{1},\dotsc,M_{\ell}}(\og) 
\hs{1} (\mbox{photon DSE in $\Tu_{\A}$})
\end{equation}
reads 
\begin{equation}\label{dyson1''}
c_{(1,t)} = \lim_{M_{1},\dotsc, M_{N} \rightarrow \infty} 
\left[ (\Qa_{M_{1}}\og)|_{(1,t-2)} + \dotsc + (\Qa_{M_{1},\dotsc,M_{N}} \og)|_{(1,t-N-1)} \right] \hs{0.5} (\mbox{photon}) 
\end{equation}
for the coefficients of the first sector, where we see again the index shift brought about by multiplication by $\mo_{2}^{\ell+1}$ for 
$\ell=1, \dotsc , N$. Because the RG recursion differs from the Yukawa case only in the parameter $s$ hidden in $\M = s \der -1$, 
one obtains the same result for QED by the analogue of the above argument with a minor modification to the first line of (\ref{og'}) which gives
\begin{equation}
 ( \M \og_{n})|_{(1,t'')} = [ s (c +t'')-1 ]\og_{n}|_{(1,t'')} + s \left[ b_{1} \og_{n}|_{(1,t''+1)} + 2 b_{2} \og_{n}|_{(1,t''+ 2)} \right].
\end{equation}
The last term allows us to let $\deg P(z) \leq 2$ and still find the same negative result. 
\end{proof}

\subsection{Higher sectors}\label{subsec:neg}
However, in both cases a vanishing first sector entails much more, namely that \emph{all} nonperturbative sectors are absent 
if the degree of the instanton polynomial is not large enough. 

\begin{Proposition}[All sectors trivial]\label{trivi}
All nonperturbative sectors of the anomalous dimension's transseries vanish if the degree of the instanton polynomial $P(z)$ is too small, 
that is, if
\begin{enumerate}
 \item [(i)]  $\deg P(z) = 1$ in the case of the Kilroy DSE and 
 \item [(ii)] $\deg P(z) \leq 2$ in the case of the photon DSE.
\end{enumerate}
\end{Proposition}
\begin{proof}
We start with the Kilroy case (i) and let $\deg P(z)=m=1$. On account of Lemma \ref{1sttri}, the first nonperturbative sector vanishes completely. 
We now perform the induction step with respect to $\sigma \geq 1$ to show that 
\begin{equation}\label{Kil5}
 c_{(\sigma,0)} = C_{1} \underbrace{c_{(\sigma,-1)}}_{=0} + C_{2} \left[ \frac{1}{2!}\og_{2}|_{(\sigma,-1)} 
 + (\og \cdot \og)|_{(\sigma,-1)} \right] + \dotsc 
\end{equation}
vanishes and then $c_{(\sigma,t)}=0$ as a consequence for all $\sigma \geq 1, t \geq 0$. Consider
\begin{equation}\label{ogs}
\og_{n+1}|_{(\sigma,t)}= (\Rn \og_{n})|_{(\sigma,t)} = \sum_{\sigma' + \sigma'' = \sigma} \  
\sum_{t'+t''=t} c_{(\sigma',t')} ( \M\og_{n})|_{(\sigma'',t'')} .
\end{equation}
For the induction, we assume $\og_{n}|_{(\tau,t)} =0$ for all $n \geq 1$ and $\tau$ with $1 \leq \tau < \sigma$ and all $t$ (it is true for $\sigma=2$ by 
Lemma \ref{1sttri}). Consequently, the sum in (\ref{ogs}) shrinks giving
\begin{equation}\label{rek'}
\og_{n+1}|_{(\sigma,t)} 
=  \sum_{t'+t''=t} \left[ c_{(\sigma,t')} ( \M\og_{n})|_{(0,t'')} + c_{(0,t')} ( \M\og_{n})|_{(\sigma,t'')}  \right].
\end{equation}
We first consider this expression for $n=1$ to compute $\og_{2}|_{(\sigma,-1)}$ on the rhs of (\ref{Kil5}). To this end, first note that, on the face of it,
only the second term under the summation sign in (\ref{rek'}) survives due to $c_{(\sigma,t)}=0=(\M\og)|_{(0,t)}$ if $t <0$, the former by definition and the latter 
by $(\M\og)|_{(0,t)}$ being purely perturbative, see the second line of (\ref{og'}).
Therefore, we find
\begin{equation}\label{rek''}
\og_{2}|_{(\sigma,-1)} 
= c_{(0,1)} ( \M\og)|_{(\sigma,-2)} + c_{(0,2)} ( \M\og)|_{(\sigma,-3)} + c_{(0,3)} ( \M\og)|_{(\sigma,-4)} +  \dotsc
\end{equation}
However, these terms also vanish due to 
\begin{equation}
 ( \M\og)|_{(\sigma,t)}= ( s [\sigma c +t]-1 )c_{(\sigma,t)} + s \sigma b_{1} c_{(\sigma,t+1)} = 0 \hs{2} \mbox{if} \hs{0.5} t \leq -2
\end{equation}
which in turn originates in the fact that $c_{(\sigma,t)}=0$ for all $t <0$. 
Thus we find $\og_{2}|_{(\sigma,-1)}=0$ and likewise $\og_{n+1}|_{(\sigma,-1)}=0$ which follows by induction on $n$ from 
\begin{equation}\label{rek2}
\og_{n+1}|_{(\sigma,-1)} 
= c_{(0,1)} ( \M\og_{n})|_{(\sigma,-2)} + c_{(0,2)} ( \M\og_{n})|_{(\sigma,-3)}  + c_{(0,3)} ( \M\og_{n})|_{(\sigma,-4)} + \dotsc 
\end{equation}
Because the coefficients of products of RG functions with double index $(\sigma,-1)$ vanish, 
\begin{equation}
 (\og_{n_{1}} \cdot \dotsc \cdot \og_{n_{\ell}})|_{(\sigma,-1)} = 0 
\end{equation}
we have $c_{(\sigma,0)}=0$ in (\ref{Kil5}). The remainder of the argument is analogous to the proof of Lemma \ref{1sttri}: 
assume $c_{(\sigma,t)}=0$ for all $t \leq t_{*}$ (we have shown this for $t_{*}=0$). Note that 
\begin{equation}
 ( \M\og)|_{(\sigma,t)}= ( s [\sigma c + t]-1 )c_{(\sigma,t)} + s \sigma b_{1} c_{(\sigma,t+1)} = 0 \hs{2} \mbox{if} \hs{0.5} t \leq t_{*}-1
\end{equation}
implies $\og_{2}|_{(\sigma,t)}=0$ due to 
\begin{equation}
\og_{2}|_{(\sigma,t)} 
= c_{(0,1)} ( \M\og)|_{(\sigma,t-2)} + c_{(0,2)} ( \M\og)|_{(\sigma,t-3)} + c_{(0,3)} ( \M\og)|_{(\sigma,t-4)} +  \dotsc  = 0
\end{equation}
if $t \leq t_{*}$ and then $\og_{n+1}|_{(\sigma,t)}=0$ inductively on $n$ for all $t \leq t_{*}$. The DSE 
\begin{equation}\label{Kil6}
 c_{(\sigma,t_{*}+1)} = C_{1} c_{(\sigma,t_{*})} + C_{2} 
 \left[ \frac{1}{2!}\og_{2}|_{(\sigma,t_{*})} + (\og \cdot \og)|_{(\sigma,t_{*})} \right] + \dotsc 
\end{equation}
consequently tells us $c_{(\sigma,t_{*}+1)}=0$. This finishes the proof for the Kilroy case. 

The case (ii) of the photon DSE is analogous and starts with the fact that according to (\ref{dyson1'}), one needs the coefficients 
$\og_{n}|_{(\sigma,-j)}$ for $j=2, \dotsc , N+1$ to compute $c_{(\sigma,0)}$. If $m \leq 2$, then these coefficients can be shown to vanish: first one 
finds that  
\begin{equation}\label{ogsi}
 ( \M \og_{n})|_{(\sigma,t'')} = [ s (c \sigma +t'')-1 ]\og_{n}|_{(\sigma,t'')} 
 + s \sigma \left[ b_{1} \og_{n}|_{(\sigma,t''+1)} + 2 b_{2} \og_{n}|_{(\sigma,t''+2)} \right]
\end{equation}
is zero for $n=1$ if $t'' \leq -3$. This is relevant because (\ref{rek'}) asks for (\ref{ogsi}) to be evaluated at such index values to 
compute $\og_{2}|_{(\sigma,-j)}$ for $j=2, \dotsc , N+1$. By induction on $n$ then follows 
$\og_{n}|_{(\sigma,-j)}=0$ for all $n$ and thus $c_{(\sigma,0)}=0$ by (\ref{dyson1''}). By induction on $t$ in $c_{(\sigma,t)}$, one finally ends up concluding
$c_{(\sigma,t)}=0$ for all $t$.
\end{proof}

Within the narrow context of our transseries ansatz, this negative result is related to the 'skeleton operation' of multiplying the operators 
$\Qa_{M_{1}, \dotsc ,M_{n}}$ by powers of $\mo_{2} = z^{-1}$, by itself a contractive operation. It leads to a shift in the second index of the 
coefficients that can only be compensated for by increasing the degree of the instanton polynomial $P(z)$. 

\subsection{Kilroy ODE and transseries ansatz} 
The Yukawa model is considerably more amenable because the ODE (\ref{Kilode}) is much better suited to put various transseries ans\"atze to the test 
and vary their parameters. Given the above negative results, a pressing question is whether increasing the degree of the instanton polynomial might help,
as suggested by Proposition \ref{trivi}. In the transseries algebra $\Tu_{\A}$, we write this equation as 
\begin{equation}\label{kilroytrans}
 \og(z) + \Rn \og(z)  = \frac{1}{2}z^{-1} .
\end{equation}
It has a simple message, namely that the lhs has a vanishing nonperturbative part because the rhs has none. 

For the first instanton sector, the coefficients must satisfy 
\begin{equation}\label{1stKil}
(\og + \Rn \og)|_{(1,t)}  = c_{(1,t)} + \sum_{t'+t''=t} \left[ c_{(1,t')} ( \M\og )|_{(0,t'')} + c_{(0,t')} ( \M\og)|_{(1,t'')} \right] = 0,
\end{equation}
which we will now explore. Before we start, we remind ourselves of (\ref{og2'}) with $s=2$
\begin{equation}
\begin{split}
 ( \M \og)|_{(1,t'')} &= [ 2 (c +t'')-1 ] c_{(1,t'')} + 2 \left[ b_{1} c_{(1,t''+1)} + \dotsc + m b_{m} c_{(1,t''+ m)} \right] \\
 (\M \og)|_{(0,t'')}  &= ( 2 t''-1 )c_{(0,t'')} .
\end{split}
\end{equation}
For negative $t$, (\ref{1stKil}) shrinks to 
\begin{equation}\label{1stKil'}
0=(\og + \Rn \og)|_{(1,t)}  = \sum_{t'+t''=t} c_{(0,t')} ( \M\og )|_{(1,t'')}  \hs{2} (t<0),
\end{equation}
on account of $( \M \og)|_{(0,t'')}=0$ and $c_{(1,t')}=0$ for $t',t''<0$. We assume now that $\deg P(z) = m \geq 2$. From this we get a system of equations,
\begin{equation}
\begin{split}
 0 &= (\og + \Rn \og)|_{(1,-m+1)} =  c_{(0,1)} ( \M\og )|_{(1,-m)} \\
 0 &= (\og + \Rn \og)|_{(1,-m+2)} =  c_{(0,1)} ( \M\og )|_{(1,-m+1)}  + c_{(0,2)} ( \M\og )|_{(1,-m)}  \\
  & \vdots \hs{4}  \vdots  \\
 0 &= (\og + \Rn \og)|_{(1,-1)} =  c_{(0,1)} ( \M\og )|_{(1,-2)} + \dotsc + c_{(0,m-1)} ( \M\og )|_{(1,-m)} .
\end{split}
\end{equation}
Because none of the perturbative coefficients vanish, this implies $( \M\og )|_{(1,t)}= 0$ for all $t \leq -2$. This in turn implies 
\begin{equation}
\begin{split}
 0 = (\M \og)|_{(1,-m)} &=  2 m b_{m}c_{(1,0)} \\
 0 = (\M \og)|_{(1,-m+1)} &=  2 [ (m-1) b_{m-1} c_{(1,0)}  + m b_{m} c_{(1,1)} ] \\
  & \vdots \hs{4}  \vdots  \\
 0 = (\M \og)|_{(1,-2)} &=  2 [ 2 b_{2} c_{(1,0)} + \dotsc  + m b_{m} c_{(1,-2+m)} ]
\end{split}
\end{equation}
which enforces $c_{(1,t)}=0$ for $t=0, \dotsc , m-2$ because by assumption $b_{m} \neq 0$. This is the start of an induction because we have shown
\begin{equation}
(\M \og)|_{(1,t)}=0  \hs{2} c_{(1,t+m)}=0
\end{equation}
for all $t \leq t_{*}$ if $t_{*}=-2$. Assume now that it holds for some $t_{*} \geq -2$. Consider
\begin{equation}
 (\M \og)|_{(1,t_{*}+1)} =  2 [ b_{1} c_{(1,t_{*}+2)} + 2 b_{2} c_{(1,t_{*}+3)} + \dotsc  + m b_{m} c_{(1,t_{*}+m+1)}] .
\end{equation}
Because of $c_{(1,t_{*}+m)}=0$ by assumption and $m \geq 2$, this expression collapses to 
\begin{equation}
 (\M \og)|_{(1,t_{*}+1)} =  2 m b_{m} c_{(1,t_{*}+m+1)} .
\end{equation}
and must itself vanish on account of
\begin{equation}
 0=(\og + \Rn \og)|_{(1,t_{*}+2)} = c_{(1,t_{*}+2)} + c_{(0,1)} (\M \og)_{(1,t_{*}+1)} = c_{(0,1)} (\M \og)_{(1,t_{*}+1)}
\end{equation}
where $c_{(1,t_{*}+2)}=0$ because of $m \geq 2$. But this entails $c_{(1,t_{*}+m+1)}=0$ which completes the induction on the second index $t$ and 
proves that the entire first sector vanishes. The induction step for the sectors $\sigma \geq 1$ leads to 

\begin{Proposition}[Kilroy]\label{Kilprop}
Let $\deg P(z) \geq 2$ and $\wg \in \Tu$ a solution of the Kilroy ODE (\ref{kilroytrans}). Then the entirety of all nonperturbative sectors vanishes,
ie $\W_{\sigma}\wg=0$ for all sectors $\sigma \geq 1$. 
\end{Proposition}
\begin{proof}
The preceding arguments have shown that $\W_{\sigma}\wg=0$ for $\sigma=1$. Assume $\W_{\sigma}\wg=0$ for all $\sigma \leq \sigma_{*}$. 
Then (\ref{kilroytrans}) implies  
\begin{equation}
0 = (\og + \Rn \og)|_{(\sigma_{*}+1,t)} = c_{(\sigma_{*}+1,t)} + \sum_{\sigma'+\sigma''= \sigma_{*}+1} \  \sum_{t'+t''= t} c_{(\sigma',t')} (\M \og)|_{(\sigma'',t'')}  
\end{equation}
which entails 
\begin{equation}\label{kili}
0 = c_{(\sigma_{*}+1,t)} +  \sum_{t'+t''= t}  \left[ c_{(\sigma_{*}+1,t')} (\M \og)|_{(0,t'')}  + c_{(0,t')} (\M \og)|_{(\sigma_{*}+1,t'')} \right]  .
\end{equation}
The remainder of the proof goes along the same lines as above, starting from (\ref{1stKil}) and going through all steps again. This is possible because
the above argument did not use the fact that $\sigma=1$ except when (\ref{1stKil}) was formulated. In fact, (\ref{kili}) sets off the same cascade of 
vanishing coefficients.
\end{proof}

The case $\deg P(z)=1$ is very distinct from the other cases because triviality of the transseries cannot be proven in the way we have done it here. 
Although the Kilroy ODE (\ref{kilroytrans}) is sector homogeneous, the coefficients of the first sector are not determined (to vanish) by the 
equations of this sector but by those of higher sectors, yet another reminder we did not get our transseries ansatz right in this work. 
However, we spare the reader these details and 'content' ourselves with the combination of the above no-go propositions \ref{Kilprop} and \ref{trivi} to 
conclude that the Kilroy DSE 
\begin{equation}
 \og = \mo_{2} \lim_{M \rightarrow \infty}\Qa_{M}(\og) =  \mo_{2} \left[ C_{0} + C_{1} \og 
 + C_{2} \left( \frac{\og_{2}}{2!} + \og \cdot \og \right) + \dotsc  \right] 
\end{equation}
has only one solution in $\Tu$, namely the perturbative one, whatever the degree of the instanton polynomial $P(z)$. For the QED case, we cannot be sure 
because increasing the degree of this polynomial might fix the problem. We suspect this not to be the case, but will not attempt to prove it here. 

\subsection{Logarithmic transmonomials}
Given these disappointing results, one might ask whether logarithmic transseries have the potential to make a difference. The answer is implied by what 
we found in §\ref{sec:DynSys}: adding $\mo_{3}=\log z$ as an extra transmonomial, we find that the first RG step involves 
\begin{equation}
(\M \og) |_{(l_{1},l_{2},l_{3})} = [s(cl_{1}+l_{2})-1]c_{(l_{1},l_{2},l_{3})} - s (l_{3}+1) c_{(l_{1},l_{2},l_{3}+1)} 
+ s l_{1} \sum_{i=1}^{m} i b_{i} c_{(l_{1},l_{2}+i,l_{3})}  
\end{equation}
which follows from (\ref{logs}). The devastating index shift we have discussed at the end of §\ref{subsec:neg} affects these coefficients in much the 
same way. Because the coefficients of $\og$ vanish for negative indices, the outcome is actually the same, whether there is an additional index 
attached to the coefficients or not. 

\subsection{ODE for QED} 
One might ask whether there is an analogue of the Kilroy ODE (\ref{kilroytrans}) in QED. The answer is that it is possible to formulate it, namely
\begin{equation}\label{QEDtrans}
 \gamma(\alpha) - \Rn \gamma(\alpha) = \gamma(\alpha) - \gamma(\alpha) (\alpha \partial_{\alpha} -1)\gamma(\alpha) = P(\alpha),
\end{equation}
where practically nothing is known about $P(\alpha)$, a function which is not to be confused with the instanton polynomial used in this paper.

We have studied this equation in some detail in \cite{KlaK13} and found some criteria concerning 
a Landau pole, which depends on the behaviour of $P(\alpha)$ (see also \cite{BaKUY09,BaKUY10} for a different approach). 
In particular, the toy situation $P(\alpha)=\alpha$, a direct analogue of the Kilroy ODE, is exactly solvable: the family 
\begin{equation}
 \gamma(\alpha) = \alpha + \alpha W \left(\xi e^{-1/\alpha}\right)
\end{equation}
indexed by $\xi \in \R$ has all possible solutions \cite{BaKUY09,Y11}, where $W(x)$ is the Lambert W function, defined by the transcendental equation 
$x = W(x) \exp W(x)$. We can use this function's convergent Taylor series 
\begin{equation}
 W(x) = \sum_{n \geq 1} \frac{(-n)^{n-1}}{n!} x^{n}
\end{equation}
and find the transseries expansion in $\alpha = z^{-1}$ given by 
\begin{equation}
 \wg(z) = z^{-1} + \sum_{n \geq 1} \frac{(-n)^{n-1}}{n!} \xi^{n} e^{-nz} z^{-1} = z^{-1} + \xi e^{-z} z^{-1} - \xi^{2} e^{-2z} z^{-1} + 
 \frac{3}{2} \xi^{3} e^{-3z} z^{-1} + \dotsc    
\end{equation}
which is very interesting since the perturbative part is totally trivial. Revisiting (\ref{pertsec}) shows why it happens if we use our transseries 
ansatz: set $s=1$ in (\ref{pertsec}) and use (\ref{QEDtrans}) with $P(\alpha)=\alpha$ to see it for the perturbative sector. For the nonperturbative sector, one 
finds $b_{1}=1$ and $\xi=c_{(1,0)}$ by means of (\ref{1stsec}). We can therefore not expect $P(\alpha)$ to be that simple and QED 
to have a direct analogue of the Kilroy ODE (\ref{kilroytrans}).

\section{Conclusion}\label{sec:Conclu}				                          
As we explained, there should be little doubt about the divergence of perturbation theory in our models. We 
interprete this as an indication for the possibility that their Dyson-Schwinger equations (DSEs) capture 
nontrivial physics, ie nonperturbative field configurations. Given the tenuous nonperturbative 
status of both theories, these may of course be only fictitious and not correspond to any real-world effects. 
But to describe such most likely high-energy states, the perturbative series need to be upgraded to 
resurgent transseries with additional transmonomials. 

The shortest possible summary of this paper is to say we have proved that for the anomalous dimension there is no transseries representation of the form 
\begin{equation}
 \wg(z) = \sum_{(\sigma,t) \geq (0,0)} c_{(\sigma,t)} z^{-\sigma c} e^{-\sigma P(z)} z^{-t} , \hs{0.5} \text{where } P(z)=b_{1}z + b_{2}z^{2} 
\end{equation}
with $\deg P(z) \in \{ 1,2 \}$ neither in the case of the Kilroy fermion nor of the photon, except the trivial perturbative series 
obtained from evaluating Feynman diagrams and that adding log polynomials
\begin{equation}
 \wg(z) = \sum_{(\sigma,t,j) \geq (0,0,0)} c_{(\sigma,t,j)} z^{-\sigma c} e^{-\sigma P(z)} z^{-t} (\log z)^{j} 
\end{equation}
makes no difference: all nonperturbative sectors vanish. 

Since our approach has been purely algebraic and technical, the only answer available as to why this happened 
is that the DSEs' skeletons involve a multiplication by powers of the coupling constant.
For the Yukawa model, we have seen that increasing the degree $m=\deg P(z)$ of the instanton polynomial 
\begin{equation}
 P(z) = b_{1}z + \dotsc + b_{m} z^{m}
\end{equation}
is no expedient. Although it is not entirely clear, we do not believe that lowering the lower bound of the support of the nonperturbative part provides 
a resolution to the problem. 

Because all proofs were inductive on the sector index $\sigma$, splitting the nonperturbative sectors into several components, eg by using the 
transmonomials
\begin{equation}
 \mathfrak{t}_{1} = e^{-b_{1}z} , \dotsc , \mathfrak{t}_{m}=e^{-b_{m}z^{m}},
\end{equation}
can also offer no way out: we would simply find that all coefficients with indices $(\sigma_{1}, 0, \dotsc , 0, t)$ vanish and then the remainder with indices
$(\sigma_{1}, \sigma_{2}, \dotsc , \sigma_{m}, t)$ as well.

The fact that we have used a one-parameter series is immaterial for our results. The reason is that it does not matter what number one puts in front of the 
instanton polynomial $P(z)$; if such exponential factors are not meant to be participating transmonomials, their coefficients will inexorably vanish.   

However, although we cannot resolve this issue, we have allowed ourselves to indulge in some measure of speculation to argue that 
\emph{renormalisation may very likely act as a game changer}. From the heuristic viewpoint of the semi-classical expansion, it introduces a highly
nontrivial coupling dependence. In perturbation theory, this may manifest itself by increasing the exponential size of the Borel transform, even 
though the original perturbative series may still be of type Gevrey 1. This brings upon us the necessity to consider multisummability with 
Borel-Laplace transforms of higher indices which leads to higher level Stokes effects and hence $\deg P(z) \geq 2$ in our transseries ansatz. 

Since our negative result for the Yukawa model suggests that even this may not be sufficient, \emph{superexponentials} like the example (\ref{super})
have to be reckoned with. Interestingly, Stingl let his imagination fly\footnote{In his own words.} to speculate in \cite{Sti02} on p.60 about 
such (perhaps appealing) monstrosities. 

It is clearly too early and certainly premature to carry on creating uneducated guesswork: 
(multiplicative) renormalisation is a very subtle issue and things may be entirely different, especially 
for nonabelian gauge theories which have been and are still to this day being successfully tackled on 
the lattice.

Nonetheless, the author's humble personal opinion is that one should keep these possibilities in the back of 
one's mind when investigating nonperturbative aspects of renormalised quantum field theories. Besides, given their extreme smallness, 
superexponentials are very unlikely to be detected by lattice calculations any time soon: 
just evaluate the transmonomial $\mathfrak{n}_2$ in (\ref{super}) at $z=1/\alpha = 137$ and see what happens.

As mentioned briefly in §\ref{subsec:transQFT}, the results of \cite{APSaW15} (tentatively) suggest that 
bound states such as positronium manifest themselves in the form of transmonomials like 
\begin{equation}
\mathfrak{l}^{(n,m)} = \mathfrak{l}_1^{n}\mathfrak{l}_2^{m} = z^{-n} (\ln z)^{m}  
\end{equation}
which seem not directly related to the Stokes effect as they have no exponentials attached to them. 
They obviously make themselves felt at weak couplings, ie at large $z$, which makes perfect sense since 
such states lie in low-energy regimes. Unfortunately, we have not tested our equations
for such transmonomials but investigations in this direction are underway. 

Finally, we would like to mention that the Mellin transform method utilised in this work does not straightforwardly carry over to lower dimensions. 

\section*{Acknowledgements}

I thank Dirk Kreimer for supporting me and suggesting this interesting topic. I am grateful to David Broadhurst for a clarifying discussion concerning 
his work on the Kilroy ODE and to Robert Delbourgo, Stanley J. Brodsky and John Gracey for helping me trying to track down the original 
source of the self-consistent photon DSE (\ref{phoDSE1}). Furthermore, I owe Gerald Dunne special thanks. He has provided me with many 
references and interesting comments regarding the history of the divergence of perturbation theory in QFT. Thanks also go out to Ilmar Gahramanov and 
Daniel Krefl who made me aware of their contributions and, in particular, to Ricardo Schiappa for his valuable feedback to the first arXiv version.           

Last but not least, I would like to express my appreciation for the critical comments made by the anonymous 
referee.


\begin{thebibliography}{9}


\bibitem[APSaW15]{APSaW15} G.S. Adkins, C. Parsons, M.D. Salinger, R. Wang: \emph{Positronium energy levels at order $m\alpha^7$: vacuum polarisation
corrections in the two-photon-annihilation channel}. Phys. Lett. B 747 (2015) 551-555, arXiv: hep-th/1506.03835 

\bibitem[ARG94]{ARG94} A.H. Al-Ramadhan, D.W. Gidley: \emph{New Precision Measurement of the Decay Rate of Singlet Positronium}. Phys. Rev. Lett. 72, 
1632-1635 (1994)

\bibitem[AlS01]{AlS01} R. Alkofer, L. v. Smekal: \emph{On the infrared behaviour of QCD Green's functions}. Phys. Rep. 353, 281 (2001), 
arXiv: hep-ph/0007355

\bibitem[An15]{An15} I. Aniceto: \emph{The Resurgence of the Cusp Anomalous Dimension}. hep-th/1506.03388 

\bibitem[ARuS15]{ARuS15} I. Aniceto, J.G. Russo, R. Schiappa: \emph{Resurgent Analysis of Localizable Observables in Supersymmetric Gauge Theories}.
JHEP 1503 (2015) 172, arxiv: hep-th/1410.5834

\bibitem[AS14]{AS14} I. Aniceto, R. Schiappa: \emph{Nonperturbative Ambiguities and the Reality of Resurgent Transseries}. 
Comm. Math. Phys. 335, 183-245 (2015), arXiv: hep-th/1308.1115

\bibitem[ASVo12]{ASVo12} I. Aniceto, R. Schiappa, M. Vonk: \emph{The Resurgence of Instantons in String Theory}. Comm. Num. Theor. Phys. 6, 339-496 (2012),
arXiv: hep-th/1106.5922 

\bibitem[Ba00]{Ba00} W. Balser: \emph{Formal power series and linear system of meromorphic ordinary differential equations}. Springer (2000)

\bibitem[Ba09]{Ba09} W. Balser: \emph{From Divergent Power Series to Analytic Functions}. Springer (2009)

\bibitem[BaDU13]{BaDU13} G. Basar, G. Dunne, M. \"Unsal: \emph{Resurgence theory, ghost instantons, and analytic continuation of path integrals}. 
JHEP, 1310 (2013) 041, arXiv: hep-th/1308.1108

\bibitem[BaD15]{BaD15} G. Basar, G. Dunne: \emph{Resurgence and the Nekrasov-Shatashvili limit: connecting weak and strong coupling in the Mathieu
and Lam\'e systems}. JHEP 1502 (2015) 160, arXiv: hep-th/1501.05671

\bibitem[BaKUY09]{BaKUY09} G. van Baalen, D.Kreimer, D.Uminsky, K.Yeats: \emph{The QED $\beta$-function from global solutions to 
Dyson-Schwinger equations}, Ann. Phys. 324 (2009), 205-219, arXiv: hep-th/0805.0826

\bibitem[BaKUY10]{BaKUY10} G. van Baalen, D.Kreimer, D.Uminsky, K.Yeats: \emph{The QCD $\beta$-function from global solutions to 
Dyson-Schwinger equations}, Ann. Phys. 325 (2010), 300-324, arXiv: hep-th/0906.1754

\bibitem[BeC15]{BeC15} M.P. Bellon, P. Clavier: \emph{A Dyson-Schwinger Equation in the Borel Plane: Singularities of the Solution}. 
Lett. Math. Phys. 105, 795-825 (2015), arXiv: math-ph/1411.7190  

\bibitem[Bel10]{Bel10} M. Bellon: \emph{Approximate differential equations for renormalisation group function in models free of vertex divergences}. 
Nucl. Phys. B 826, 522-531 (2010), arXiv: hep-th/0907.2296

\bibitem[Ben99]{Ben99} M. Beneke: \emph{Renormalons}, Phys. Rep. 317 (1999)

\bibitem[BeO91]{BeO91} C.M. Bender, S.A. Orszag: \emph{Advanced Mathematical Methods for Scientists and Engineers I}. Springer (1991)

\bibitem[BeWu71]{BeWu71} C.M. Bender, T.T. Wu: \emph{Large-Order Behaviour of Perturbation Theory}. Phys. Rev. Lett. 27, 461-465 (1971)

\bibitem[BjoDre65]{BjoDre65} J.D. Bjorken, S.D. Drell: \emph{Relativistic Quantum Fields}. McGraw-Hill 
(1965)

\bibitem[BroK00]{BroK00} D.J. Broadhurst, D. Kreimer: \emph{Combinatoric explosion of renormalization tamed by Hopf algebra: thirty loop Pad\'e-Borel
resummation.} Phys. Lett. B 475 (2000), 63-70, arXiv: hep-th/9912093

\bibitem[BroK01]{BroK01} D.J. Broadhurst, D. Kreimer: \emph{Exact solution of Dyson-Schwinger equations for iterated one-loop integrals
and propagator coupling duality}, Nucl. Phys. B 600 (2001), 403 - 422, arXiv: hep-th/0012146 

\bibitem[CasLe79]{CasLe79} W.E. Caswell, G.P. Lepage: \emph{O($\alpha^2 \ln(1/\alpha))$ corrections in positronium: Hyperfine splitting and decay rate}. 
Phys. Rev. A 20, 36-43 (1979)

\bibitem[CESVo15]{CESVo15} R. Couso-Santamaria, J. D. Edelstein, R. Schiappa, M. Vonk: \emph{Resurgent transseries and the Holomorphic Anomaly:
Nonperturbative Closed Strings in Local $\C \mathbb{P}^{2}$}. Comm. Math. Phys. 338, 285-346 (2015), arXiv: hep-th/1407.4821  

\bibitem[CESVo16]{CESVo16} R. Couso-Santamaria, J. D. Edelstein, R. Schiappa, M. Vonk: \emph{Resurgent transseries and the Holomorphic Anomaly}. 
Ann. Henri Poincar\'e 17, 331-399 (2016), arXiv: hep-th/1308.1695

\bibitem[Co09]{Co09} O. Costin: \emph{Asymptotics and Borel summability}. CRC Press (2009)

\bibitem[CoDuJo77]{CoDuJo77} J.C. Collins, A. Duncan, S.D. Joglekar: \emph{Trace and dilatation anomalies in gauge theories}. Phys. Rev. D 16, 438 (1977)

\bibitem[CoSVa15]{CoSVa15} R. Couso-Santamaria, R. Schiappa, R. Vaz: \emph{Finite N from Resurgent large N}. Ann. Phys. 356, 1-28 (2015), 
arXiv: hep-th/1501.01007

\bibitem[CurP90]{CurP90} D.C. Curtis, M.R. Pennington: \emph{Truncating the Schwinger-Dyson equations: How multiplicative renormalizability and the 
Ward identity restrict the three-point vertex in QED}. Phys. Rev. D 42 (1990), 4165-4169 

\bibitem[DeKaTh97]{DeKaTh97} R. Delbourgo, A.C. Kalloniatis, G. Thompson: \emph{Dimensional renormalisation: ladders and rainbows}. Phys. Rev. D 55, 
5230-5233 (1997)

\bibitem[Do14]{Do14} D. Dorigoni: \emph{An Introduction to Resurgence, Trans-Series and Alien Calculus}. arXiv: hep-th/1411.3585

\bibitem[DoH15]{DoH15} D. Dorigoni, Y. Hatsuda: \emph{Resurgence of the Cusp Anomalous Dimension}. JHEP 1509, 138 (2015), arXiv: hep-th/1506.03763

\bibitem[Dun08]{Dun08} G. Dunne: \emph{New strong-field QED effects at ELI: Nonperturbative Vacuum Pair Production}. Eur. Phys. J. D 55 (2009) 327-340, 
arXiv: hep-th/0812.3163 

\bibitem[DunU12]{DunU12} G.V. Dunne, M. \"Unsal: \emph{Resurgence and trans-series in Quantum Field Theory: the $\C \mathbb{P}^{N-1}$ model}.
JHEP 11, 170 (2012), arXiv: hep-th/1210.2423

\bibitem[DunU13]{DunU13} G.V. Dunne, M. \"Unsal: \emph{Continuity and Resurgence: Towards a continuum definition of the $\C \mathbb{P}^{N-1}$ model}.
Phys. Rev. D 87, 025015 (2013) 

\bibitem[DunU14]{DunU14} G.V. Dunne, M. \"Unsal: \emph{Generating nonperturbative physics from perturbation theory}. Phys. Rev. D 89, 041701 (2014)
arXiv: hep-th/1306.4405v2 (2014)

\bibitem[Dys49]{Dys49} F.J. Dyson: \emph{The S Matrix in Quantum Electrodynamics}. Phys. Rev. 75, 1736 -1755 (1949) 

\bibitem[Dys51]{Dys51} F.J. Dyson: \emph{Divergence of Perturbation Theory in Quantum Electrodynamics}, Phys. Rev. 85 (1951), 631

\bibitem[Eca81]{Eca81} J. \'Ecalle: \emph{Les fonctions resurgentes}. Publ. Math. d'Orsay, Vol.I-III (1981)

\bibitem[Eca93]{Eca93} J. \'Ecalle: \emph{Six Lectures on Transseries, Analysable Functions and the Constructive Proof of Dulac's Conjecture}. In: 
\emph{Bifurcations and periodic orbits of vector fields}, D.Schlomiuk (editor), Kluwer (1993)

\bibitem[Eca04]{Eca04} J. \'Ecalle: \emph{Recent advances in the analysis of divergence and singularities}. In: 
\emph{Normal forms, Bifurcations and Finiteness Problems in Differential Equations.} Y. Ilyashenko, Ch. Rousseaou (editors), Kluwer (2004)

\bibitem[Ed09]{Ed09} G.A. Edgar: \emph{Transseries for beginners}. Real Anal. Exchange 35, 253-310 (2009), arXiv: math.RA/0801.4877v5 

\bibitem[EMaS75]{EMaS75} J.-P. Eckmann, J. Magnen, R. S\'en\'eor: \emph{Decay Properties and Borel Summability for the Schwinger Functions in 
$P(\Phi)_{2}$ Theories.} Comm. Math. Phys. 39, 251-271 (1975)

\bibitem[FeMaRiS85]{FeMaRiS85} J. Feldman, J. Magnen, V. Rivasseau, R. S\'en\'eor: \emph{Massive Gross-Neveu Model: A Rigorous Perturbative Construction}.
Phys. Rev. Lett. 54, 1479- 1481 (1985)

\bibitem[FeMaRiS86]{FeMaRiS86} J. Feldman, J. Magnen, V. Rivasseau, R. S\'en\'eor: \emph{A Renormalizable Field Theory: The Massive Gross-Neveu Model
in Two Dimensions}. Comm. Math. Phys. 103, 67-103 (1986)
 
\bibitem[GaT15]{GaT15} I. Gahramanov, K.Tezgin: \emph{A remark on the Dunne-\"Unsal relation in exact semi-classics}. (2015) arXiv: hep-th/1512.08466 

\bibitem[GliJa81]{GliJa81} J. Glimm, A. Jaffe: \emph{Quantum Physics. A Functional Integral Point of View}. Springer (1981)

\bibitem[GoKLaS91]{GoKLaS91} S.G. Gorishny, A.L. Kataev, S.A. Larin, L.R. Surguladze: \emph{The analytic four-loop corrections to the QED $\beta$-function in the MS scheme and
to the QED $\psi$-function. Total reevaluation.} Phys. Lett. B 256 (1991), 81  

\bibitem[GraMaZ15]{GraMaZ15} A. Grassi, M. Marino, S. Zakany: \emph{Resumming the perturbation series}. JHEP 1505 (2015) 038, arXiv: hep-th/1405.4214

\bibitem[GroN74]{GroN74} D.J. Gross, A. Neveu: \emph{Dynamical symmetry breaking in asymptotically free theories}. Phys. Rev. D 10, 3235- 3253 (1974)   

\bibitem[GuZi90]{GuZi90} J.C. Le Guillou, J. Zinn-Justin (eds.): \emph{Large-Order Behaviour of Perturbation Theory}. North Holland (1990)

\bibitem[Ha49]{Ha49} C.G. Hardy: \emph{Divergent series}. Oxford (1949)

\bibitem[vH06]{vH06} J. van der Hoeven: \emph{Transseries and real Differential Algebra}. Springer (2006)

\bibitem[vH07]{vH07} J. van der Hoeven: \emph{Efficient accelero-summation of holonomic functions}. Journal of Symbolic Computation 42, 389-428 (2007)

\bibitem[Hu52a]{Hu52a} C.A. Hurst: \emph{An example of a divergent perturbation expansion in field theory}. Math. Proc. Camb. Phil. Soc. 48, 625-639 (1952)

\bibitem[Hu52b]{Hu52b} C.A. Hurst: \emph{The enumeration of graphs in the Feynman-Dyson technique}. Proc. R. Soc. Lond. A 214, 44-61 (1952)

\bibitem[Hu06]{Hu06} C.A. Hurst: \emph{Perturbation expansions in quantum field theory}. Rep. Math. Phys. 57, 121-129 (2006) 

\bibitem[IPZu77]{IPZu77} C. Itzykson, G. Parisi, J.B. Zuber: \emph{Asymptotic Estimates in Scalar Electrodynamics}. Phys. Rev. Lett. 38, 306-310 (1977) 

\bibitem[Ja65]{Ja65} A. Jaffe: \emph{Divergence of perturbation theory for bosons}. Comm. Math. Phys. 1, 127-149 (1965)

\bibitem[JenZ04]{JenZ04} U-D. Jentschura, J. Zinn-Justin: \emph{Instantons in quantum mechanics and resurgent expansions}. 
Phys. Lett. B 596, 138-144 (2004)

\bibitem[KlaK13]{KlaK13} L. Klaczynski, D. Kreimer: \emph{Avoidance of a Landau pole by flat contributions in QED}. Ann. Phys. 344, 213-231 (2014),
arXiv: hep-th/1309.5061

\bibitem[Krei02]{Krei02} D. Kreimer: \emph{Combinatorics of (perturbative) quantum field theory}. Phys. Rept. 363, 387-424 (2002), arXiv: hep-th/0010059 

\bibitem[Krei06]{Krei06} D. Kreimer: \emph{Etude for linear Dyson-Schwinger Equations}, IHES/P/06/23 (see internet)

\bibitem[Kre14a]{Kre14a} D. Krefl: \emph{Non-perturbative Quantum Geometry}. JHEP 1402 (2014) 084, arXiv: hep-th/1311.0584  

\bibitem[Kre14b]{Kre14b} D. Krefl: \emph{Non-perturbative Quantum Geometry II}. JHEP 1412 (2014) 118, arXiv: hep-th/1410.7116  

\bibitem[KrY06]{KrY06} D. Kreimer, K. Yeats: \emph{An étude in non-linear Dyson-Schwinger equations}, 
Nucl. Phys. B Proc. Supp., 160, 116-121 (2006), arXiv: hep-th/0605096

\bibitem[Lau77]{Lau77} B. Lautrup: \emph{On high order estimates in QED}. Phys. Lett. 69B, 109-111 (1977)

\bibitem[Li77]{Li77} L.N. Lipatov: \emph{Divergence of the perturbation-theory series and the quasi-classical theory}. Sov. Phys. JETP 45, 216-223 (1977) 

\bibitem[Mar14]{Mar14} M. Marino: \emph{Lectures on nonperturbative effects in large $N$ gauge theories, matrix models and strings}. 
arXiv: hep-th/1206.6272v2

\bibitem[Mar15]{Mar15} M. Marino: \emph{Instantons and Large $N$.} Cambridge University Press (2015)

\bibitem[Ost86]{Ost86} K. Osterwalder: \emph{Constructive quantum field theory: goals, methods, results}. Helv. Phys. Acta 59, 220-228 (1986)

\bibitem[Pa77]{Pa77} G. Parisi: \emph{Asymptotic estimates in perturbation theory with fermions}. Phys. Lett. B 66, 382-384 (1977) 

\bibitem[Pe53a]{Pe53a} A. Petermann: \emph{Divergence of Perturbation Expansion}. Phys. Rev. 89, 1160 (1953) 

\bibitem[Pe53b]{Pe53b} A. Petermann: \emph{Renormalisation dans les s\'eries divergentes}. Helv. Phys. Acta 26, 291-299 (1953) 

\bibitem[PeSch95]{PeSch95} M.E. Peskin, D.V. Schroeder: \emph{An Introduction to Quantum Field Theory}.
Westview Press (1995)

\bibitem[Ri91]{Ri91} V. Rivasseau: \emph{From Perturbative to Constructive Renormalisation}. Princeton University Press (1991)

\bibitem[RoWi94]{RoWi94} C.D. Roberts, A.G. Williams: \emph{Dyson-Schwinger equations and their applications to hadron physics}. 
Prog. Part. Nucl. Phys. 33, 477-575 (1994), arXiv: hep-ph/9403224 

\bibitem[RuVeX09]{RuVeX09} R. Ruffini, G. Vereshchagin, S. Xue: \emph{Electron-positron pairs in physics and astrophysics: from heavy nuclei to 
black holes}. Phys. Rept. 487 (2010) 1-140, arXiv: astro-ph.HE/0910.0974  

\bibitem[Sa07]{Sa07} D. Sauzin: \emph{Resurgent functions and splitting problems}. arXiv: math.DS/0706.0137v1

\bibitem[Sa14]{Sa14} D. Sauzin: \emph{Introduction to 1-summability and resurgence}. arXiv: math.DS/1405.0356

\bibitem[Sal99]{Sal99} M. Salmhofer: \emph{Renormalization. An Introduction.} Springer (1999)

\bibitem[SchiVa14]{SchiVa14} R. Schiappa, R. Vaz: \emph{The Resurgence of Instantons: Multi-Cut Stokes Phases and the Painlev\'e II Equation}. 
Comm. Math. Phys. 330, 655-721 (2014), arXiv: hep-th/1302.5138  

\bibitem[Schwi51]{Schwi51} J. Schwinger: \emph{On the Green's functions of quantized fields. I \& II}.
Proc. Nat. Acad. Sc. 37, 452-459 (1951)

\bibitem[Shi12]{Shi12} M. Shifman: \emph{Advanced Topics in Quantum Field Theory: A Lecture Course.} Cambridge University Press (2012)

\bibitem[Si69]{Si69} B. Simon: \emph{Convergence of regularized, renormalized perturbation series for super-renormalisable field theories.}
Nuovo Cimento A 59, 199-214 (1969)

\bibitem[Sti02]{Sti02} M. Stingl: \emph{Field-Theory Amplitudes as Resurgent Functions}. arXiv: hep-ph/0207349

\bibitem[Th53]{Th53} W. Thirring: \emph{On the Divergence of Perturbation Theory for Quantized Fields}. Helv. Phys. Acta 26, 33-52 (1953)

\bibitem[VaZNoSh82]{VaZNoSh82} A.I. Vainstein, V.I. Zakharov, V.A. Novikov, M.A. Shifman: \emph{ABC of instantons}. Sov. Phys. Usp. 25 (1982), 195

\bibitem[Wa50]{Wa50} J.C. Ward: \emph{An identity in Quantum Electrodynamics}. Phys. Rev. 78, 182 (1950) 

\bibitem[We79]{We79} R. Weder: \emph{Instantons and renormalization}. Phys. Lett. 85 B, 249 (1979) 

\bibitem[Y11]{Y11} K. Yeats: \emph{Rearranging Dyson-Schwinger equations}, AMS, Vol. 211, 995, arXiv: 
'Growth estimates for Dyson-Schwinger equations', hep-th/08102249 (2011) (Different sign conventions!)

\bibitem[ZiJ04]{ZiJ04} J. Zinn-Justin, U-D. Jentschura: \emph{Multi-instantons and exact results I: conjectures, WKB expansions, and 
instanton interactions}. Ann. Phys. 313 (2004), 197-267 

\bibitem[Zi81]{Zi81} J. Zinn-Justin: \emph{Perturbation Series at Large Orders in Quantum Mechanics and Field Theories: Application to the 
Problem of Resummation}. Phys. Rep. 70, 109 (1981)

\bibitem[Zi02]{Zi02} J. Zinn-Justin: \emph{Quantum Field Theory and Critical Phenomena}. Oxford University Press (2002)


\end{thebibliography}
\end{document}